\documentclass[%
aps,pra,
reprint,
superscriptaddress,
twocolumn, 
nofootinbib,
floats,floatfix,
  amsmath,amssymb,
]{revtex4-2} 
    \usepackage[utf8]{inputenc}
    \usepackage[english]{babel}
\usepackage{amsthm}
\makeatletter
\newcommand{\sectionnotoc}[1]{%
  \begingroup
  \def\addcontentsline##1##2##3{}%
  \section{#1}%
  \endgroup
}
\makeatother
\newcommand{\starsectionnotoc}[1]{%
  \begingroup
  \def\addcontentsline##1##2##3{}%
  \section*{#1}%
  \endgroup
}
\makeatother
\makeatletter
\newcommand{\subsectionnotoc}[1]{%
  \begingroup
  \def\addcontentsline##1##2##3{}%
  \subsection{#1}%
  \endgroup
}
\makeatother
\newtheorem{theorem}{Theorem}
    \newtheorem{corollary}{Corollary}
    
    \newtheorem{proposition}{Proposition}
    \newtheorem{lemma}{Lemma}
    \usepackage{float}
    \usepackage{bm}
    
    \newtheorem{remark}{Remark}
    \usepackage{setspace}
    \usepackage{booktabs}
    \usepackage{multirow}
    \usepackage{amsmath,amssymb}
    \usepackage{qcircuit}
    \usepackage{caption}
    \usepackage{lipsum}
    \usepackage{tabularx}
    \usepackage{natbib}
    \usepackage{graphicx}
    \usepackage{rotating}
    \usepackage{subcaption}
    \usepackage{float}
    \usepackage{comment}
\usepackage{tikz}
\usepackage{mathtools}
\usepackage{pgfplots}
\usepackage{placeins}
\usepackage[ruled,vlined]{algorithm2e} 

\pgfplotsset{compat=newest} 
\usepackage{xcolor}
\usepackage[
    colorlinks=true,
    linkcolor=blue,
    citecolor=blue,
    urlcolor=blue
]{hyperref}

    \usepackage{makecell}
    \setcellgapes{2.3pt}
    
    \usepackage{subcaption} 
    \usepackage{physics}
    \usepackage[figurename=Fig., justification=RaggedRight]{caption}
    \usepackage{amsmath}
    \usepackage{xcolor}
    \usepackage{hyperref}

    \usepackage{adjustbox,expl3,etoolbox}
    \letcs\replicate{prg_replicate:nn}

\usepackage{cleveref}    

\begin{document} 
\title{Coherent-State Propagation: \\ A Computational Framework for Simulating Bosonic Quantum Systems}

\author{Nikita Guseynov}
\affiliation{Global College, Shanghai Jiao Tong University, Shanghai 200240, China}

\author{Zo\"{e} Holmes}
\affiliation{Institute of Physics, Ecole Polytechnique F\'{e}d\'{e}rale de Lausanne (EPFL),   Lausanne, Switzerland}
\affiliation{Centre for Quantum Science and Engineering, Ecole Polytechnique F\'{e}d\'{e}rale de Lausanne (EPFL),   Lausanne, Switzerland}

\author{Armando Angrisani}
\affiliation{Institute of Physics, Ecole Polytechnique F\'{e}d\'{e}rale de Lausanne (EPFL),   Lausanne, Switzerland}
\affiliation{Centre for Quantum Science and Engineering, Ecole Polytechnique F\'{e}d\'{e}rale de Lausanne (EPFL),   Lausanne, Switzerland}

\date{\today}

\begin{abstract}
We introduce \textit{coherent-state propagation}, a computational framework for simulating bosonic systems. We focus on bosonic circuits composed of displaced linear optics augmented by Kerr nonlinearities, a universal model of bosonic quantum computation that is also physically motivated by driven Bose--Hubbard dynamics. The method works in the Schr\"odinger picture representing the evolving state as a sparse superposition of coherent states. We develop approximation strategies that keep the simulation cost tractable in physically relevant regimes, notably when the number of Kerr gates is small or the Kerr nonlinearities are weak, and prove rigorous guarantees for both observable estimation and sampling. In particular, bosonic circuits with logarithmically many Kerr gates admit quasi-polynomial-time classical simulation at exponentially small error in trace distance. We further identify a weak-nonlinearity regime in which the runtime is polynomial for arbitrarily small constant precision.  We complement these results with numerical benchmarks on the Bose--Hubbard model with all-to-all connectivity. The method reproduces Fock-basis and matrix-product-state reference data, suggesting that it offers a useful route to the classical simulation of bosonic systems.
\end{abstract}

\maketitle

\sectionnotoc{Introduction} 
Simulating quantum systems is a fundamental task in modern physics. Indeed, it was one of the original motivations for quantum computation: Feynman proposed quantum computers as a way to simulate quantum dynamics that appear classically intractable~\cite{feynman1982simulating}. A key challenge is therefore to identify which quantum dynamics are accessible to classical algorithms, and what resources drive quantum systems beyond classical simulability leaving room for a quantum advantage.

One approach is to identify regimes in which quantum dynamics are \textit{exactly} classically simulable, and then quantify the resources required to move beyond them. This program is well developed for spin and fermionic systems. In qubit circuits, the Gottesman--Knill theorem shows that stabilizer computations are efficiently simulable~\cite{gottesman1998heisenberg, aaronson2004improved}; in fermionic systems, the analogous tractable class is that of fermionic Gaussian circuits~\cite{knill2001fermionic, valiant2001quantum, terhal2002classical}. These examples isolate the resources needed to leave the simulable regime: magic, or non-stabilizerness, for qubits, and fermionic non-Gaussianity for fermions. A related question is how simulation cost grows with limited amounts of these resources. In the standard Clifford+$T$ model, for instance, classical algorithms scale exponentially in the number of non-Clifford gates, so circuits with only logarithmically many $T$ gates remain polynomial-time simulable~\cite{bravyi2016improved}. In fermionic systems, analogous results hold: dynamics with only logarithmic non-Gaussianity remain polynomial-time simulable~\cite{dias2023classical, reardon2023improved}.

Propagation-based methods have been used to study the \textit{approximate} classical simulability of broader circuit families. In Pauli propagation, one tracks the Pauli terms of an observable through the circuit, applying truncation rules to maintain efficiency\ \cite{rall2019simulation, beguvsic2023simulating, beguvsic2024real, teng2025leveraging,  rudolph2025pauli}. Approximate simulation guarantees have been established for these methods in both noisy~\cite{aharonov2022polynomial, schuster2024polynomial, gonzalez2024pauli, angrisani2025simulating} and noiseless regimes~\cite{angrisani2024classically, lerch2024efficient}. These ideas have also inspired propagation-based algorithms beyond qubits, including Majorana propagation for fermionic systems~\cite{miller2025simulation, d2025majorana, rudolph2026thermal, facelli2026fast}.

Analogously to the spin and fermionic cases,  it is well understood that Gaussian bosonic processes—namely, Gaussian states evolving under Gaussian operations and probed with Gaussian measurements—are efficiently classically simulable~\cite{bartlett2002universal} and therefore bosonic non-Gaussianity is necessary for quantum computational advantage.
This prompts the question: \emph{How far beyond the Gaussian regime does efficient classical simulability persist?} Several classical approaches have been developed to try and address this question, including positive-Wigner-function methods~\cite{mari2012positive}, quasiprobability and related phase-space techniques~\cite{pashayan2015estimating, rahimi2016sufficient}, and superpositions of Gaussian or coherent states~\cite{chabaud2021classical, dias2024classical}. These results show that some useful non-Gaussian states remain efficiently simulable, but they are mostly tailored to scenarios in which the non-Gaussianity is restricted to the input state or the final measurement.

More recently, displacement propagation---a classical simulation framework based on propagating displacement operators---has been proposed for noisy bosonic circuits~\cite{upreti2025quantum}. The latter provides a continuous-variable analogue of Pauli
propagation, but its guarantees rely crucially on the presence of noise. Thus, despite
the success of propagation-based methods in qubit and fermionic settings, a comparable
framework for noiseless bosonic circuits remains missing.

\begin{figure*}[t]
    \centering
    \includegraphics[width=0.98\textwidth]{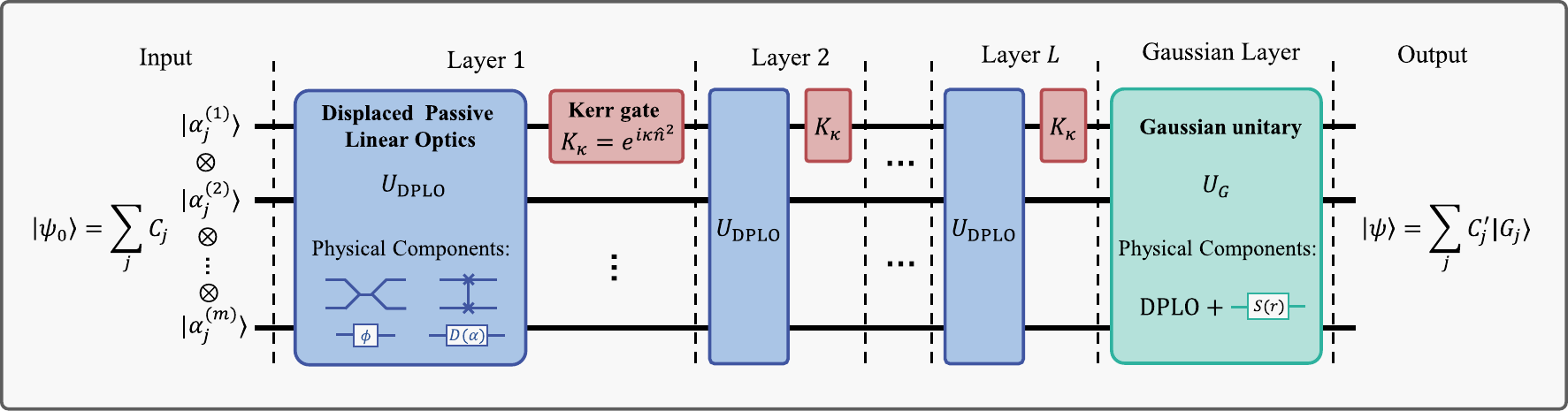}
    \caption{
\textbf{Schematic of universal Gaussian+Kerr circuit architecture.} The input state is represented as a finite superposition of multimode coherent-product states, Eq.~\eqref{eq:coherent_superposition_background}. Each layer consists of a displaced passive linear-optics (DPLO) unitary $U_{\mathrm{DPLO}}$ followed by a local Kerr gate $K_\kappa = e^{i\kappa \hat n^2}$, which induces branching in the coherent-state representation. After $L$ such layers, an arbitrary Gaussian unitary $U_G$ acts on the propagated state, yielding a superposition of Gaussian states. We note that for clarity the Kerr gate is shown on a single wire, but this is without loss of generality, since swaps within the DPLO layers can route it to any mode. This architecture is universal for continuous-variable computation (in the energy-cutoff sense~\cite{arzani2025can,Lloyd_PhysRevLett.82.1784}) and is physically motivated (e.g. can naturally capture the Bose-Hubbard model). 
}
    \label{fig:circuit_scheme}
\end{figure*}

In this work we consider noiseless bosonic circuits in which non-Gaussianity is generated repeatedly during the computation by interleaving displaced passive linear optics with local Kerr nonlinearities. This architecture (shown in Fig.~\ref{fig:circuit_scheme}) is both computationally expressive (it gives a universal gate set) and physically natural (it arises from driven interacting bosonic systems after standard product-formula decompositions). Kerr nonlinearities are especially useful for isolating the computational role of non-Gaussianity. Unlike cubic phase gates, which inject both non-Gaussianity and energy, Kerr gates generate non-Gaussianity while preserving photon number.

This distinction matters for classical simulation. Previous work has studied universal continuous-variable gate sets based on Gaussian unitaries augmented with cubic phase gates. In that setting, brute-force approaches such as direct propagation of polynomials in the quadratures can scale \emph{doubly} exponentially with the number of non-Gaussian gates~\cite{upreti2025interplay}, in sharp contrast with the merely exponential scaling familiar from stabilizer+T circuits. However, it has also been shown that cubic gates can drive doubly exponential energy growth\ \cite{chabaud2025energybosonscomputationalcomplexity}. This raises the question of whether the resulting hardness reflects
intrinsic computational power in physically relevant bosonic dynamics, or
is instead tied to gates that rapidly populate high-energy
sectors.  

Motivated by this state-of-the-art, we introduce \emph{coherent-state propagation}, a Schr\"odinger-picture framework for simulating non-Gaussian bosonic circuits directly in phase-space terms. The evolving state is represented as a finite superposition of multimode coherent-product states. Displaced passive linear-optical layers act by updating the coherent amplitudes and do not increase the number of components. The only branching step is the Kerr evolution, whose action we approximate by explicit coherent-state reconstructions and controlled truncation procedures. A schematic representation of the algorithm is shown in Fig.~\ref{fig:propogation_small_kerr_intro}.

Our contributions are threefold. First, we formulate coherent-state propagation for circuits generated by displaced passive linear optics and Kerr gates. Second, we give approximation schemes with rigorous guarantees, including a general Kerr-expansion method and a small-nonlinearity regime in which truncation keeps the representation tractable. Third, we show how to use the propagated representation for observable estimation and sampling. Finally, we illustrate our algorithm by simulating the $m=6$ Bose--Hubbard model with all-to-all connectivity. Taken together, these results provide a propagation-based simulation method tailored to bosonic systems and a new tool for studying when dynamically generated non-Gaussian bosonic dynamics remains classically simulable.

\begin{figure*}[t]
    \centering
    \includegraphics[width=\textwidth]{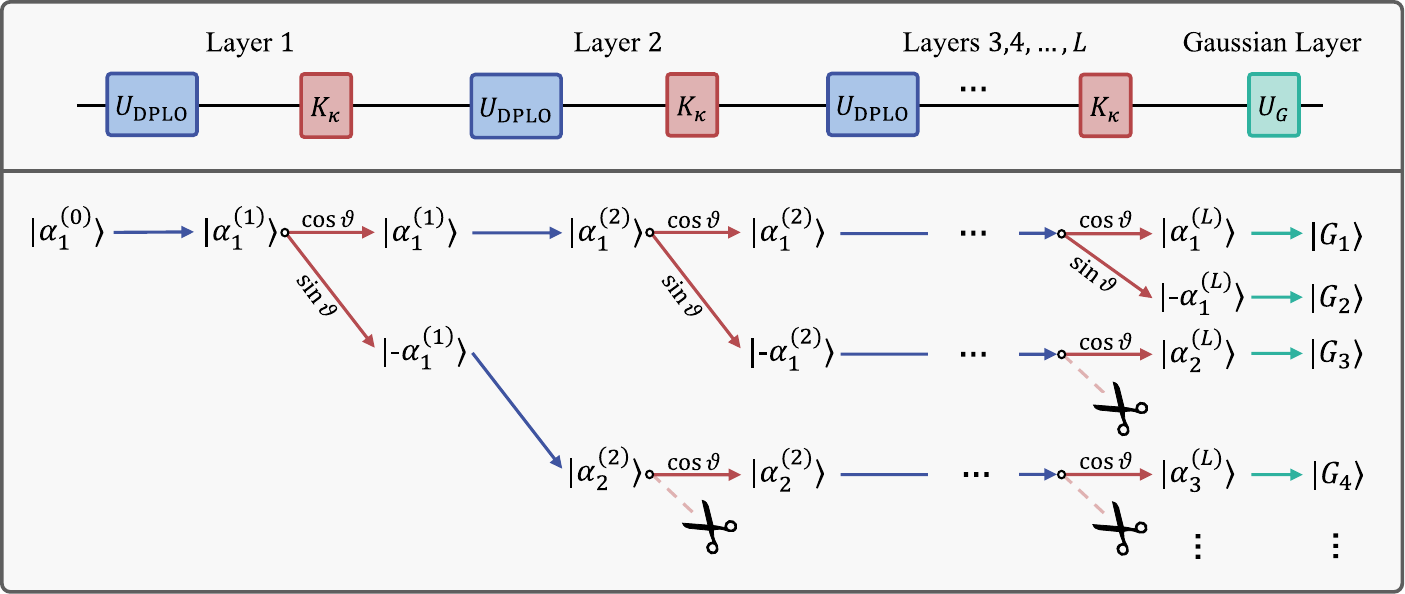}
    \caption{
\textbf{A visualization of coherent-state propagation in the regime of weak non-Gaussianity.} In this regime, the Kerr gate is approximated by the two-term update
\(
|\alpha\rangle \mapsto \cos\vartheta \ket{\alpha}
-i \sin\vartheta \ket{-\alpha}
\)
with \(|\sin \vartheta|\ll |\cos\vartheta|\); see Appendix~\ref{subsec: two term}. For visual clarity, only the coherent state in the first qumode, where the Kerr gate acts, is shown.
The branches carrying too many factors of \(\sin\) are truncated as less contributing. Unlike in the rest of the paper, the superscript in this figure labels the layer number, while the subscript distinguishes branches within the same layer, so \(\alpha_i^{(\ell)}\) and \(\alpha_j^{(\ell)}\) are unrelated for \(i\neq j\). After the final Gaussian layer, the surviving branches are mapped to Gaussian states, and the output is organized by powers of \(\sin\): the leading term \(|G_1\rangle\) corresponds to no non-Gaussian gates, with smaller corrections \(|G_2\rangle\), \(|G_3\rangle\), and so on.
}
    \label{fig:propogation_small_kerr_intro}
\end{figure*}

\bigskip\medskip
\sectionnotoc{Background}

Our focus is on bosonic computations generated by displaced passive linear evolution together with local Kerr nonlinearities. This choice is motivated from two directions. On the computational side, these ingredients already form a universal gate set (in the energy-cutoff sense; see Appendix~\ref{sec:cv-universality-kerr}). On the physical side, they arise naturally in interacting bosonic many-body systems. In particular, for a driven Bose--Hubbard model in a rotating frame, standard product-formula approximations lead to a passive linear term, a coherent drive, and onsite quartic terms, yielding precisely the layered dynamics considered here; see Appendix~\ref{sec:driven-bose-hubbard}.

Our simulation framework is based on coherent states. For a single bosonic mode with annihilation operator $\hat a$, the coherent state $|\alpha\rangle$ is defined by
\begin{equation}
    |\alpha\rangle = D(\alpha)|0\rangle,
\end{equation}
where $|0\rangle$ is the vacuum state and
\begin{equation}
    D(\alpha) = \exp(\alpha \hat a^\dagger - \alpha^* \hat a)
\end{equation}
is the displacement operator. Coherent states are eigenstates of the annihilation operator,
\begin{equation}
    \hat a|\alpha\rangle = \alpha |\alpha\rangle,
\end{equation}
and therefore provide a natural phase-space-adapted representation of bosonic dynamics.

For an $m$-mode system, we use tensor products of coherent states,
\begin{equation}
    |\boldsymbol{\alpha}\rangle :=
    |\alpha^{(1)}\rangle \otimes \cdots \otimes |\alpha^{(m)}\rangle,
\end{equation}
where $\boldsymbol{\alpha}=(\alpha^{(1)},\dots,\alpha^{(m)}) \in \mathbb{C}^m$. The starting point of our framework is to represent the evolving pure state as a finite superposition of such coherent-product states,
\begin{equation}
    |\psi\rangle = \sum_{k=1}^{N} C_k |\boldsymbol{\alpha}_k\rangle,
    \label{eq:coherent_superposition_background}
\end{equation}
where $N$ is the number of coherent-product components, $C_k \in \mathbb{C}$ are coefficients, and each $\boldsymbol{\alpha}_k \in \mathbb{C}^m$ specifies one $m$-mode coherent amplitude vector.

\medskip

The circuit architecture that we consider is shown schematically in Fig.~\ref{fig:circuit_scheme}. Each layer consists of two steps. First, we apply a displaced passive linear-optical (DPLO) transformation generated by a Hamiltonian of the form
\begin{equation}
    H_{\mathrm{DPLO}}
    =
    \sum_{j,k=1}^{m} \hat a_j^\dagger K_{jk} \hat a_k
    +
    \sum_{j=1}^{m}\bigl(\eta_j \hat a_j^\dagger + \eta_j^* \hat a_j\bigr),
    \label{eq:HDPLO_background}
\end{equation}
where $K=K^\dagger$ determines passive mode mixing and $\eta \in \mathbb{C}^m$ specifies coherent driving. 
Accordingly, the DPLO layer captures the standard optical building blocks of bosonic platforms—namely beam splitters, phase shifters, mode swaps, delay lines, and coherent displacements; see \cite{Scully_Zubairy_1997}.
The corresponding unitary preserves the coherent-product structure of Eq.~\eqref{eq:coherent_superposition_background}. In particular, each component is mapped to a single new coherent-product component, i.e.,
\begin{equation}\label{eq:DPLO_alpha_update_main}
    \boldsymbol{\alpha}_k \rightarrow \boldsymbol{\alpha'}_k \, .
\end{equation}
In this manner the DPLO gates update Eq.~\eqref{eq:coherent_superposition_background} term-by-term and the total number of terms does not increase.

Second, we apply local Kerr gates. For a single mode, the Kerr unitary is
\begin{equation}
    K_\kappa = \exp(i\kappa \hat n^2),
\end{equation}
where $\hat n = \hat a^\dagger \hat a$ is the number operator and $\kappa \in \mathbb{R}$ is the Kerr strength. In contrast with the DPLO evolution, the action of $K_\kappa$ on a coherent state is not closed within the coherent-state manifold. Controlling and approximating this step is the central technical problem addressed in this paper.

The final Gaussian unitary $U_G$ extends the DPLO toolbox by allowing active Gaussian operations, in particular single- and multimode squeezing.
The complete circuit architecture is thus of the form:
\begin{equation}
    U = U_G\prod_{\ell=1}^{L} U^{(\ell)} \coloneqq U_G\prod_{\ell=1}^{L}  U_{\mathrm{Kerr}}^{(\ell)} U_{\mathrm{DPLO}}^{(\ell)},
    \label{eq:citrcuit_architecture_intro}
\end{equation}
where each $U_{\mathrm{DPLO}}^{(\ell)}$ is generated by Eq.~\eqref{eq:HDPLO_background} and each $U_{\mathrm{Kerr}}^{(\ell)}$ is a single-mode Kerr gate (acting on the first mode without loss of generality).

\bigskip\medskip
\sectionnotoc{Coherent-state Propagation }
\label{sec: coherent propogation main}

We now present a workflow for the update rules underlying \emph{coherent-state propagation}. We assume that we start with an initial state that can be written (or well approximated) as the superposition of $N_0$ coherent-product states
\begin{equation}
    \ket{\psi_0}
    =
    \sum_{k=1}^{N_0}
    C_k^{(0)}\ket{\bm{\alpha}_k^{(0)}} \, .
    \label{eq:coherent_superposition_main}
\end{equation}
We then describe how a single circuit layer transforms both the coherent amplitudes and their coefficients, and then iterate this update through the full circuit. Throughout, we consider the circuit family discussed in the previous section.

For the encoding in Eq.~\eqref{eq:coherent_superposition_main}, the DPLO sublayer is structurally simple: it updates each coherent-amplitude vector (see Eq.\ \eqref{eq:DPLO_alpha_update_main}), while preserving the coherent-product form and, in particular, not increasing the number of branches $N$. The nontrivial ingredient is the Kerr gate, that does not admit such a termwise update and instead drives the growth of the coherent-state representation. In our framework, we therefore focus on approximating the action of the Kerr gate on coherent states. We distinguish two important regimes.
The first treats arbitrary Kerr strength $\kappa$ and makes explicit the resulting exponential growth of \(N\) with the number \(L\) of Kerr layers. This regime has similarities with the Clifford+T regime considered in spin systems. The second regime focuses on the weak-nonlinearity regime, that is, small \(\kappa\), with the aim of quantifying how much dynamically generated non-Gaussianity remains classically tractable. This regime is analogous to the Clifford and small $Z$ rotation regime in spin systems.

\bigskip
\subsectionnotoc{General Kerr regime}

\paragraph*{Framework.} We first consider arbitrary Kerr strength. The basic task is to characterize the action of a single-mode Kerr gate $K_\kappa$ on a coherent state $\ket{\alpha}$. The output state can be approximated as the sum $N_F$ coherent states on the phase orbit $\ket{\alpha e^{i\phi}}$. Concretely, in Appendix~\ref{sec:kerr_via_fock} we show that 
\begin{equation}
\begin{aligned}
K_\kappa\ket{\alpha}
&\approx
\sum_{r=0}^{N_F}
c_r\,
\ket{\alpha \exp\!\left(2\pi i \frac{r}{N_F+1}\right)} \, ,
\end{aligned}
\label{eq:general_kerr_update_main1}
\end{equation}
where the $N_F+1$ coefficients are computed from a discrete Fourier transform of the diagonal of the Kerr gate:
\begin{equation}
\begin{aligned}
\bm c
&:=
\mathcal C_\alpha\,
\mathrm{DFT}_{N_F+1}\!\left(e^{i\kappa \bm n^2}\right)
\end{aligned}
\label{eq:general_kerr_update_main2}
\end{equation}
where $\mathcal C_\alpha$ is $r$-independent constant and $\bm n=(0,1,\ldots,2M-1)^{\mathsf T}$. This Fourier transform can then be explicitly evaluated to give 
\begin{equation}
\begin{aligned}
\bm c
&=
\mathcal C_\alpha
\sum_{n=0}^{N_F}
\exp\!\left(i\kappa n^2-2\pi i \frac{n\bm r}{N_F+1}\right) \, .
\end{aligned}
\label{eq:general_kerr_update_main3}
\end{equation}
where $\bm r := (0,1,\ldots, N_F)$.

Let $\varepsilon_{\mathrm{kerr}}$ denote the error associated with approximating the output of applying a single Kerr gate to a coherent state as a sum of $N_F+1$ coherent states in Eq.~\eqref{eq:general_kerr_update_main1}.
Concretely, we define the single Kerr-gate approximation error as
\begin{equation}\label{eq: error singel kerr main}
\varepsilon_{\mathrm{kerr}}
:=
\left\|
K_\kappa\ket{\alpha}
-
\sum_{r=0}^{N_F}
c_r\,
\ket{\alpha \omega_r}
\right\|_2 
\end{equation}
where 
\begin{equation}\label{eq:orbit_exponents main}
\omega_r
:=
\exp\!\left(2\pi i \frac{r}{N_F+1}\right),
\qquad
r=0,1,\ldots,N_F .
\end{equation}
In Appendix~\ref{sec:coherent-state-calculus}, we show that $N_F+1$ states, with
\begin{equation}
\label{eq:NF_gap_scaling_main}
N_F
=
|\alpha|^2
+
\mathcal{O}\!\left(\log\frac{1}{\varepsilon_{\mathrm{kerr}}}\right),
\end{equation}
suffice to achieve a single Kerr gate approximation error of $\varepsilon_{\mathrm{kerr}}$.

\medskip

\paragraph*{Algorithm.} Combining the DPLO step with the Kerr update yields the following \emph{general-Kerr coherent-state} propagation algorithm for the universal bosonic circuit sketched in
Fig.~\ref{fig:circuit_scheme}.

\begin{enumerate}
\item[\textbf{I.}] \textbf{Initialize.} We initialize the state using the coherent-state
representation in Eq.~\eqref{eq:coherent_superposition_main}. For each
coherent-product branch \(k=1,\ldots,N_0\), the initial amplitudes are
collected into the vector
\begin{equation}
    \bm\alpha_k^{(0)}
    :=
    \bigl(
    \alpha_k^{(1,0)},\ldots,\alpha_k^{(m,0)}
    \bigr)^{\mathsf T}
    \in \mathbb C^m .
\end{equation}
Here the superscript labels the layer number. The corresponding complex
coefficients are stored in
\begin{equation}
    \bm C^{(0)}
    :=
    \bigl(
    C_1^{(0)},\ldots,C_{N_0}^{(0)}
    \bigr)^{\mathsf T}
    \in \mathbb C^{N_0}.
\end{equation}
Together, the data \(\{\bm\alpha_k^{(0)}\}_{k=1}^{N_0}\) and
\(\bm C^{(0)}\) specify the initial state.

    \item[\textbf{II.}] \textbf{Propagate each layer.} For each layer $\ell=1,2,\ldots,L$ and for each
coherent-product branch $k$, repeat the following.

\textit{(a)}
The linear-optical part updates the coherent-amplitude vector by the affine map
\begin{equation}
    \bm\alpha_k^{(\ell)}
    =
    e^{-iK^{(\ell)} t_\ell}\bm\alpha_k^{(\ell-1)}
    +
    \bm\gamma_\ell .
\end{equation}
Here \(K^{(\ell)}\) is the passive mode-mixing matrix for the \(\ell\)-th layer, see Eq.~\eqref{eq:HDPLO_background}, \(\bm\eta_\ell\) is the corresponding displacement vector, and \(t_\ell\) is the evolution time. When \(K^{(\ell)}\) is invertible, the displacement shift is
\begin{equation}
    \bm\gamma_\ell
    =
    \bigl(K^{(\ell)}\bigr)^{-1}
    \left(
    e^{-iK^{(\ell)} t_\ell}-I
    \right)\bm\eta_\ell .
\end{equation}
This affine update is derived in
Appendix~\ref{subsec:displaced-linear-optics}.

  \textit{(b)}
Choose the local reconstruction size \(N_F\) based on the coherent amplitude \(\alpha_k^{(1,\ell)}\) entering the Kerr gate and the desired single Kerr-gate approximation error \(\varepsilon_{\mathrm{kerr}}\), as defined in Eq.~\eqref{eq: error singel kerr main}. The required Fock cutoff is controlled by the Poisson tail of the coherent-state photon-number distribution in the Fock basis. Thus, it suffices to choose \(N_F\) according to Eq.~\eqref{eq:NF_gap_scaling_main}.
    
    \textit{(c)}
Expand the first-mode coherent state $|\alpha_k^{(1,\ell)}\rangle$ into a finite sum of coherent states on its phase orbit, while all other modes remain unchanged. Using the phase factors $\omega_r$ defined in Eq.~\eqref{eq:orbit_exponents main}, we obtain
\begin{equation}
\begin{aligned}
&\qquad\ket{\alpha_k^{(1,\ell)}}\otimes
\ket{\alpha_k^{(2,\ell)}}\otimes \cdots
\\
&\quad\qquad \mapsto
\sum_{r=0}^{N_F}
c_r^{(k,\ell)}
\ket{\alpha_k^{(1,\ell)}\omega_r}\otimes
\ket{\alpha_k^{(2,\ell)}}\otimes \cdots .
\end{aligned}
\end{equation}
Here $c_r^{(k,\ell)}$ are the coefficients from
Eq.~\eqref{eq:general_kerr_update_main3}.

\item[\textbf{III.}] \textbf{Apply the final Gaussian layer.} After the \(L\) Kerr--DPLO layers, apply the final
Gaussian unitary \(U_G\) termwise to the propagated coherent-state
superposition. In other words, each propagated coherent-product branch
$|\bm\alpha_k^{(L)}\rangle$ is mapped to a Gaussian state, which we denote by
\(\ket{G_k}\),
\begin{equation}
    U_G\sum_{k=1}^{N_L} C_k^{(L)}\ket{\bm\alpha_k^{(L)}}
    =
    \sum_{k=1}^{N_L} C_k^{(L)}\ket{G_k}.
\end{equation}
This representation allows us to apply the readout techniques developed in
Appendix~\ref{sec:readout} to compute observables or simulate quadrature sampling.

\end{enumerate}

\medskip\medskip
\subsectionnotoc{Weak Kerr regime}

A complementary regime arises when each
Kerr strength $\kappa$ is small, as in small Trotter steps or weak onsite nonlinearities.
In this case, the Kerr evolved state can be captured by a few
phase orbit coherent states. \\

\paragraph*{Framework.} In the weak-Kerr regime, it is convenient to use an alternative expansion tailored directly to small \(\kappa\) which we describe in detail in  Appendix~\ref{sec:small}. The starting point is the identity
\begin{equation}
\label{eq:n2_phase_orbit_identity_main}
\hat n^2\ket{\alpha}
=
-\left.
\frac{\partial^2}{\partial\phi^2}
\ket{\alpha e^{i\phi}}
\right|_{\phi=0}.
\end{equation}
In other words, the action of the operator \(\hat n^2\) on a coherent state can be represented by \(-\partial_\phi^2\), so the Kerr gate \(K_\kappa=e^{i\kappa \hat n^2}\) on a coherent state can be viewed as the propagator \(e^{-i\kappa\partial_\phi^2}\) on the phase variable \(\phi\).

This turns the action of the Kerr gate into a one-dimensional propagation problem on the phase circle. We then discretize that circle using \(2M\) equally spaced phase points, where \(M\in\mathbb N\) controls the size of the approximation, and replace \(\partial_\phi^2\) by the standard second-order finite-difference operator. The resulting discrete propagator is unitary, and its action on \(\ket{\alpha}\) produces a superposition of exactly \(2M\) coherent states at uniformly spaced phases,
\begin{equation}\label{eq:weak_kerr_finite_fourier_main}
K_\kappa \ket{\alpha}
\approx
\sum_{m=0}^{2M-1} c_m^{(M)}
\left| \alpha \exp\!\left(i\pi \frac{m}{M}\right)\right\rangle \, .
\end{equation}

To compute the coefficients, we diagonalize the discrete propagator in the \(2M\)-point discrete Fourier basis. On the Fourier mode indexed by \(n\), the discrete Laplacian has eigenvalue
\begin{equation}
\Lambda_n
=
\frac{4M^2}{\pi^2}\sin^2\!\left(\frac{\pi n}{2M}\right),
\qquad n=0,\dots,2M-1,
\end{equation}
so the propagator contributes the phase factor \(e^{i\kappa\Lambda_n}\). Writing \(c^{(M)}=(c_0^{(M)},\dots,c_{2M-1}^{(M)})\), the coefficient vector is therefore the inverse discrete Fourier transform of these phase factors:
\begin{equation}\label{eq:coefs as DFTweak kerr main}
c^{(M)}
=
\frac{1}{2M}\,
\mathrm{DFT}_{2M}
\!\left[
\exp\!\left(
i\kappa\,\frac{4M^2}{\pi^2}\sin^2\!\frac{\pi n}{2M}
\right)
\right].
\end{equation}

One can see from Eq.~\eqref{eq:weak_kerr_finite_fourier_main} that with this alternative Fourier decomposition strategy the number of coherent state terms in the sum still grows exponentially with depth,
\begin{equation}
N_L \le N_0(2M)^L .
\end{equation}
For this reason, in addition to the finite-Fourier approximation \eqref{eq:weak_kerr_finite_fourier_main}, we
apply a truncation scheme such that after each Kerr update 
only the $\mathcal{S}$ largest-coefficient branches are retained. 

While in general our algorithm can be run for arbitrary $M$ and $S$, for our theoretical guarantees, we specialize this truncation scheme to the extreme case where $M = 1$ and only two terms are kept in Eq.~\eqref{eq:weak_kerr_finite_fourier_main} such that we have
\begin{equation}
\begin{aligned}
K_\kappa\ket{\alpha}
&\approx
\cos\vartheta \ket{\alpha}
-i \sin \vartheta \ket{-\alpha},
\end{aligned}
\label{eq:weak_kerr_update_main}
\end{equation}
where $\vartheta=\theta_{\mathrm{opt}}(\alpha,\kappa)/2$ is chosen
to minimize the single gate Kerr error (as described in
Appendix~\ref{subsec: two term}). For compactness, we omitted the overall phase factor $\exp(i\theta_\mathrm{opt}/2)$.
As for small $\kappa$ we have that $\vartheta$ is small and so $| \cos\vartheta | \gg |\sin\vartheta|$. Thus a single Kerr update \eqref{eq:weak_kerr_update_main} is approximated by one dominant branch $\cos\vartheta$ and one weak branch $\sin\vartheta$. 
This justifies the following branch truncation strategy: once a branch accumulates more than $s$ factors of $\sin$, it is discarded and no longer propagated. We illustrate the case of $s=1$ in Fig.~\ref{fig:propogation_small_kerr_intro}.

With this structured truncation rule, only branches containing at most $s$ weak
choices are retained. Hence, after $L$ Kerr layers,
\begin{equation}
\label{eq:weak_kerr_branch_scaling_main}
N_L
\le
N_0 \sum_{q=0}^{s}\binom{L}{q}
=
O(N_0L^s),
\qquad
s\ \mathrm{fixed}.
\end{equation}
This gives the polynomial-depth scaling used in the weak-nonlinearity
guarantees. \\

\paragraph*{Algorithm.} Now we provide the \emph{weak-Kerr coherent-state} propagation algorithm for the bosonic circuit sketched in Fig.~\ref{fig:circuit_scheme}. It coincides with the general-Kerr algorithm in Steps \textbf{I} and \textbf{III}; only the propagation step \textbf{II} is modified, by replacing the general Kerr update with the weak Kerr finite-Fourier rule.

First, fix the finite-Fourier parameter \(M\) and the truncation parameter \(\mathcal S\).

\begin{itemize}
    \item \(M\in\mathbb N\) specifies the weak-Kerr finite-Fourier update. First define the phase-orbit points
    \begin{equation}
        \widetilde\omega_r
        :=
        \exp\!\left(i\pi \frac{r}{M}\right),
        \qquad
        r=0,\ldots,2M-1.
    \end{equation}
    Thus, each weak-Kerr update replaces one coherent state by a superposition of
    $2M$ coherent states on this orbit
    \begin{equation}
        K_\kappa\ket{\alpha}
        \approx
        \sum_{r=0}^{2M-1}
        c_r^{(M)}\ket{\alpha \widetilde\omega_r}.
        \label{eq:weak_kerr_finite_fourier_main_alg}
    \end{equation}

    \item \(\mathcal S\in\mathbb N\cup\{\infty\}\) is the truncation parameter:
    after each weak-Kerr update, we retain only the \(\mathcal S\) branches with
    the largest coefficient magnitudes.
\end{itemize}

We do not provide a specific prescription for choosing $M$ and $\mathcal S$. In the weak-Kerr algorithm, they are therefore treated as free approximation parameters. What we do provide is a way to monitor the error introduced by the finite-Fourier approximation and the subsequent truncation during the propagation; in this way, one can track these errors during the simulation and tune $M$ and $\mathcal S$ until the desired accuracy is reached; see Step \textbf{II(d)} below.

\begin{enumerate}

\item[\textbf{I.}] \textbf{Initialize.} This step is exactly the same as in the
general-Kerr algorithm. We store the initial coherent-amplitude vectors as
\begin{equation}
    \widetilde{\bm\alpha}_k^{(0)}
    :=
    \bigl(
    \widetilde{\alpha}_k^{(1,0)},\ldots,\widetilde{\alpha}_k^{(m,0)}
    \bigr)^{\mathsf T}
    \in \mathbb C^m .
\end{equation}
and the corresponding coefficients as
\begin{equation}
    \bm D^{(0)}
    :=
    \bigl(
    D_1^{(0)},\ldots,D_{N_0}^{(0)}
    \bigr)^{\mathsf T}
    \in \mathbb C^{N_0}.
\end{equation}

\item[\textbf{II.}] \textbf{Propagate each layer in the weak-Kerr regime.}
For each layer \(\ell=1,2,\ldots,L\), and for each retained
coherent-product branch $k$, repeat the following.

\textit{(a)}
Apply the DPLO sublayer exactly as in Step \textbf{II(a)} of the
general-Kerr algorithm:
\begin{equation}
    \widetilde{\bm\alpha}_k^{(\ell)}
    =
    e^{-iK^{(\ell)}t_\ell}\widetilde{\bm{\alpha}}_k^{(\ell-1)}
    +
    \bm\gamma_\ell .
\end{equation}

\textit{(b)}
Apply the weak-Kerr finite-Fourier update to the first mode, as in
Eq.~\eqref{eq:weak_kerr_finite_fourier_main_alg}:
\begin{equation}\label{eq:weak kerr branches update main}
\begin{aligned}
&\qquad\ket{\widetilde\alpha_k^{(1,\ell)}}\otimes
\ket{\widetilde\alpha_k^{(2,\ell)}}\otimes\cdots
\\
&\quad\qquad\mapsto
\sum_{r=0}^{2M-1}
\tilde{c}_r^{(k,\ell)}
\ket{\widetilde\alpha_k^{(1,\ell)}\widetilde\omega_r}\otimes
\ket{\widetilde\alpha_k^{(2,\ell)}}\otimes\cdots.
\end{aligned}
\end{equation}
Here the coefficients $c_r^{(k,\ell)}$ are the finite-Fourier coefficients
defined in Eq.~\eqref{eq:coefs as DFTweak kerr main}. 

\textit{(c)}
Truncate the expanded branches. Form the new coefficient list
$\bm D^{(\ell)}$ by substituting each incoming coefficient
$D_k^{(\ell-1)}$ with the block
\begin{equation}
    \left(
    D_k^{(\ell-1)}c_0^{(M,\ell)},\ldots,
    D_k^{(\ell-1)}c_{2M-1}^{(M,\ell)}
    \right).
\end{equation}
The corresponding list of coherent-amplitude vectors is updated
accordingly, see Eq.~\eqref{eq:weak kerr branches update main}. After
performing these substitutions for all retained incoming branches, sort
the resulting branches by decreasing coefficient magnitude,
$|D_1^{(\ell)}|\ge |D_2^{(\ell)}|\ge \cdots$, and keep only the first
$\mathcal S$ entries, discarding the rest.

\textit{(d)}
Monitor the approximation error introduced by the algorithm (\emph{optionally}).
Propagate the same input $|\widetilde{\bm \alpha}_k\rangle$ in parallel with the general-Kerr update, and denote
the two states after layer $\ell$ by
\begin{equation}
    \ket{\psi_{\mathrm{wk}}^{(\ell)}}
    =
    \sum_{j=1}^{\mathcal S}
    D_j^{(\ell)}
    \ket{\widetilde{\bm\alpha}_j^{(\ell)}},
\end{equation}
and
\begin{equation}
    \ket{\psi_{\mathrm{gen}}^{(\ell)}}
    =
    \sum_{i=1}^{\mathcal S N_F}
    C_i^{(\ell)}
    \ket{\bm\alpha_i^{(\ell)}}.
\end{equation}
Their overlap $ \braket{\psi_{\mathrm{gen}}^{(\ell)}}{\psi_{\mathrm{wk}}^{(\ell)}}$ can be computed analytically. This allows us to numerically upper bound (albeit loosely) the error induced at each step, in a manner analogous to that done in MPS~\cite{A_J_Daley_2004,SCHOLLWOCK201196} and Pauli Propagation~\cite{rudolph2025pauli, fuller2025improved}. For more details see Appendix~\ref{subsec:finite-fourier-cutoff}. 

\item[\textbf{III.}] \textbf{Apply the final Gaussian layer.} This step is exactly as in the general-Kerr algorithm
\begin{equation}
    U_G\sum_{k=1}^{\mathcal S} D_k^{(L)}\ket{\widetilde{\bm\alpha}_k^{(L)}}
    =
    \sum_{k=1}^{\mathcal S} D_k^{(L)}\ket{G_k}.
\end{equation}

\end{enumerate}

\sectionnotoc{Theoretical guarantees}

We are now ready to state the theoretical guarantees for coherent-state
propagation. The simulation cost is governed by the circuit size, through the number of
modes $m$ and the number of layers $L$, together with an effective energy scale
associated with the coherent-state representation used during the simulation.
Accordingly, we define
\begin{align}
\label{eq:lambda-intro}
    \lambda
    \coloneqq
    \max_{\ell}\max_{k\in[N_\ell]}
    \bigl|\alpha_{k,\ell}^{(1)}\bigr|^2,
\end{align}
where $\ket{\alpha_{k,\ell}^{(1)}}$ denotes the first-mode component appearing in the propagated representation after $\ell$ layers, and $N_\ell$ is the number of such components at that stage. Thus, $\lambda$ is the largest local mean photon number on the first mode among all coherent states encountered during the simulation. We stress that, if the initial state satisfies
\begin{equation}
\zeta_0 \coloneqq \max_{k\in[N_0]} \|\boldsymbol{\alpha}_{k,0}\|_2^2 =  \max_{k\in[N_0]} \sum_{i=1}^m\bigl|\alpha_{k,0}^{(i)}\bigr|^2,
\end{equation}
and the displacement layers inject only a constant amount of local energy per layer, then
\begin{equation}\label{eq: lambda linear growth  main}
\lambda = \zeta_0 + \mathcal{O}(L).
\end{equation}
Therefore, under natural assumptions, $\lambda$ grows at most linearly with the circuit depth.

Our first result gives a general runtime guarantee for coherent-state
propagation. It shows that the cost grows slightly faster than exponentially in the number of
Kerr layers, and therefore remains efficient whenever that number is
sufficiently small.

\begin{theorem}[Circuits with few Kerr layers]
\label{thm:few-kerr-main}
Let $\ket{\psi_0}$ be a superposition of $N_0 = \mathrm{poly}(m)$ coherent states and let $\ket{\psi} = U \ket{\psi_0}$, where $U$ is the circuit defined in Eq.\ \eqref{eq:citrcuit_architecture_intro}.
For any $\varepsilon>0$, coherent-state propagation outputs a state
$\widetilde{\ket{\psi}}$ such that
\begin{equation}
    \bigl\|\ket{\psi}-\widetilde{\ket{\psi}}\bigr\|_2
    \le \varepsilon
\end{equation}
in time
\begin{equation}
    \mathcal{O}\!\left(
    Lm^3 + m^2\bigl(\lambda+\log(L/\varepsilon)\bigr)^L
    \right).
\end{equation}
Consequently, for circuits with $L=\mathcal{O}(\log m)$ layers:
\begin{enumerate}
    \item if $\lambda=\operatorname{poly}(m)$ and
    $\varepsilon=\exp\!\bigl(-\mathcal{O}(m)\bigr)$, then the simulation runs
    in quasi-polynomial time,
    \begin{equation}
        m^{\mathcal{O}(\log m)} ;
    \end{equation}
    \item if $\lambda=\mathcal{O}(\log m)$ and
    $\varepsilon=1/\operatorname{poly}(m)$, then the simulation runs in
    almost-polynomial time,
    \begin{equation}
        m^{\mathcal{O}(\log\log m)} .
    \end{equation}
\end{enumerate}
\end{theorem}

A proof is provided in Corollary~\ref{cor:log-depth-special-regimes} in
Sec.~\ref{subsec:forward-propagation-layered}. \\

The previous theorem is completely general, but its cost can still grow rapidly
with the depth. Our second result identifies a complementary weak-nonlinearity
regime in which truncation controls the branching induced by Kerr evolution
and restores a polynomial dependence on $L$.

\begin{theorem}[Circuits with small nonlinearity]
\label{thm:small-kappa-main}
Let $\ket{\psi_0}$ be a superposition of $N_0 = \mathrm{poly}(m)$ coherent states and let $\ket{\psi} = U \ket{\psi_0}$, where $U$ is the circuit defined in Eq.\ \eqref{eq:citrcuit_architecture_intro}.
Assume the strength of the Kerr gates is upper bounded by
\begin{equation}
\label{eq:small-kappa-main-assumption}
\kappa \le \frac{\delta}{L(1+\lambda^2)},
\end{equation}
for a sufficiently small constant $\delta$.

Then, for every $\varepsilon\in(0,1)$, coherent-state propagation with
a small-angle truncation scheme outputs a state $\widetilde{\ket{\psi}}$ such that
\begin{equation}
\label{eq:small-kappa-main-error}
\bigl\|\ket{\psi}-\widetilde{\ket{\psi}} \bigr\|_2
\le
\mathcal{O}(\delta)+\varepsilon,
\end{equation}
with runtime
\begin{equation}
\label{eq:small-kappa-main-runtime}
t
=
\mathcal{O}\!\left(
Lm^3 + m^2 L^{\,1+\mathcal{O}(\log(1/\varepsilon))}
\right).
\end{equation}
In particular, for constant $\varepsilon$, the runtime is polynomial in $m$
and $L$.
\end{theorem}

A proof is provided in Corollary~\ref{cor:simple-small-kappa} in
Sec.~\ref{subsec: tree structure}.

\medskip

Taken together, Theorems~\ref{thm:few-kerr-main} and~\ref{thm:small-kappa-main}
clarify how dynamically generated bosonic non-Gaussianity affects classical
simulation cost. 

Theorem~\ref{thm:few-kerr-main} shows that, at fixed effective
energy scale \(\lambda\), the overhead from Kerr layers is
\begin{equation}
(\lambda+\log(L/\varepsilon))^L,
\end{equation}
that is, exponential in the number of Kerr layers up to a logarithmic
accuracy-dependent correction. This is a qualitative improvement over previous
bounds for Gaussian circuits augmented with cubic gates, where propagation-based
approaches can scale \emph{doubly} exponentially with the number of non-Gaussian
layers\ \cite{upreti2025interplay}. In this sense, Kerr nonlinearities behave much more like the familiar
non-Clifford resources in qubit simulation, such as \(T\) gates, for which the
cost is also exponential in the number of non-stabilizer insertions.
The extra \(\log(L/\varepsilon)\) factor arises from error accumulation across
the \(L\) Kerr layers: achieving total error \(\varepsilon\) requires local
accuracy of order \(\varepsilon/L\).
It would be interesting to determine whether a sharper analysis can remove this
overhead and yield a strictly exponential scaling in the number of Kerr layers.

Theorem~\ref{thm:small-kappa-main} identifies a complementary weak-nonlinearity
regime in which this branching becomes much less costly. When
\(\kappa \lesssim 1/[L(1+\lambda^2)]\), an alternative truncation scheme keeps the runtime
polynomial in the depth for constant precision. This is directly analogous in
spirit to results for qubit circuits composed of Clifford unitaries plus small-angle Pauli rotations\ \cite{lerch2024efficient}. It is also the regime naturally
realized in short-time Trotterized evolutions, such as driven Bose--Hubbard
dynamics, where each Kerr step is perturbatively small.

\sectionnotoc{Sampling and Observable Estimation}\label{sec:Sampling and Observable Estimation}

\begin{figure*}[t]
    \centering
    \begin{subfigure}[b]{0.32\textwidth}
        \centering
        \includegraphics[width=\textwidth]{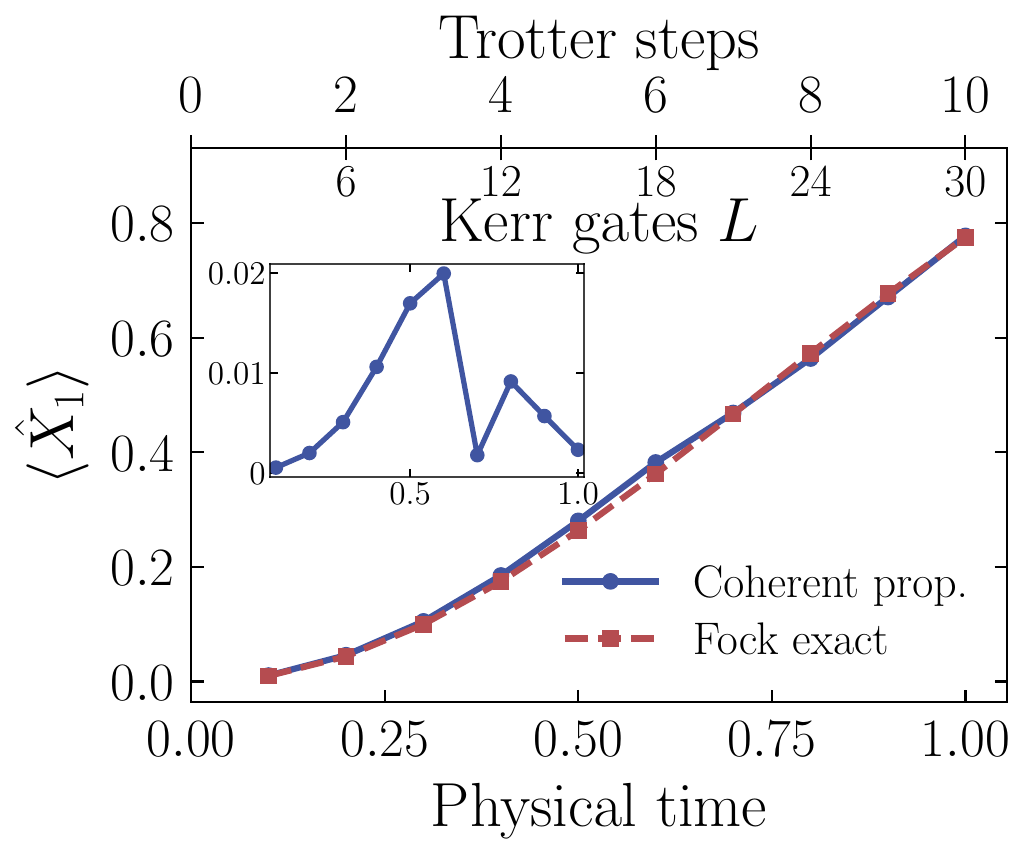}
        \label{fig:numerics_x1}
    \end{subfigure}
    \hfill
    \begin{subfigure}[b]{0.32\textwidth}
        \centering
        \includegraphics[width=\textwidth]{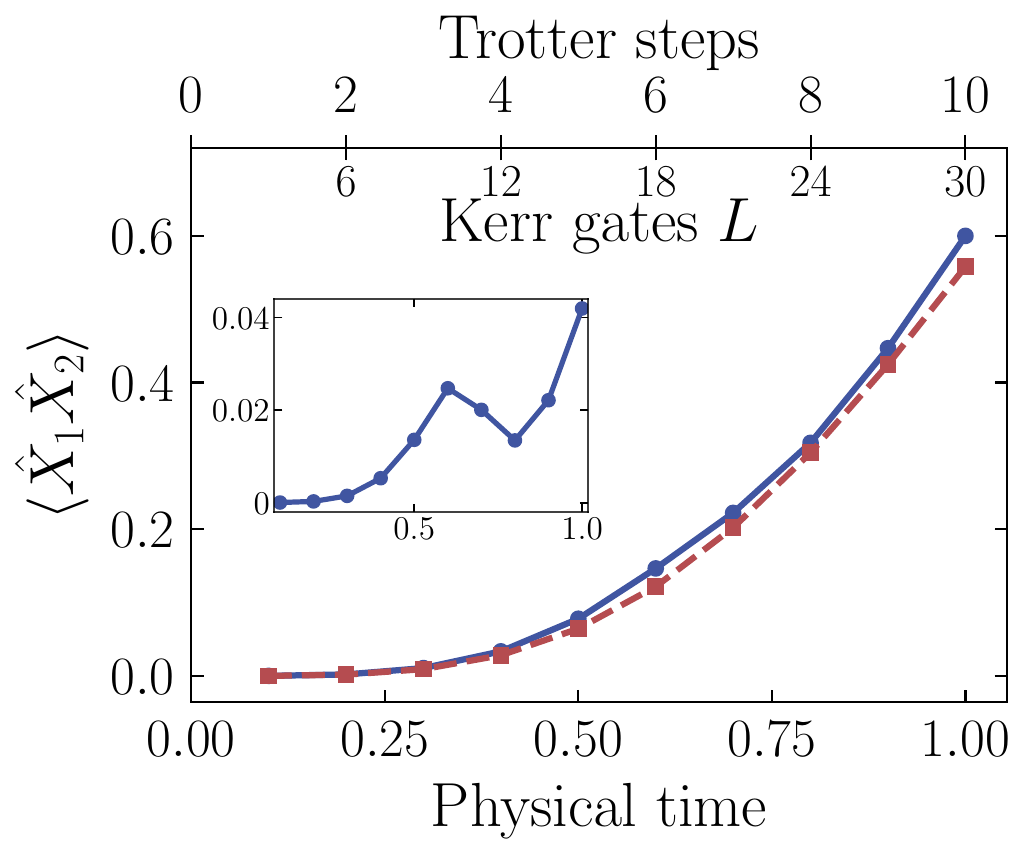}
        \label{fig:numerics_x1x2}
    \end{subfigure}
    \hfill
    \begin{subfigure}[b]{0.32\textwidth}
        \centering
        \includegraphics[width=\textwidth]{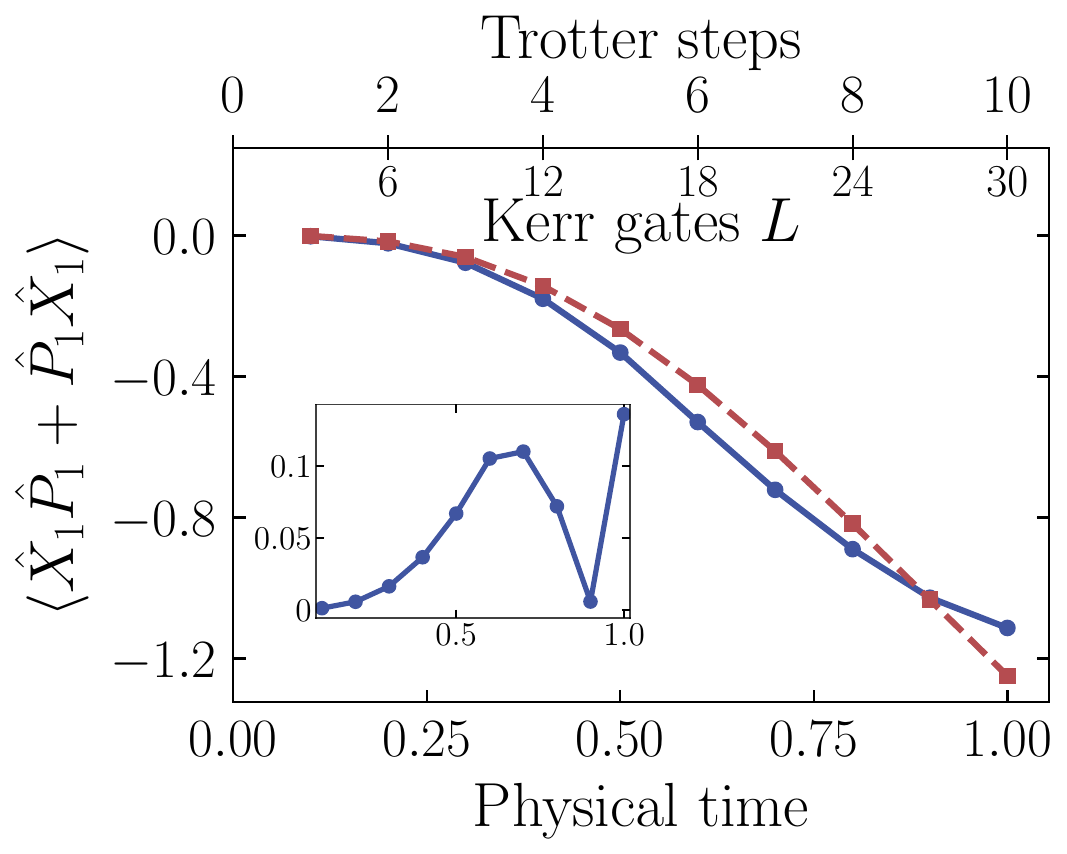}
        \label{fig:numerics_x1p1}
    \end{subfigure}
    \caption{\textbf{Benchmarking against exact Fock-basis simulation.}
Comparison between coherent-state propagation and the exact Fock-basis simulation for representative observables in the three-mode driven Bose--Hubbard model with all-to-all connectivity; see Appendix~\ref{sec:driven-bose-hubbard}. Panels (a)--(c) show the time evolution of $\langle \hat{X}_1 \rangle$, $\langle \hat{X}_1 \hat{X}_2 \rangle$, and $\langle \hat{X}_1 \hat{P}_1 + \hat{P}_1 \hat{X}_1 \rangle$, respectively. The inset in each panel shows the absolute error. The simulation-time comparison is given in Table~\ref{tab:runtime_fock}. Further numerical details are given in Appendix~\ref{subsec:numerics_supplementary}.
}
    \label{fig:numerics_three_observables}
\end{figure*}

\begin{figure*}[t!]
    \centering
    \subcaptionbox{}[0.32\textwidth]{%
        \includegraphics[width=\linewidth]{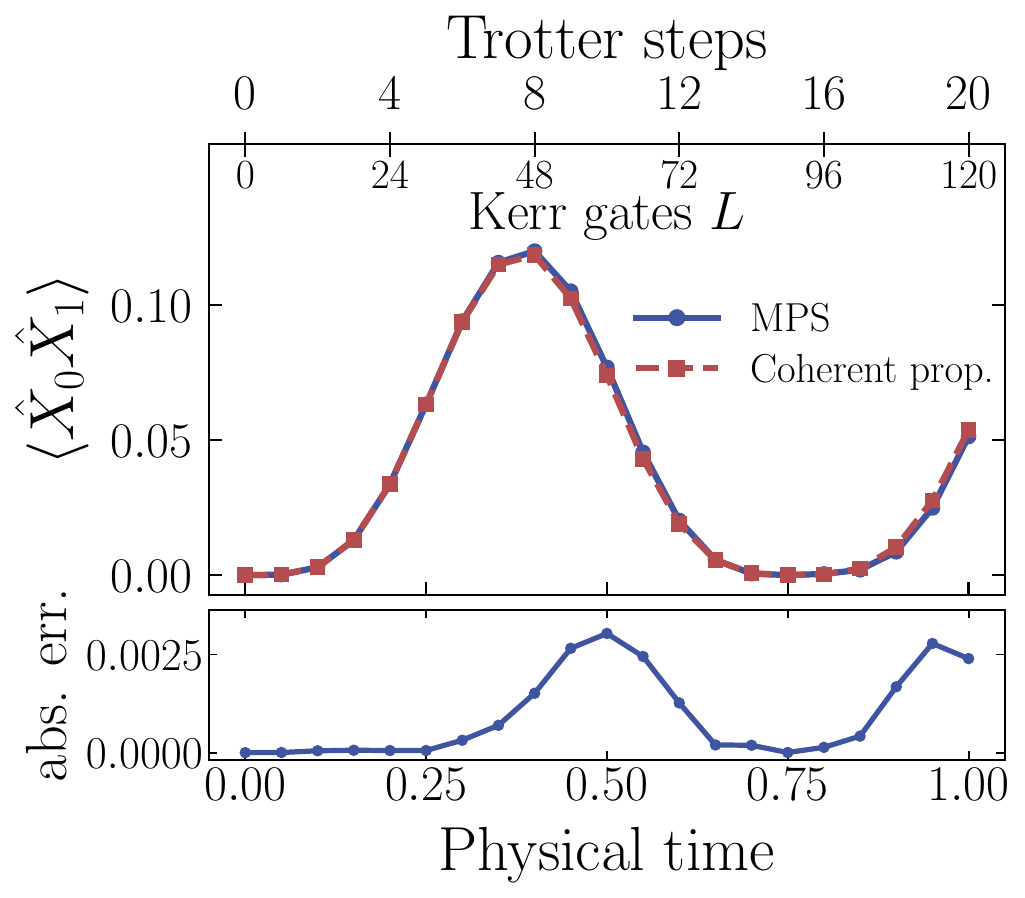}%
        \label{fig:mps_x0x1_traj}
    }
    \hfill
    \subcaptionbox{}[0.32\textwidth]{%
        \includegraphics[width=\linewidth]{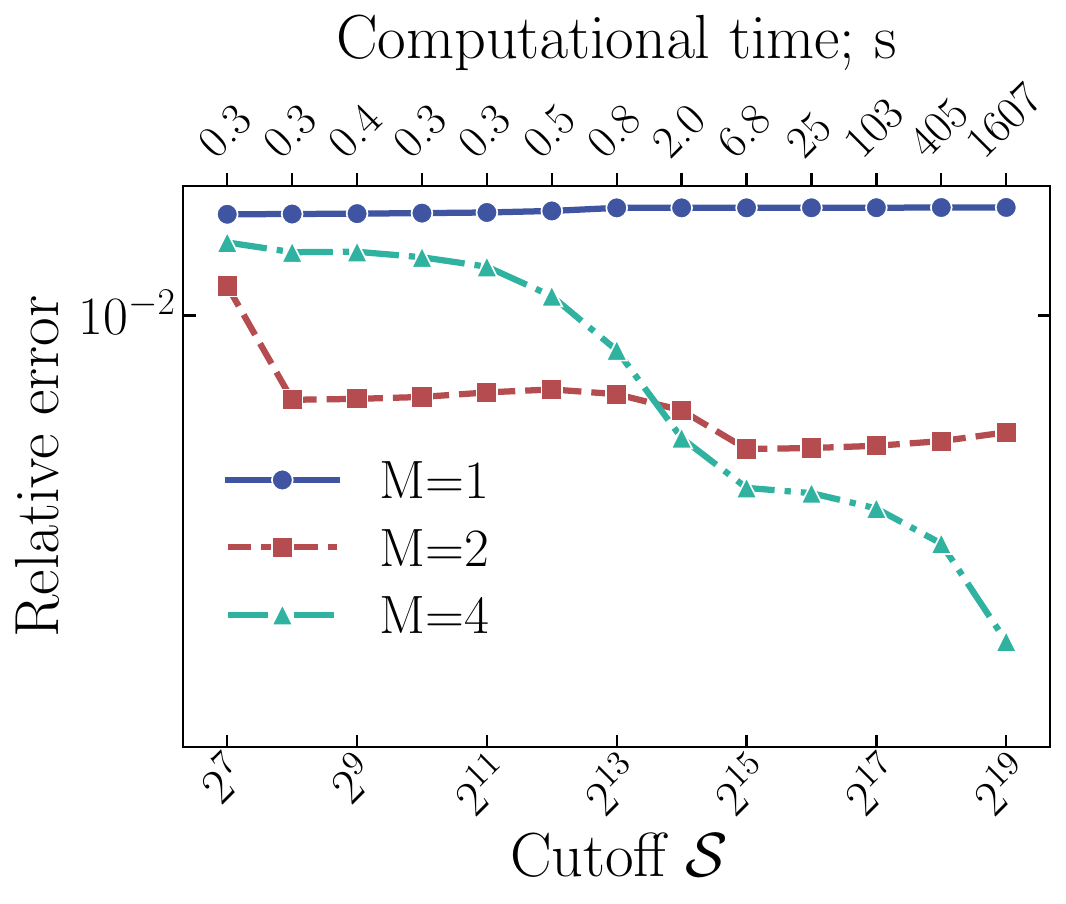}%
        \label{fig:mps_x2_cutoff}
    }
    \hfill
    \subcaptionbox{}[0.32\textwidth]{%
        \includegraphics[width=\linewidth]{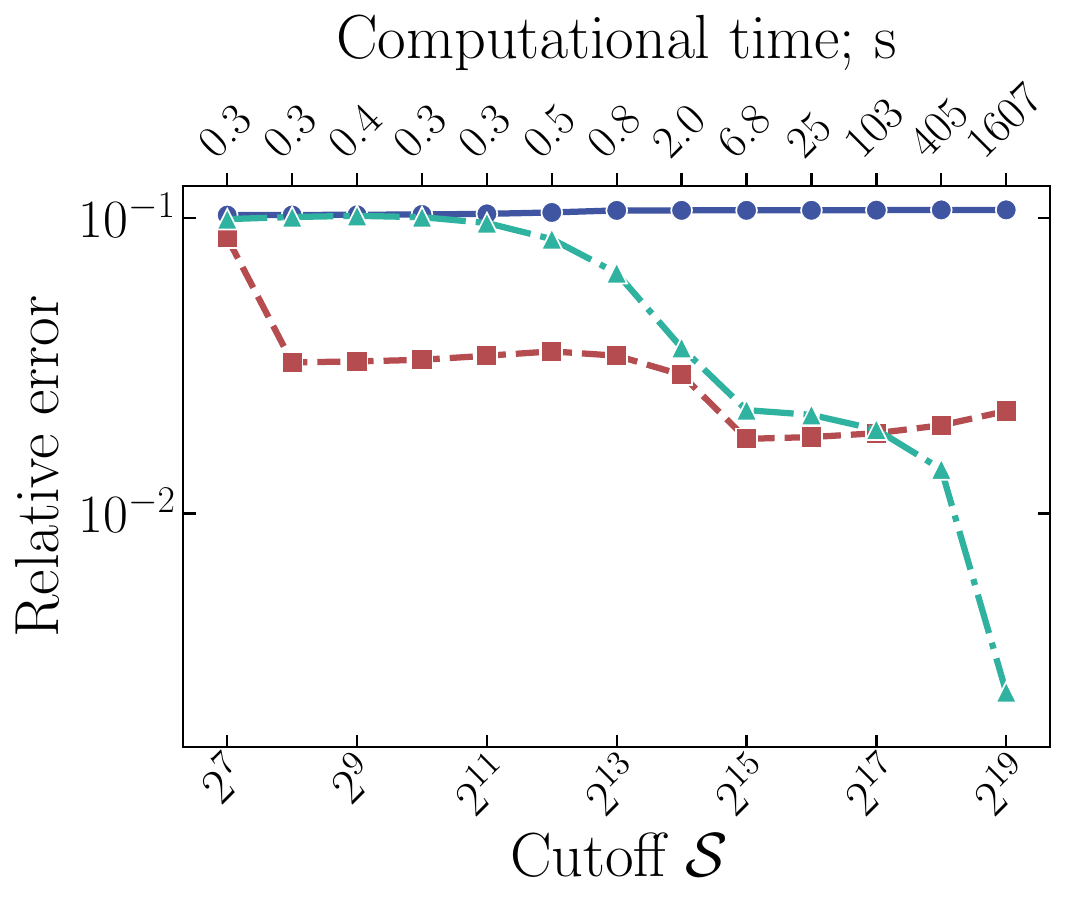}%
        \label{fig:mps_x0x1_cutoff}
    }
    \caption{\textbf{Benchmarking against MPS simulation.}
    Here we consider an \(m=6\) Bose-Hubbard system with an all-to-all topology. 
    (a) Time evolution of \(\langle \hat{X}_0 \hat{X}_1\rangle\) and the corresponding absolute error for coefficient truncation parameter \(\mathcal{S}=2^{12}\) and \(M=3\). 
    (b) Relative error of \(\langle \hat{X}_0^2\rangle\) at \(T=1\) as a function of the cutoff \(\mathcal{S}\). 
    (c) Same as (b) for \(\langle \hat{X}_0 \hat{X}_1\rangle\). 
    Here \(M\) controls the number of terms in the expansion \(K_{\kappa}\ket{\alpha}\), see Appendix~\ref{subsec:finite-fourier-cutoff}. 
    Further details of the simulation are given in Appendix~\ref{subsec:numerics_supplementary}.
    }
    \label{fig:mps_benchmark}
\end{figure*}

After the final Gaussian layer (see Fig.~\ref{fig:circuit_scheme}), coherent-state propagation outputs a finite
superposition of Gaussian states,
\begin{equation}
\label{eq:readout_gaussian_superposition_main}
\ket{\psi_{\mathrm{out}}}
=
\sum_{j=1}^{N} C_j \ket{G_j}.
\end{equation}
This representation is well suited to standard continuous-variable
readout. In experimental settings, the most common outputs are quadrature-based
measurements, including homodyne, rotated-quadrature, and Bell-type schemes
\cite{Scully_Zubairy_1997}. Our framework naturally supports all of these
sampling-based measurements by reducing them, after an appropriate final
Gaussian basis change $U_G$, to sampling in the $X$-quadrature basis.

\begin{theorem}[Quadrature sampling]
\label{thm:readout_sampling_main}
Given a classical description of the state
\eqref{eq:readout_gaussian_superposition_main}, describing an $m$-mode
bosonic system as a superposition of $N$ Gaussian components, one can sample
$N_{\mathrm{shot}}\in\mathbb N$ outcomes from any quadrature distribution
obtained by a Gaussian basis change $U_G$, including rotated quadratures
$\hat X_\theta=\hat X\cos(\theta/2)+\hat P\sin(\theta/2)$, in time
\begin{equation}
\mathcal{O}\!\left(N^2 m^2 \log(1/\varepsilon)+N_{\mathrm{shot}}Nm^2\right)
\end{equation}
to mixing accuracy $\varepsilon$. If the Gaussian components are product
coherent states \eqref{eq:coherent_superposition_background}, the cost improves to
\begin{equation}
\mathcal{O}\!\left(N^2 m \log(1/\varepsilon)+N_{\mathrm{shot}}Nm\right).
\end{equation}
\end{theorem}
    A proof is provided in Appendix~\ref{subsec:readout-sampling}.\\

For applications and theoretical diagnostics, one is often interested in expectation values of moments, correlations, and more general observables. Our framework supports computing expectation values of Hermitian observables of the form
\begin{equation}
\label{eq:readout_poly_observable_main}
\begin{aligned}
\hat O
&=
\sum_{k=1}^{\eta}
\left(
o_k\,
\hat{\bm X}^{\bm q^{(k)}}\hat{\bm P}^{\bm s^{(k)}}
+
o_k^*\,
\hat{\bm P}^{\bm s^{(k)}}\hat{\bm X}^{\bm q^{(k)}}
\right),\\
\hat{\bm X}^{\bm q}
&:=
\prod_{j=1}^{m}\hat X_j^{q_j},
\qquad
\hat{\bm P}^{\bm s}
:=
\prod_{j=1}^{m}\hat P_j^{s_j}.
\end{aligned}
\end{equation}
where $\bm q^{(k)},\bm s^{(k)}\in\mathbb{Z}_{\ge 0}^{m}$ are multi-indices; $o_k\in\mathbb{C}$.

\begin{theorem}[Polynomial observables]
\label{thm:readout_observables_main}
Given a classical description of the state
\eqref{eq:readout_gaussian_superposition_main}, describing an $m$-mode
bosonic system as a superposition of $N$ Gaussian components, and a Hermitian
polynomial observable $\hat O$ of the form
\eqref{eq:readout_poly_observable_main} with $\eta$ monomial terms and
quadrature-power multi-indices $\bm q^{(k)},\bm s^{(k)}$, the expectation value
$\bra{\psi_{\mathrm{out}}}\hat O\ket{\psi_{\mathrm{out}}}$ can be computed in
time
\begin{equation}
\mathcal{O}\!\left(N^2(m^3+\eta m^2)\right).
\end{equation}
If the output instead has the coherent-product form~\eqref{eq:coherent_superposition_background}, the cost improves to
\begin{equation}
\mathcal{O}\!\Big(N^2 m\big(\eta+q_{\max}s_{\max}\big)\Big),
\end{equation}
where $q_{\max}:=\max_{k,r} q_r^{(k)}$ and
$s_{\max}:=\max_{k,r} s_r^{(k)}$.
\end{theorem}
    A proof is provided in Appendix~\ref{subsec:readout-expectation-values}.\\

Taken together, Theorems~\ref{thm:readout_sampling_main} and~\ref{thm:readout_observables_main} show that the readout cost remains controlled as long as the number $N$ of Gaussian branches is moderate. The $N^2$ dependence has a simple origin: readout must keep the interference between all pairs of branches. For example,
\begin{equation}
\bra{\psi_{\mathrm{out}}}\hat O\ket{\psi_{\mathrm{out}}}
=
\sum_{i,j=1}^{N}
C_i^* C_j\,
\bra{G_i}\hat O\ket{G_j}.
\end{equation}
Thus, the algorithm evaluates pairwise Gaussian overlaps or cross-moments. The remaining factors in $m$ come from the cost of dense Gaussian algebra; when the components are product coherent states, this algebra factorizes over modes and the scaling improves. 

We end by noting that while the present guarantees cover quadrature sampling and polynomial quadrature observables, they do not include photon-number-resolving (Fock-basis) measurements, which form an important readout primitive in optical quantum information~\cite{Knill2001KLM}. 
We do not expect an efficient general treatment of photon-number-resolving measurements for arbitrary output states in our setting, since after the final Gaussian layer this would already subsume Gaussian boson sampling—a prominent candidate for quantum advantage—as a special case~\cite{PhysRevLett.119.170501}. Nevertheless, efficient photon-number readout may still be possible in more structured or restricted regimes, and identifying such cases remains an open problem.

\sectionnotoc{Numerical results}
\label{sec:numerics}

We benchmark coherent-state propagation on the driven Bose--Hubbard model in the rotating frame. The Bose--Hubbard model is a fundamental lattice model for interacting bosons: it captures the competition between coherent hopping and onsite interactions, and provides a standard framework for boson localization and the superfluid-Mott-insulator transition~\cite{Bose_hubbard2,Bose_hubbard_1}. Its driven version is also central in photonic lattices and cavity-array platforms, where coherent pumping injects phase-controlled amplitude into an interacting bosonic system~\cite{driven_bose_hubbard,driven_bose_hubbard_2}. A detailed discussion of the model is given in Appendix~\ref{sec:driven-bose-hubbard}.

We consider the `rotating-frame' form of the Bose-Hubbard model with the following Hamiltonian:
\begin{equation}
\label{eq:driven_bh_rot_numerics}
\begin{aligned}
\hat H_{\mathrm{rot}}
&=
-J\sum_{\langle i,j\rangle}
\left(
\hat a_i^\dagger \hat a_j
+
\hat a_j^\dagger \hat a_i
\right)
-\Delta\sum_i \hat n_i
\\
&\quad+\frac{U}{2}\sum_i \hat n_i^2
+
\sum_i
\left(
\Omega_i \hat a_i^\dagger
+
\Omega_i^* \hat a_i
\right).
\end{aligned}
\end{equation}
Here $J$ is the hopping strength, $U$ is the onsite interaction, $\Delta$ is the detuning, and $\Omega_i$ is the coherent drive on mode $i$. The notation $\langle i,j\rangle$ denotes the edge set of the chosen hopping topology, specifying which modes are coupled.
This Hamiltonian is well matched to our circuit structure in Fig.~\ref{fig:circuit_scheme}: the hopping, detuning, and drive terms form the displaced passive linear-optical part, while the onsite interaction gives the local Kerr part. Based on the model \eqref{eq:driven_bh_rot_numerics} we report two numerical benchmarks in a regime of small non-Gaussianity; full details are given in Appendix~\ref{subsec:numerics_supplementary}.

The first benchmark compares coherent-state propagation with a dense Fock-basis statevector simulation for a three-mode all-to-all system, with Bose--Hubbard parameters given in Table~\ref{tab:runtime_fock}. The Fock simulator works in a truncated tensor-product Fock space and is therefore \emph{exact} only within the corresponding finite-dimensional cutoff subspace, where it retains the full statevector and does not rely on any further approximation of entanglement or non-Gaussianity; its cost, however, grows rapidly with both the number of modes $m$ and the local cutoff. For coherent-state propagation, we use the finite-Fourier Kerr update in Eq.~\eqref{eq:weak_kerr_finite_fourier_main} with $M=3$ and largest-coefficient cutoff $\mathcal{S}=10^4$.

The main purpose of this small-scale benchmark is to validate the coherent-state propagation and readout procedure in a setting where a full truncated-Fock reference calculation is still available.  We observed that the coherent-state curves closely follow the dense Fock-reference trajectories in Fig.~\ref{fig:numerics_three_observables}, validating the propagation and readout pipeline.

\begin{table}[htbp]
\centering
\caption{Simulation-time comparison for the three-mode driven Bose--Hubbard benchmark. Timings were measured on the same workstation equipped with an NVIDIA RTX 6000 Ada Generation GPU.}
\label{tab:runtime_fock}
\begin{tabular}{lcc}
\toprule
 & Fock exact & Coherent prop. \\
\midrule
Evolution  & 1775.7~s & 0.2~s \\
Observable & 46~ms    & 284~ms \\
\bottomrule
\end{tabular}
\end{table}

Table~\ref{tab:runtime_fock} shows the corresponding runtime: ``Evolution'' denotes the total state-propagation time, while ``Observable'' denotes the time to evaluate one observable after propagation. The speedup appears mainly in the evolution step, where only the retained coherent-state branches are propagated. By contrast, observable evaluation requires pairwise overlaps and matrix elements between the $N$ propagated coherent branches, giving the $N^2$ dependence of Theorem~\ref{thm:readout_observables_main}. The readout cost can therefore be significant even when propagation is fast, and the overall advantage relies on the small non-Gaussian regime where the retained branch number $N$ remains moderate. Nevertheless, we find a regime in which coherent-state propagation gives a significant speedup for the evolution step while still simulating strictly non-Gaussian continuous-variable dynamics.

The second benchmark aims to go beyond the small systems accessible to dense Fock simulation. We consider a larger $m=6$ all-to-all system, with Bose--Hubbard parameters given in Table~\ref{tab:mps_numerical_parameters} and as a benchmark we use matrix-product-state reference simulation~\cite{A_J_Daley_2004,SCHOLLWOCK201196}. Figure~\ref{fig:mps_benchmark}(a) plays the same role as the exact-Fock comparison in Fig.~\ref{fig:numerics_three_observables}: it shows that the coherent-state trajectory closely follows the reference trajectory for $\langle \hat X_0\hat X_1\rangle$.
Figures~\ref{fig:mps_benchmark}(b) and~\ref{fig:mps_benchmark}(c) show the relative errors at time $T=1$ as functions of the largest-coefficient cutoff $\mathcal{S}$ and the finite-Fourier parameter $M$. 
These panels therefore illustrate the accuracy--cost tradeoff of coherent-state propagation.

Together, Figs.~\ref{fig:numerics_three_observables} and~\ref{fig:mps_benchmark} show that coherent-state propagation agrees with standard reference methods in both a small exact-Fock setting and a larger MPS setting. The benchmarks also show the main practical caveat: the method is advantageous when the coherent-state expansion remains sparse, which is the small non-Gaussianity regime explored here.

\sectionnotoc{Discussion}

We have introduced coherent-state propagation as a Schr\"odinger-picture framework for simulating bosonic dynamics by representing the evolving wavefunction as a sparse superposition of coherent states. 

For displaced linear optics augmented by Kerr nonlinearities, this representation is well matched to the dynamics and yields rigorous simulation guarantees in regimes with few Kerr gates or weak nonlinearities.
A natural open question is that of \textit{optimality}. Our results upper bound the size of the coherent-state representation, but not whether this scaling is tight. Extending recent lower bounds on approximate coherent-state rank to the states arising here could rule out substantially smaller coherent-state decompositions\ \cite{cottier2026lower}.

{An important practical question is how coherent-state propagation compares with tensor-network methods such as matrix-product states. We do not claim a practical advantage over MPS from the present numerics. The two approaches instead exploit different structures: MPS compresses low-entanglement dynamics, whereas coherent-state propagation exploits sparsity in a coherent-state expansion. Identifying regimes where the latter is superior remains an important open question.}

Several extensions to our algorithmic framework would be valuable to explore. The current framework is tailored to state propagation and to observables naturally adapted to this representation; photon-number-resolving readout is not yet handled efficiently. Extending the method to Fock-basis measurement would broaden its scope toward boson-sampling-type tasks\ \cite{aaronson2011computational, hamilton2017gaussian}, including Gaussian boson sampling with weak nonlinear corrections or other non-Gaussian resources\ \cite{oszmaniec2017universal}. 

More generally, our approach is naturally tailored to circuits dominated by displaced linear optics, while squeezing is not represented efficiently in the current ansatz. Although the architecture considered here is already universal, extending the framework to displaced-squeezed states could make some physically relevant dynamics accessible through more compact representations, and therefore potentially more efficient simulation, within the same propagation paradigm.
Likewise, because the method is formulated for closed-system unitary dynamics, incorporating general noise or open-system evolution will likely require an operator-picture variant, a purification-based approach, or stochastic trajectories.

A further direction is hybrid boson--qubit dynamics, in which coherent-state propagation on the bosonic sector is combined with suitable representations for finite-dimensional degrees of freedom\ \cite{araz2025hybrid, brenner2025trading, liu2026hybrid}. Such hybridizations are especially attractive for studying open many-body quantum dynamics, where bosonic modes often serve as environments \cite{so2025quantum}. When combined with a noise-based extension, this framework could provide a new route to probing the classical simulability of open-system evolution, including regimes with genuinely non-Markovian effects, where environmental memory plays an essential role\ \cite{rivas2014quantum, breuer2016colloquium, devega2017dynamics}. One possible approach would be to combine the present method with Pauli propagation, using the representation best suited to each sector. Similarly, combining coherent-state propagation with Majorana propagation could enable propagation-based simulation methods for hybrid boson--fermion systems, for example to study the role of phononic modes in molecular systems\ \cite{denner2023hybrid, kumar2025digital}.

\medskip

\noindent\emph{Note added.} During the preparation of this work, we became aware of independent work by Upreti, Quesada, and Chabaud~\cite{UpretiQuesadaChabaud2026}, which studies classical simulation algorithms for bosonic circuits composed of Gaussian and Kerr gates. \\

\starsectionnotoc{Acknowledgments}
The authors thank Ulysse Chabaud, Ricard Puig and Varun Upreti for fruitful discussions. NG is grateful to the EPFL team for their hospitality and support. NG warmly thanks Nana Liu, his supervisor, for her thoughtful guidance and support.  NG acknowledges funding from NSFC grant W2442002. ZH and AA acknowledge support from the Sandoz Family Foundation-Monique de Meuron program for Academic Promotion. \\

\starsectionnotoc{Data Availability statement}\label{sec:data availability}

The code used to generate the numerical results and figures in this work, including the implementations of coherent-state propagation, Kerr-gate approximation routines, and the observable-estimation and sampling workflows, is available in a GitHub repository \cite{guseynov_gitlab_placeholder}. No external datasets were used in this study; all data were generated numerically from the models and parameters specified in the main text and appendices.

\let\oldaddcontentsline\addcontentsline 
\renewcommand{\addcontentsline}[3]{}    
\bibliography{references,quantum}

@article{rall2019simulation,
  title={Simulation of qubit quantum circuits via Pauli propagation},
  author={Rall, Patrick and Liang, Daniel and Cook, Jeremy and Kretschmer, William},
  journal={Physical Review A},
  volume={99},
  number={6},
  pages={062337},
  year={2019},
  publisher={APS},
url={https://journals.aps.org/pra/abstract/10.1103/PhysRevA.99.062337},
doi={10.1103/PhysRevA.99.062337}
}

@article{lerch2024efficient,
  title={Efficient quantum-enhanced classical simulation for patches of quantum landscapes},
  author={Lerch, Sacha and Puig, Ricard and Rudolph, Manuel and  Angrisani, Armando and Jones, Tyson and Cerezo, M. and Thanasilp, Supanut and Holmes, Zo\"e},
  journal={arXiv preprint arXiv:2411.19896},
  year={2024},
  url={https://arxiv.org/abs/2411.19896},
  doi = {10.48550/arXiv.2411.19896}
}

@article{cottier2026lower,
  title={Lower Bounds on Coherent State Rank},
  author={Cottier, Florian and Chabaud, Ulysse},
  journal={arXiv preprint arXiv:2604.00766},
  year={2026},
  url = {https://arxiv.org/abs/2604.00766},
  doi= {10.48550/arXiv.2604.00766}
}

@article{rivas2014quantum,
  title={Quantum non-Markovianity: characterization, quantification and detection},
  author={Rivas, {\'A}ngel and Huelga, Susana F and Plenio, Martin B},
  journal={Reports on Progress in Physics},
  volume={77},
  number={9},
  pages={094001},
  year={2014},
  publisher={IOP Publishing},
  url ={https://iopscience.iop.org/article/10.1088/0034-4885/77/9/094001},
  doi ={10.1088/0034-4885/77/9/094001}
}

@article{feynman1982simulating,
author={Feynman, Richard P.},
title={Simulating physics with computers},
journal={International Journal of Theoretical Physics},
year={1982},
month={Jun},
day={01},
volume={21},
number={6},
pages={467-488},
issn={1572-9575},
doi={10.1007/BF02650179},
url={https://doi.org/10.1007/BF02650179}
}

@article{mari2012positive,
  title={Positive Wigner functions render classical simulation of quantum computation efficient},
  author={Mari, Andrea and Eisert, Jens},
  journal={Physical review letters},
  volume={109},
  number={23},
  pages={230503},
  year={2012},
  publisher={APS},
  url ={https://journals.aps.org/prl/abstract/10.1103/PhysRevLett.109.230503},
  doi= {10.1103/PhysRevLett.109.230503}
}

@article{fuller2025improved,
  title={Improved Quantum Computation using Operator Backpropagation},
  author={Fuller, Bryce and Tran, Minh C and Lykov, Danylo and Johnson, Caleb and Rossmannek, Max and Wei, Ken Xuan and He, Andre and Kim, Youngseok and Vu, DinhDuy and Sharma, Kunal and others},
  journal={arXiv preprint arXiv:2502.01897},
  year={2025},
  url = {https://arxiv.org/abs/2502.01897},
  doi = {https://doi.org/10.48550/arXiv.2502.01897}
}

@article{knill2001fermionic,
      title={Fermionic Linear Optics and Matchgates}, 
      author={E. Knill},
      year={2001},
  url = {https://arxiv.org/abs/quant-ph/0108033},
  journal={arXiv preprint arXiv:quant-ph/0108033}, 
doi={10.48550/arXiv.quant-ph/0108033} 
}

@inproceedings{gottesman1998heisenberg,
  title={The heisenberg representation of quantum computers, talk at},
  author={Gottesman, Daniel},
  booktitle={International Conference on Group Theoretic Methods in Physics},
  year={1998},
  url={http://citeseerx.ist.psu.edu/viewdoc/summary?doi=10.1.1.252.9446},
  organization={Citeseer}
}

@inproceedings{aaronson2011computational,
  title={The computational complexity of linear optics},
  author={Aaronson, Scott and Arkhipov, Alex},
  booktitle={Proceedings of the forty-third annual ACM symposium on Theory of computing},
  pages={333--342},
  year={2011},
  url={https://dl.acm.org/doi/abs/10.1145/1993636.1993682},
  doi={10.1145/1993636.1993682}
}

@article{hamilton2017gaussian,
  title={Gaussian boson sampling},
  author={Hamilton, Craig S and Kruse, Regina and Sansoni, Linda and Barkhofen, Sonja and Silberhorn, Christine and Jex, Igor},
  journal={Physical review letters},
  volume={119},
  number={17},
  pages={170501},
  year={2017},
  publisher={APS},
  url = {https://journals.aps.org/prl/abstract/10.1103/PhysRevLett.119.170501},
  doi = {10.1103/PhysRevLett.119.170501}
}

@article{upreti2025quantum,
  title={When quantum resources backfire: Non-gaussianity and symplectic coherence in noisy bosonic circuits},
  author={Upreti, Varun and Chabaud, Ulysse and Holmes, Zo{\"e} and Angrisani, Armando},
  journal={arXiv preprint arXiv:2510.07264},
  year={2025},
  url = {https://arxiv.org/abs/2510.07264}, 
  doi={
https://doi.org/10.48550/arXiv.2510.07264
}
}

@article{bartlett2002universal,
  title = {Universal continuous-variable quantum computation: Requirement of optical nonlinearity for photon counting},
  author = {Bartlett, Stephen D and Sanders, Barry C},
  journal = {Phys. Rev. A},
  volume = {65},
  issue = {4},
  pages = {042304},
  numpages = {5},
  year = {2002},
  month = {Mar},
  publisher = {American Physical Society},
  doi = {10.1103/PhysRevA.65.042304},
  url = {https://link.aps.org/doi/10.1103/PhysRevA.65.042304}
}

@article{upreti2025interplay,
  title={Interplay of resources for universal continuous-variable quantum computing},
  author={Upreti, Varun and Chabaud, Ulysse},
  journal={arXiv preprint arXiv:2502.07670},
  year={2025},
  doi = {10.48550/arXiv.2502.07670},
  url = {https://arxiv.org/abs/2502.07670}
}

@article{aharonov2022polynomial,
  title={A polynomial-time classical algorithm for noisy random circuit sampling},
  author={Aharonov, Dorit and Gao, Xun and Landau, Zeph and Liu, Yunchao and Vazirani, Umesh},
  journal={Proceedings of the 55th Annual ACM Symposium on Theory of Computing},
  pages={945},
  year={2023},
  doi={10.1145/3564246.3585234},
  url={https://dl.acm.org/doi/abs/10.1145/3564246.3585234},
  publisher={Association for Computing Machinery},
  address={New York, NY, USA},
}

@article{terhal2002classical,
  title={Classical simulation of noninteracting-fermion quantum circuits},
  author={Terhal, Barbara M and DiVincenzo, David P},
  journal={Physical Review A},
  volume={65},
  number={3},
  pages={032325},
  year={2002},
  publisher={APS},
  doi = {10.1103/PhysRevA.65.032325},
  url = {https://link.aps.org/doi/10.1103/PhysRevA.65.032325}  
}

@inproceedings{valiant2001quantum,
  title={Quantum computers that can be simulated classically in polynomial time},
  author={Valiant, Leslie G},
  booktitle={Proceedings of the thirty-third annual ACM symposium on Theory of computing},
  pages={114--123},
  year={2001},
  url = {https://doi.org/10.1145/380752.380785},
	doi = {10.1145/380752.380785},
}

@article{aaronson2004improved,
  title={Improved simulation of stabilizer circuits},
  author={Aaronson, Scott and Gottesman, Daniel},
  journal={Physical Review A},
  volume={70},
  number={5},
  pages={052328},
  year={2004},
  publisher={APS},
  doi = {10.1103/PhysRevA.70.052328},
  url = {https://link.aps.org/doi/10.1103/PhysRevA.70.052328}
}

@article{reardon2023improved,
  title={Improved simulation of quantum circuits dominated by free fermionic operations},
  author={Reardon-Smith, Oliver and Oszmaniec, Micha{\l} and Korzekwa, Kamil},
  journal={arXiv preprint arXiv:2307.12702},
  year={2023},
url = {
https://doi.org/10.48550/arXiv.2307.12702}
}

@article{dias2023classical,
  title={Classical simulation of non-Gaussian fermionic circuits},
  author={Dias, Beatriz and Koenig, Robert},
  journal={Quantum},
  volume={8},
  pages={1350},
  year={2024},
  publisher={Verein zur F{\"o}rderung des Open Access Publizierens in den Quantenwissenschaften},
url={https://quantum-journal.org/papers/q-2024-05-21-1350/},
doi={10.22331/q-2019-09-02-181}
}

@article{dias2024classical,
  title={Classical simulation of non-Gaussian bosonic circuits},
  author={Dias, Beatriz and K{\"o}nig, Robert},
  journal={Physical Review A},
  volume={110},
  number={4},
  pages={042402},
  year={2024},
  publisher={APS},
  url ={https://journals.aps.org/pra/abstract/10.1103/PhysRevA.110.042402},
  doi ={PhysRevA.110.042402}
}

@article{beguvsic2023simulating,
  title={Simulating quantum circuit expectation values by {C}lifford perturbation theory},
  author={Begu{\v{s}}i{\'c}, Tomislav and Hejazi, Kasra and Chan, Garnet Kin},
  journal={The Journal of Chemical Physics},
  volume={162},
  number={15},
  year={2025},
  publisher={AIP Publishing},
url={https://pubs.aip.org/aip/jcp/article/162/15/154110/3344130/Simulating-quantum-circuit-expectation-values-by},
doi={10.1063/5.0269149
}
}

@article{beguvsic2024real,
  title={Real-time operator evolution in two and three dimensions via sparse Pauli dynamics},
  author={Begu{\v{s}}i{\'c}, Tomislav and Chan, Garnet Kin-Lic},
  journal={PRX Quantum},
  volume={6},
  number={2},
  pages={020302},
  year={2025},
  publisher={APS},
  url={https://doi.org/10.1103/PRXQuantum.6.020302},
}

@article{rudolph2026thermal,
  title={Thermal State Simulation with Pauli and Majorana Propagation},
  author={Rudolph, Manuel S and Angrisani, Armando and Wright, Andrew and Sanderski, Iwo and Puig, Ricard and Holmes, Zo{\"e}},
  journal={arXiv preprint arXiv:2602.04878},
  year={2026},
  url = {https://www.arxiv.org/abs/2602.04878},
  doi = {https://doi.org/10.48550/arXiv.2602.04878}
}

@article{miller2025simulation,
  title={Simulation of Fermionic circuits using Majorana Propagation},
  author={Miller, Aaron and Favre, Joachim and Holmes, Zo{\"e} and Salehi, {\"O}zlem and Chakraborty, Rahul and Nyk{\"a}nen, Anton and Zimboras, Zoltan and Glos, Adam and Garc{\'\i}a-P{\'e}rez, Guillermo},
  journal={arXiv preprint arXiv:2503.18939},
  year={2025},
  url={https://doi.org/10.48550/arXiv.2503.18939}, 
doi={10.48550/arXiv.2503.18939}
}

@article{d2025majorana,
  title={Majorana string simulation of nonequilibrium dynamics in two-dimensional lattice fermion systems},
  author={D'Anna, Matteo and Nys, Jannes and Carrasquilla, Juan},
  journal={ arXiv:2511.02809},
  year={2025},
  url={https://www.arxiv.org/abs/2511.02809},
}

@article{gonzalez2024pauli,
  title={Pauli path simulations of noisy quantum circuits beyond average case},
  author={Gonz{\'a}lez-Garc{\'\i}a, Guillermo and Cirac, J Ignacio and Trivedi, Rahul},
  journal={Quantum},
  volume={9},
  pages={1730},
  year={2025},
  publisher={Verein zur F{\"o}rderung des Open Access Publizierens in den Quantenwissenschaften},
  url ={https://quantum-journal.org/papers/q-2025-05-05-1730/},
  doi ={10.22331/q-2025-05-05-1730}
}

@article{bravyi2016improved,
  title = {Improved Classical Simulation of Quantum Circuits Dominated by Clifford Gates},
  author = {Bravyi, Sergey and Gosset, David},
  journal = {Phys. Rev. Lett.},
  volume = {116},
  issue = {25},
  pages = {250501},
  numpages = {5},
  year = {2016},
  month = {Jun},
  publisher = {American Physical Society},
  doi = {10.1103/PhysRevLett.116.250501},
  url = {https://link.aps.org/doi/10.1103/PhysRevLett.116.250501}
}

@article{angrisani2025simulating,
  title={Simulating quantum circuits with arbitrary local noise using Pauli Propagation},
  author={Angrisani, Armando and Mele, Antonio A and Rudolph, Manuel S and Cerezo, M and Holmes, Zoe},
  journal={arXiv preprint arXiv:2501.13101},
  year={2025},
  url={https://arxiv.org/abs/2501.13101},
  doi ={10.48550/arXiv.2501.13101}
}

@article{angrisani2024classically,
  title = {Classically Estimating Observables of Noiseless Quantum Circuits},
  author = {Angrisani, Armando and Schmidhuber, Alexander and Rudolph, Manuel S. and Cerezo, M. and Holmes, Zo\"e and Huang, Hsin-Yuan},
  journal = {Phys. Rev. Lett.},
  volume = {135},
  issue = {17},
  pages = {170602},
  numpages = {10},
  year = {2025},
  month = {Oct},
  publisher = {American Physical Society},
  doi = {10.1103/lh6x-7rc3},
  url = {https://link.aps.org/doi/10.1103/lh6x-7rc3}
}

@article{chabaud2021classical,
  title={Classical simulation of Gaussian quantum circuits with non-Gaussian input states},
  author={Chabaud, Ulysse and Ferrini, Giulia and Grosshans, Fr{\'e}d{\'e}ric and Markham, Damian},
  journal={Physical Review Research},
  volume={3},
  number={3},
  pages={033018},
  year={2021},
  publisher={APS},
  url ={https://journals.aps.org/prresearch/abstract/10.1103/PhysRevResearch.3.033018},
  doi= {10.1103/PhysRevResearch.3.033018}
}

@article{pashayan2015estimating,
  title={Estimating outcome probabilities of quantum circuits using quasiprobabilities},
  author={Pashayan, Hakop and Wallman, Joel J and Bartlett, Stephen D},
  journal={Physical review letters},
  volume={115},
  number={7},
  pages={070501},
  year={2015},
  publisher={APS},
  url ={https://journals.aps.org/prl/abstract/10.1103/PhysRevLett.115.070501},
  doi ={10.1103/PhysRevLett.115.070501}
}

@article{araz2025hybrid,
  title={Hybrid quantum simulations with qubits and qumodes on trapped-ion platforms},
  author={Araz, Jack Y and Grau, Matt and Montgomery, Jake and Ringer, Felix},
  journal={Physical Review A},
  volume={112},
  number={1},
  pages={012620},
  year={2025},
  publisher={APS},
  url ={https://journals.aps.org/pra/abstract/10.1103/kbv4-jj51},
  doi ={10.1103/kbv4-jj51}
}

@article{oszmaniec2017universal,
  title={Universal extensions of restricted classes of quantum operations},
  author={Oszmaniec, Michal and Zimboras, Zoltan},
  journal={Physical review letters},
  volume={119},
  pages={220502},
  year={2017},
  publisher={APS},
  doi={10.1103/PhysRevLett.119.220502},
  url={https://journals.aps.org/prl/abstract/10.1103/PhysRevLett.119.220502}
}

@article{schuster2024polynomial,
  title={A polynomial-time classical algorithm for noisy quantum circuits},
  author={Schuster, Thomas and Yin, Chao and Gao, Xun and Yao, Norman Y},
  journal={Physical Review X},
  volume={15},
  number={4},
  pages={041018},
  year={2025},
  publisher={APS},
  url={https://journals.aps.org/prx/abstract/10.1103/xct1-7kf2},
  doi = {10.1103/xct1-7kf2}
}

@article{so2025quantum,
  title={Quantum simulation of charge and exciton transfer in multi-mode models using engineered reservoirs},
  author={So, Visal and Duraisamy Suganthi, Midhuna and Zhu, Mingjian and Menon, Abhishek and Tomaras, George and Zhuravel, Roman and Pu, Han and Wolynes, Peter G and Onuchic, Jos{\'e} N and Pagano, Guido},
  journal={Nature Communications},
  year={2025},
  publisher={Nature Publishing Group UK London},
  url ={https://www.nature.com/articles/s41467-025-67116-6},
  doi={10.1038/s41467-025-67116-6}
}

@article{rahimi2016sufficient,
  title={Sufficient conditions for efficient classical simulation of quantum optics},
  author={Rahimi-Keshari, Saleh and Ralph, Timothy C and Caves, Carlton M},
  journal={Physical Review X},
  volume={6},
  number={2},
  pages={021039},
  year={2016},
  publisher={APS},
  url = {https://journals.aps.org/prx/abstract/10.1103/PhysRevX.6.021039},
  doi ={10.1103/PhysRevX.6.021039}
}

@article{breuer2016colloquium,
  title={Colloquium: Non-Markovian dynamics in open quantum systems},
  author={Breuer, Heinz-Peter and Laine, Elsi-Mari and Piilo, Jyrki and Vacchini, Bassano},
  journal={Reviews of Modern Physics},
  volume={88},
  number={2},
  pages={021002},
  year={2016},
  publisher={APS},
  doi={10.1103/RevModPhys.88.021002},
  url={https://doi.org/10.1103/RevModPhys.88.021002}
}

@article{denner2023hybrid,
  title={A hybrid quantum-classical method for electron-phonon systems},
  author={Denner, M Michael and Miessen, Alexander and Yan, Haoran and Tavernelli, Ivano and Neupert, Titus and Demler, Eugene and Wang, Yao},
  journal={Communications Physics},
  volume={6},
  number={1},
  pages={233},
  year={2023},
  publisher={Nature Publishing Group UK London},
  url ={https://www.nature.com/articles/s42005-023-01353-3},
  doi ={10.1038/s42005-023-01353-3}
}

@article{kumar2025digital,
  title={Digital-analog quantum computing of fermion-boson models in superconducting circuits},
  author={Kumar, Shubham and Hegade, Narendra N and Visuri, Anne-Maria and Bhargava, Balaganchi A and Hernandez, Juan FR and Solano, Enrique and Albarr{\'a}n-Arriagada, Francisco and Barrios, G Alvarado},
  journal={npj Quantum Information},
  volume={11},
  number={1},
  pages={43},
  year={2025},
  publisher={Nature Publishing Group UK London},
  doi = {10.1038/s41534-025-01001-4},
  url ={https://www.nature.com/articles/s41534-025-01001-4}
}

@article{devega2017dynamics,
  title={Dynamics of non-Markovian open quantum systems},
  author={De Vega, In{\'e}s and Alonso, Daniel},
  journal={Reviews of Modern Physics},
  volume={89},
  number={1},
  pages={015001},
  year={2017},
  publisher={APS},
  doi={10.1103/RevModPhys.89.015001},
  url={https://doi.org/10.1103/RevModPhys.89.015001}
}

@article{rudolph2025pauli,
  title={Pauli Propagation: A Computational Framework for Simulating Quantum Systems},
  author={Rudolph, Manuel S and Jones, Tyson and Teng, Yanting and Angrisani, Armando and Holmes, Zo{\"e}},
  journal={arXiv preprint arXiv:2505.21606},
  year={2025},
  url={https://arxiv.org/abs/2505.21606},
doi={10.48550/arXiv.2505.21606}
}

@article{liu2026hybrid,
  title={Hybrid oscillator-qubit quantum processors: Instruction set architectures, abstract machine models, and applications},
  author={Liu, Yuan and Singh, Shraddha and Smith, Kevin C and Crane, Eleanor and Martyn, John M and Eickbusch, Alec and Schuckert, Alexander and Li, Richard D and Sinanan-Singh, Jasmine and Soley, Micheline B and others},
  journal={PRX Quantum},
  volume={7},
  number={1},
  pages={010201},
  year={2026},
  publisher={APS},
  url ={https://journals.aps.org/prxquantum/abstract/10.1103/4rf7-9tfx},
  doi ={10.1103/4rf7-9tfx}
}

@article{brenner2025trading,
  title={Trading modes against energy},
  author={Brenner, Lukas and Dias, Beatriz and Koenig, Robert},
  journal={arXiv preprint arXiv:2509.18854},
  year={2025},
  url ={https://arxiv.org/abs/2509.18854},
  doi ={10.48550/arXiv.2509.18854}
}

@article{teng2025leveraging,
  title={Leveraging Symmetry Merging in Pauli Propagation},
  author={Teng, Yanting and Chang, Su Yeon and Rudolph, Manuel S and Holmes, Zo{\"e}},
  journal={arXiv preprint arXiv:2512.12094},
  year={2025}, 
url={https://doi.org/10.48550/arXiv.2512.12094}, 
doi={10.48550/arXiv.2512.12094}
}

@article{facelli2026fast,
  title={Fast convergence of Majorana Propagation for weakly interacting fermions},
  author={Facelli, Giorgio and Fawzi, Hamza and Fawzi, Omar},
  journal={arXiv preprint arXiv:2601.05226},
  year={2026},
  url = {https://arxiv.org/abs/2601.05226},
  doi = {
https://doi.org/10.48550/arXiv.2601.05226}
}

@misc{UpretiQuesadaChabaud2026,
  title         = {Exponentially-improved effective descriptions of physical bosonic systems},
  author        = {Varun Upreti and Nicol{\'a}s Quesada and Ulysse Chabaud},
  year          = {2026},
  month         = apr,
  eprint        = {2604.XXXX},
  archivePrefix = {arXiv},
  primaryClass  = {quant-ph}
}

@article{PhysRevLett.80.869,
  title = {Teleportation of Continuous Quantum Variables},
  author = {Braunstein, Samuel L. and Kimble, H. J.},
  journal = {Phys. Rev. Lett.},
  volume = {80},
  issue = {4},
  pages = {869--872},
  numpages = {0},
  year = {1998},
  month = {Jan},
  publisher = {American Physical Society},
  doi = {10.1103/PhysRevLett.80.869},
  url = {https://link.aps.org/doi/10.1103/PhysRevLett.80.869}
}

@article{PhysRevLett.88.097904,
  title = {Efficient Classical Simulation of Continuous Variable Quantum Information Processes},
  author = {Bartlett, Stephen D. and Sanders, Barry C. and Braunstein, Samuel L. and Nemoto, Kae},
  journal = {Phys. Rev. Lett.},
  volume = {88},
  issue = {9},
  pages = {097904},
  numpages = {4},
  year = {2002},
  month = {Feb},
  publisher = {American Physical Society},
  doi = {10.1103/PhysRevLett.88.097904},
  url = {https://link.aps.org/doi/10.1103/PhysRevLett.88.097904}
}

@article{Knill2001KLM,
  title   = {A scheme for efficient quantum computation with linear optics},
  author  = {Knill, E. and Laflamme, R. and Milburn, G. J.},
  journal = {Nature},
  volume  = {409},
  number  = {6816},
  pages   = {46--52},
  year    = {2001},
  doi     = {10.1038/35051009}
}

@misc{upreti2025boundingcomputationalpowerbosonic,
      title={Bounding the computational power of bosonic systems}, 
      author={Varun Upreti and Dorian Rudolph and Ulysse Chabaud},
      year={2025},
      eprint={2503.03600},
      archivePrefix={arXiv},
      primaryClass={quant-ph},
      url={https://arxiv.org/abs/2503.03600}, 
}

@article{Corless1996,
  author    = {R. M. Corless and G. H. Gonnet and D. E. G. Hare and D. J. Jeffrey and D. E. Knuth},
  title     = {On the LambertW function},
  journal   = {Advances in Computational Mathematics},
  year      = {1996},
  volume    = {5},
  number    = {1},
  pages     = {329--359},
  doi       = {10.1007/BF02124750},
  url       = {https://doi.org/10.1007/BF02124750},
  issn      = {1572-9044}
}

@article{hahn2025classical,
  title={Classical simulation and quantum resource theory of non-Gaussian optics},
  author={Hahn, Oliver and Takagi, Ryuji and Ferrini, Giulia and Yamasaki, Hayata},
  journal={Quantum},
  volume={9},
  pages={1881},
  year={2025},
  publisher={Verein zur F{\"o}rderung des Open Access Publizierens in den Quantenwissenschaften},
  url={https://doi.org/10.22331/q-2025-10-13-1881}
}

@article{MengersenTweedie1996Rates,
  author    = {Mengersen, K. L. and Tweedie, R. L.},
  title     = {Rates of Convergence of the {Hastings} and {Metropolis} Algorithms},
  journal   = {The Annals of Statistics},
  year      = {1996},
  volume    = {24},
  number    = {1},
  pages     = {101--121},
  month     = feb,
  doi       = {10.1214/aos/1033066201},
  url       = {https://doi.org/10.1214/aos/1033066201},
  publisher = {Institute of Mathematical Statistics}
}

@article{PhysRevLett.119.170501,
  title = {Gaussian Boson Sampling},
  author = {Hamilton, Craig S. and Kruse, Regina and Sansoni, Linda and Barkhofen, Sonja and Silberhorn, Christine and Jex, Igor},
  journal = {Phys. Rev. Lett.},
  volume = {119},
  issue = {17},
  pages = {170501},
  numpages = {5},
  year = {2017},
  month = {Oct},
  publisher = {American Physical Society},
  doi = {10.1103/PhysRevLett.119.170501},
  url = {https://link.aps.org/doi/10.1103/PhysRevLett.119.170501}
}

@article{Marshall:23,
author = {Jeffrey Marshall and Namit Anand},
journal = {Optica Quantum},
number = {2},
pages = {78--93},
publisher = {Optica Publishing Group},
title = {Simulation of quantum optics by coherent state decomposition},
volume = {1},
month = {Dec},
year = {2023},
url = {https://opg.optica.org/opticaq/abstract.cfm?URI=opticaq-1-2-78},
doi = {10.1364/OPTICAQ.504311},
}

@misc{chabaud2025energybosonscomputationalcomplexity,
      title={Energy, Bosons and Computational Complexity}, 
      author={Ulysse Chabaud and Sevag Gharibian and Saeed Mehraban and Arsalan Motamedi and Hamid Reza Naeij and Dorian Rudolph and Dhruva Sambrani},
      year={2025},
      eprint={2510.08545},
      archivePrefix={arXiv},
      primaryClass={quant-ph},
      url={https://arxiv.org/abs/2510.08545}, 
}

@Article{10.21468/SciPostPhysCore.4.3.025,
	title={{Bosonic and fermionic Gaussian states from Kähler structures}},
	author={Lucas Hackl and Eugenio Bianchi},
	journal={SciPost Phys. Core},
	volume={4},
	pages={025},
	year={2021},
	publisher={SciPost},
	doi={10.21468/SciPostPhysCore.4.3.025},
	url={https://scipost.org/10.21468/SciPostPhysCore.4.3.025},
}

@article{PhysRevA.110.012607,
  title = {Simulating quantum field theories on continuous-variable quantum computers},
  author = {Abel, Steven and Spannowsky, Michael and Williams, Simon},
  journal = {Phys. Rev. A},
  volume = {110},
  issue = {1},
  pages = {012607},
  numpages = {13},
  year = {2024},
  month = {Jul},
  publisher = {American Physical Society},
  doi = {10.1103/PhysRevA.110.012607},
  url = {https://link.aps.org/doi/10.1103/PhysRevA.110.012607}
}

@book{Gaussian_States_in_Quantum_Information,
title = "Gaussian States in Quantum Information",
author = "A. Ferraro and S. Olivares and Paris, \{Matteo G. A.\}",
year = "2005",
language = "English",
isbn = "88-7088-483-X",
series = "Napoli Series on physics and Astrophysics",
publisher = "Bibliopolis",
}

@techreport{do2008more,
  title       = {More on Multivariate {G}aussians},
  author      = {Do, Chuong B.},
  institution = {Stanford University},
  journal     = {CS229: Machine Learning (Lecture Notes)},
  year        = {2008},
  month       = nov,
  day         = {21},
  url         = {https://cs229.stanford.edu/section/more_on_gaussians.pdf},
  note        = {Course notes for CS229 (Machine Learning), Stanford University}
}

@article{arzani2025can,
  title={Can effective descriptions of bosonic systems be considered complete?},
  author={Arzani, Francesco and Booth, Robert I and Chabaud, Ulysse},
  journal={arXiv preprint arXiv:2501.13857},
  year={2025}
}

@article{loock_PhysRevLett.107.170501,
  title = {How to Decompose Arbitrary Continuous-Variable Quantum Operations},
  author = {Sefi, Seckin and van Loock, Peter},
  journal = {Phys. Rev. Lett.},
  volume = {107},
  issue = {17},
  pages = {170501},
  numpages = {5},
  year = {2011},
  month = {Oct},
  publisher = {American Physical Society},
  doi = {10.1103/PhysRevLett.107.170501},
  url = {https://link.aps.org/doi/10.1103/PhysRevLett.107.170501}
}

@article{PhysRevX.15.011070,
  title = {Preserving Phase Coherence and Linearity in Cat Qubits with Exponential Bit-Flip Suppression},
  author = {Putterman, Harald and Noh, Kyungjoo and Patel, Rishi N. and Peairs, Gregory A. and MacCabe, Gregory S. and Lee, Menyoung and Aghaeimeibodi, Shahriar and Hann, Connor T. and Jarrige, Ignace and Marcaud, Guillaume and He, Yuan and Moradinejad, Hesam and Owens, John Clai and Scaffidi, Thomas and Arrangoiz-Arriola, Patricio and Iverson, Joe and Levine, Harry and Brand\~ao, Fernando G. S. L. and Matheny, Matthew H. and Painter, Oskar},
  journal = {Phys. Rev. X},
  volume = {15},
  issue = {1},
  pages = {011070},
  numpages = {31},
  year = {2025},
  month = {Mar},
  publisher = {American Physical Society},
  doi = {10.1103/PhysRevX.15.011070},
  url = {https://link.aps.org/doi/10.1103/PhysRevX.15.011070}
}

@article{Bose_hubbard_1,
  author  = {Greiner, Markus and Mandel, Olaf and Esslinger, Tilman and H{\"a}nsch, Theodor W. and Bloch, Immanuel},
  title   = {Quantum phase transition from a superfluid to a {Mott} insulator in a gas of ultracold atoms},
  journal = {Nature},
  year    = {2002},
  volume  = {415},
  number  = {6867},
  pages   = {39--44},
  doi     = {10.1038/415039a},
  url     = {https://doi.org/10.1038/415039a}
}

@article{Bose_hubbard2,
  title = {Boson localization and the superfluid-insulator transition},
  author = {Fisher, Matthew P. A. and Weichman, Peter B. and Grinstein, G. and Fisher, Daniel S.},
  journal = {Phys. Rev. B},
  volume = {40},
  issue = {1},
  pages = {546--570},
  numpages = {0},
  year = {1989},
  month = {Jul},
  publisher = {American Physical Society},
  doi = {10.1103/PhysRevB.40.546},
  url = {https://link.aps.org/doi/10.1103/PhysRevB.40.546}
}

@article{driven_bose_hubbard,
  title = {Steady-State Phases and Tunneling-Induced Instabilities in the Driven Dissipative Bose-Hubbard Model},
  author = {Le Boit\'e, Alexandre and Orso, Giuliano and Ciuti, Cristiano},
  journal = {Phys. Rev. Lett.},
  volume = {110},
  issue = {23},
  pages = {233601},
  numpages = {5},
  year = {2013},
  month = {Jun},
  publisher = {American Physical Society},
  doi = {10.1103/PhysRevLett.110.233601},
  url = {https://link.aps.org/doi/10.1103/PhysRevLett.110.233601}
}

@article{AlvermannFehske2011,
title = {High-order commutator-free exponential time-propagation of driven quantum systems},
journal = {Journal of Computational Physics},
volume = {230},
number = {15},
pages = {5930-5956},
year = {2011},
issn = {0021-9991},
doi = {https://doi.org/10.1016/j.jcp.2011.04.006},
url = {https://www.sciencedirect.com/science/article/pii/S0021999111002300},
author = {A. Alvermann and H. Fehske}
}

@article{Kirchmair2013,
  author  = {Gerhard Kirchmair and Brian Vlastakis and Zaki Leghtas and Simon E. Nigg and Hanhee Paik and Eran Ginossar and Mazyar Mirrahimi and Luigi Frunzio and S. M. Girvin and R. J. Schoelkopf},
  title   = {Observation of quantum state collapse and revival due to the single-photon Kerr effect},
  journal = {Nature},
  volume  = {495},
  pages   = {205--209},
  year    = {2013},
  doi     = {10.1038/nature11902},
  url     = {https://doi.org/10.1038/nature11902}
}

@article{tenpy2024,
  title        = {{Tensor network Python (TeNPy) version 1}},
  author       = {Hauschild, Johannes and Unfried, Jakob and Anand, Sajant and Andrews, Bartholomew and Bintz, Marcus and Borla, Umberto and Divic, Stefan and Drescher, Markus and Geiger, Jan and Hefel, Martin and H{\'e}mery, K{\'e}vin and Kadow, Wilhelm and Kemp, Jack and Kirchner, Nico and Liu, Vincent S. and M{\"o}ller, Gunnar and Parker, Daniel and Rader, Michael and Romen, Anton and Scalet, Samuel and Schoonderwoerd, Leon and Schulz, Maximilian and Soejima, Tomohiro and Thoma, Philipp and Wu, Yantao and Zechmann, Philip and Zweng, Ludwig and Mong, Roger S. K. and Zaletel, Michael P. and Pollmann, Frank},
  journal      = {SciPost Physics Codebases},
  pages        = {41},
  year         = {2024},
  publisher    = {SciPost},
  doi          = {10.21468/SciPostPhysCodeb.41},
  url          = {https://scipost.org/10.21468/SciPostPhysCodeb.41}
}

@article{RevModPhys.84.621,
  title = {Gaussian quantum information},
  author = {Weedbrook, Christian and Pirandola, Stefano and Garc\'{\i}a-Patr\'on, Ra\'ul and Cerf, Nicolas J. and Ralph, Timothy C. and Shapiro, Jeffrey H. and Lloyd, Seth},
  journal = {Rev. Mod. Phys.},
  volume = {84},
  issue = {2},
  pages = {621--669},
  numpages = {0},
  year = {2012},
  month = {May},
  publisher = {American Physical Society},
  doi = {10.1103/RevModPhys.84.621},
  url = {https://link.aps.org/doi/10.1103/RevModPhys.84.621}
}

@article{haegeman2016tdvp,
  title = {Unifying time evolution and optimization with matrix product states},
  author = {Haegeman, Jutho and Lubich, Christian and Oseledets, Ivan and Vandereycken, Bart and Verstraete, Frank},
  journal = {Phys. Rev. B},
  volume = {94},
  issue = {16},
  pages = {165116},
  numpages = {10},
  year = {2016},
  month = {Oct},
  publisher = {American Physical Society},
  doi = {10.1103/PhysRevB.94.165116},
  url = {https://link.aps.org/doi/10.1103/PhysRevB.94.165116}
}

@article{SCHOLLWOCK201196,
title = {The density-matrix renormalization group in the age of matrix product states},
journal = {Annals of Physics},
volume = {326},
number = {1},
pages = {96-192},
year = {2011},
note = {January 2011 Special Issue},
issn = {0003-4916},
doi = {https://doi.org/10.1016/j.aop.2010.09.012},
url = {https://www.sciencedirect.com/science/article/pii/S0003491610001752},
author = {Ulrich Schollwöck}
}

@article{A_J_Daley_2004,
doi = {10.1088/1742-5468/2004/04/P04005},
url = {https://doi.org/10.1088/1742-5468/2004/04/P04005},
year = {2004},
month = {apr},
publisher = {},
volume = {2004},
number = {04},
pages = {P04005},
author = {A J Daley and C Kollath and U Schollwöck and G Vidal},
title = {Time-dependent density-matrix renormalization-group using adaptive effective Hilbert
spaces},
journal = {Journal of Statistical Mechanics: Theory and Experiment}
}

@misc{guseynov_gitlab_placeholder,
  author       = {Nikita Guseynov},
  title        = {Code repository for ``Coherent-State Propagation: A Computational Framework for Simulating Bosonic Quantum Systems''},
  year         = {2026},
  howpublished = {\href{https://github.com/nik77kin/Coherent-State-Propagation-A-Computational-Framework-for-Simulating-Bosonic-Quantum-Systems}{GitHub repository}},
  note         = {Accessed: 2026-04-20}
}

@article{PhysRevA.86.013814,
  title = {Josephson-junction-embedded transmission-line resonators: From Kerr medium to in-line transmon},
  author = {Bourassa, J. and Beaudoin, F. and Gambetta, Jay M. and Blais, A.},
  journal = {Phys. Rev. A},
  volume = {86},
  issue = {1},
  pages = {013814},
  numpages = {13},
  year = {2012},
  month = {Jul},
  publisher = {American Physical Society},
  doi = {10.1103/PhysRevA.86.013814},
  url = {https://link.aps.org/doi/10.1103/PhysRevA.86.013814}
}

@book{Scully_Zubairy_1997,
  author    = {Scully, Marlan O. and Zubairy, M. Suhail},
  title     = {Quantum Optics},
  publisher = {Cambridge University Press},
  address   = {Cambridge},
  year      = {1997},
  month     = sep,
  isbn      = {9780521435956},
  doi       = {10.1017/CBO9780511813993},
  url       = {https://www.cambridge.org/core/books/quantum-optics/08DC53888452CBC6CDC0FD8A1A1A4DD7}
}

@article{Trotter1959,
  author  = {H. F. Trotter},
  title   = {On the Product of Semi-Groups of Operators},
  journal = {Proceedings of the American Mathematical Society},
  volume  = {10},
  number  = {4},
  pages   = {545--551},
  year    = {1959},
  doi     = {10.1090/S0002-9939-1959-0108732-6},
  url     = {https://doi.org/10.1090/S0002-9939-1959-0108732-6}
}

@article{SUZUKI1992387,
title = {General theory of higher-order decomposition of exponential operators and symplectic integrators},
journal = {Physics Letters A},
volume = {165},
number = {5},
pages = {387-395},
year = {1992},
issn = {0375-9601},
doi = {https://doi.org/10.1016/0375-9601(92)90335-J},
url = {https://www.sciencedirect.com/science/article/pii/037596019290335J},
author = {Masuo Suzuki}
}

@article{SUZUKI1990319,
title = {Fractal decomposition of exponential operators with applications to many-body theories and Monte Carlo simulations},
journal = {Physics Letters A},
volume = {146},
number = {6},
pages = {319-323},
year = {1990},
issn = {0375-9601},
doi = {https://doi.org/10.1016/0375-9601(90)90962-N},
url = {https://www.sciencedirect.com/science/article/pii/037596019090962N},
author = {Masuo Suzuki}
}

@article{driven_bose_hubbard_2,
  title = {Quantum fluids of light},
  author = {Carusotto, Iacopo and Ciuti, Cristiano},
  journal = {Rev. Mod. Phys.},
  volume = {85},
  issue = {1},
  pages = {299--366},
  numpages = {0},
  year = {2013},
  month = {Feb},
  publisher = {American Physical Society},
  doi = {10.1103/RevModPhys.85.299},
  url = {https://link.aps.org/doi/10.1103/RevModPhys.85.299}
}

@article{MIRRAHIMI2016778,
title = {Cat-qubits for quantum computation},
journal = {Comptes Rendus Physique},
volume = {17},
number = {7},
pages = {778-787},
year = {2016},
note = {Quantum microwaves / Micro-ondes quantiques},
issn = {1631-0705},
doi = {https://doi.org/10.1016/j.crhy.2016.07.011},
url = {https://www.sciencedirect.com/science/article/pii/S1631070516300627},
author = {Mazyar Mirrahimi}
}

@article{RevModPhys.77.513,
  title = {Quantum information with continuous variables},
  author = {Braunstein, Samuel L. and van Loock, Peter},
  journal = {Rev. Mod. Phys.},
  volume = {77},
  issue = {2},
  pages = {513--577},
  numpages = {0},
  year = {2005},
  month = {Jun},
  publisher = {American Physical Society},
  doi = {10.1103/RevModPhys.77.513},
  url = {https://link.aps.org/doi/10.1103/RevModPhys.77.513}
}

@article{qubic_GU_mile,
  title = {Quantum computing with continuous-variable clusters},
  author = {Gu, Mile and Weedbrook, Christian and Menicucci, Nicolas C. and Ralph, Timothy C. and van Loock, Peter},
  journal = {Phys. Rev. A},
  volume = {79},
  issue = {6},
  pages = {062318},
  numpages = {16},
  year = {2009},
  month = {Jun},
  publisher = {American Physical Society},
  doi = {10.1103/PhysRevA.79.062318},
  url = {https://link.aps.org/doi/10.1103/PhysRevA.79.062318}
}

@article{Kerr_Milburn,
  title = {Quantum and classical Liouville dynamics of the anharmonic oscillator},
  author = {Milburn, G. J.},
  journal = {Phys. Rev. A},
  volume = {33},
  issue = {1},
  pages = {674--685},
  numpages = {0},
  year = {1986},
  month = {Jan},
  publisher = {American Physical Society},
  doi = {10.1103/PhysRevA.33.674},
  url = {https://link.aps.org/doi/10.1103/PhysRevA.33.674}
}

@article{Lloyd_PhysRevLett.82.1784,
  title = {Quantum Computation over Continuous Variables},
  author = {Lloyd, Seth and Braunstein, Samuel L.},
  journal = {Phys. Rev. Lett.},
  volume = {82},
  issue = {8},
  pages = {1784--1787},
  numpages = {0},
  year = {1999},
  month = {Feb},
  publisher = {American Physical Society},
  doi = {10.1103/PhysRevLett.82.1784},
  url = {https://link.aps.org/doi/10.1103/PhysRevLett.82.1784}
}

@article{PhysRevLett.119.220502,
  title = {Universal Extensions of Restricted Classes of Quantum Operations},
  author = {Oszmaniec, Micha\l{} and Zimbor\'as, Zolt\'an},
  journal = {Phys. Rev. Lett.},
  volume = {119},
  issue = {22},
  pages = {220502},
  numpages = {6},
  year = {2017},
  month = {Nov},
  publisher = {American Physical Society},
  doi = {10.1103/PhysRevLett.119.220502},
  url = {https://link.aps.org/doi/10.1103/PhysRevLett.119.220502}
}

@book{high:ASNA2,
  author    = {Nicholas J. Higham},
  title     = {Accuracy and Stability of Numerical Algorithms},
  publisher = {Society for Industrial and Applied Mathematics},
  address   = {Philadelphia, PA, USA},
  year      = {2002},
  edition   = {Second},
  doi       = {10.1137/1.9780898718027},
  url       = {https://epubs.siam.org/doi/book/10.1137/1.9780898718027}
}

@book{nielsen2002quantum,
  author    = {Michael A. Nielsen and Isaac L. Chuang},
  title     = {Quantum Computation and Quantum Information: 10th Anniversary Edition},
  edition   = {10th Anniversary},
  publisher = {Cambridge University Press},
  address   = {Cambridge},
  year      = {2010},
  isbn      = {978-1107002173},       
  doi       = {10.1017/CBO9780511976667},  
  url       = {https://www.cambridge.org/highereducation/books/quantum-computation-and-quantum-information/01E10196D0A682A6AEFFEA52D53BE9AE}
}
\let\addcontentsline\oldaddcontentsline

\clearpage

\newpage
\appendix
\onecolumngrid

\begin{center}
    \LARGE\textbf{Appendix}
\end{center}

\tableofcontents
\section{Bosonic Continuous-Variable Quantum Computing}
\label{sec:bosonic-cv-qc}

In this paper we study classical simulation methods for quantum computations encoded in bosonic degrees of freedom. Quantum computing studies how information can be encoded in quantum-mechanical degrees of freedom and processed using operations that exploit superposition, interference, and entanglement. In the circuit model, one specifies (i) an initial quantum state, (ii) a sequence of gates represented by unitary (or more generally, quantum channel) evolutions, and (iii) a measurement whose outcome statistics define the computational output \cite{nielsen2002quantum}. While the most common setting uses two-level systems (qubits), many physical platforms naturally provide \emph{bosonic} degrees of freedom described by harmonic-oscillator-like Hilbert spaces. These are often referred to as \emph{continuous-variable} (CV) systems because their canonical observables take values on a continuum.

A single bosonic mode (qumode) is the CV analogue of a qubit: it is one independent harmonic degree of freedom, such as a spatial or frequency mode of the electromagnetic field, a superconducting resonator mode, or a motional mode in trapped ions \cite{PhysRevA.110.012607}. We write $m$ for the number of modes. The Hilbert space of one mode is infinite-dimensional and admits the Fock (number) basis $\{\ket{n}\}_{n\in\mathbb{N}_0}$, where $\ket{0}$ is the vacuum and $\ket{n}$ has exactly $n$ excitations. The ladder operators $\hat{a}$ and $\hat{a}^\dagger$ satisfy the canonical commutation relation
\begin{equation}
[\hat{a},\hat{a}^\dagger]=\mathbb{I},
\end{equation}
and the number operator is
\begin{equation}
\hat{n}=\hat{a}^\dagger \hat{a}.
\end{equation}

For CV systems it is often convenient to work with the (dimensionless) position- and momentum-like quadratures, denoted $\hat{X}$ and $\hat{P}$, which can be defined from the ladder operators as
\begin{equation}
\hat{X}=\frac{\hat{a}+\hat{a}^\dagger}{\sqrt{2}},
\qquad
\hat{P}=-i\frac{\hat{a}-\hat{a}^\dagger}{\sqrt{2}}.
\end{equation}
With the convention $\hbar=1$, these obey
\begin{equation}
[\hat{X},\hat{P}]=i\,\mathbb{I}.
\end{equation}
For an $m$-mode system we write $\hat{X}_j,\hat{P}_j$ for mode $j\in\{1,\dots,m\}$, and the total energy (photon-number) operator is typically expressed as a sum of local energies. With the same convention, one convenient form is
\begin{equation}
\hat{N}=\sum_{j=1}^m \hat{N}_j
=
\sum_{j=1}^m \frac{\hat{X}_j^{\,2}+\hat{P}_j^{\,2}}{2}.
\end{equation}

Two complementary state representations appear repeatedly in bosonic quantum computing: the Fock representation and the coherent-state (phase-space) representation. Coherent states are defined as displaced vacuum. The single-mode displacement operator can be written in terms of $\hat{a},\hat{a}^\dagger$ as
\begin{equation}
\hat{D}(\alpha)=\exp\!\left(\alpha \hat{a}^\dagger-\alpha^* \hat{a}\right),
\end{equation}
where $\alpha\in\mathbb{C}$. The corresponding coherent state is
\begin{equation}
\ket{\alpha}=\hat{D}(\alpha)\ket{0}.
\end{equation}
Equivalently, one can parameterize displacements by real phase-space coordinates $(q,p)\in\mathbb{R}^2$ via
\begin{equation}
\hat{D}(q,p)=\exp\!\left(i p\,\hat{X}-i q\,\hat{P}\right),
\end{equation}
which is related to the complex amplitude by
\begin{equation}
\alpha=\frac{q+i p}{2}.
\end{equation}

\medskip

\textbf{Gaussian states} may be viewed as the bosonic analogue of classical multivariate normal distributions: in phase space they exhibit Gaussian structure, and their higher-order correlations are fixed by the one- and two-point functions. They include displaced and squeezed states in the single-mode setting and their multimode generalizations, and they also arise as thermal (Gibbs) states of Hamiltonians that are at most quadratic in bosonic creation and annihilation operators \cite{10.21468/SciPostPhysCore.4.3.025}.

A convenient characterization uses the Wigner quasiprobability distribution. For an $m$-mode state $\rho$, with coordinate eigenkets $|\vec{x}\rangle$ and phase-space variables $(\vec{x},\vec{p})\in\mathbb{R}^{2m}$, one definition is
\begin{equation}\label{eq:def-gaussian_state_wigner}
W(\vec{x},\vec{p})
=
\frac{1}{(2\pi)^m}
\int_{\mathbb{R}^m} d^m\vec{y}\;
\exp\!\left(i\vec{p}^{\,T}\vec{y}\right)
\left\langle \vec{x}-\frac{\vec{y}}{2}\right|\rho\left|\vec{x}+\frac{\vec{y}}{2}\right\rangle .
\end{equation}
A bosonic state is called Gaussian if $W(\vec{x},\vec{p})$ is a (classical) multivariate Gaussian function of $(\vec{x},\vec{p})$ \cite{10.21468/SciPostPhysCore.4.3.025}.

Gaussian states are therefore completely specified by their first and second moments: a mean (displacement) vector $\vec{\mu}\in\mathbb{R}^{2m}$ and a covariance matrix $\Sigma\in\mathbb{R}^{2m\times 2m}$ \cite{Gaussian_States_in_Quantum_Information,RevModPhys.84.621}. We use the ordered quadrature vector
\begin{equation}
\hat{\Gamma} = (\hat{X}_1,\ldots,\hat{X}_m,\hat{P}_1,\ldots,\hat{P}_m)^T,
\end{equation}
so that the mean and covariance elements are defined directly from $\hat{X}$ and $\hat{P}$ by
\begin{equation}
\mu_i = \langle \hat{\Gamma}_i\rangle,
\qquad
\Sigma_{ij}
=
\frac{1}{2}\left\langle \Delta\hat{\Gamma}_i\,\Delta\hat{\Gamma}_j+\Delta\hat{\Gamma}_j\,\Delta\hat{\Gamma}_i \right\rangle,
\qquad
\Delta\hat{\Gamma}_i=\hat{\Gamma}_i-\mu_i.
\end{equation}
With this ordering we decompose $\Sigma$ into $m\times m$ blocks,
\begin{equation}
\Sigma=
\begin{pmatrix}
\Sigma_{XX} & \Sigma_{XP}\\
\Sigma_{PX} & \Sigma_{PP}
\end{pmatrix},
\qquad
\Sigma_{PX}=\Sigma_{XP}^T,
\end{equation}
and similarly split $\vec{\mu}=(\vec{\mu}_X,\vec{\mu}_P)$ with $\vec{\mu}_X,\vec{\mu}_P\in\mathbb{R}^m$.

In the coordinate (position-quadrature) basis $|\vec{x}\rangle$ with $\vec{x}\in\mathbb{R}^m$, the wavefunction of a pure Gaussian state has the form
\begin{equation}
\langle \vec{x}\,|\,\Sigma,\vec{\mu}\rangle
=
\left(\frac{1}{2\pi}\right)^{m/4}
\det(\Sigma_{XX})^{-1/4}
\exp\!\left[
-(\vec{x}-\vec{\mu}_X)^T
\left(
\frac{1}{4}\Sigma_{XX}^{-1}
+\frac{i}{2}\Sigma_{XP}\Sigma_{XX}^{-1}
\right)
(\vec{x}-\vec{\mu}_X)
\right],\label{eq:gaussian state in coordinate}
\end{equation}
where any global phase can be absorbed into the overall normalization and does not affect observables.

\medskip

\textbf{Gaussian gates} are exactly the operations that map Gaussian states to Gaussian states; equivalently, they are the unitaries generated by Hamiltonians that are at most quadratic in the canonical operators \cite{RevModPhys.84.621}.

In terms of the quadratures $\hat{X}_j,\hat{P}_j$, a general Gaussian Hamiltonian can be written as
\begin{equation}
\hat{H}_G
=
c
+
\sum_{j=1}^m \left(\xi^{(X)}_j \hat{X}_j + \xi^{(P)}_j \hat{P}_j\right)
+
\frac{1}{2}\sum_{j,k=1}^m
\left(
A^{(XX)}_{jk}\hat{X}_j\hat{X}_k
+
A^{(PP)}_{jk}\hat{P}_j\hat{P}_k
+
A^{(XP)}_{jk}\frac{\hat{X}_j\hat{P}_k+\hat{P}_k\hat{X}_j}{2}
\right),
\end{equation}
where $c\in\mathbb{R}$, the vectors $\xi^{(X)},\xi^{(P)}\in\mathbb{R}^m$ describe linear drives (displacements), and the quadratic terms include mode mixing and squeezing.

Equivalently, using bosonic ladder operators $\hat{a}_j,\hat{a}_j^\dagger$, Gaussian Hamiltonians are precisely those containing only terms up to second order,
\begin{equation}
\hat{H}_G
=
c
+
\sum_{j=1}^m \left(\eta_j \hat{a}_j^\dagger + \eta_j^*\hat{a}_j\right)
+
\sum_{j,k=1}^m \left(
h_{jk}\hat{a}_j^\dagger \hat{a}_k
+
\frac{1}{2}\kappa_{jk}\hat{a}_j^\dagger \hat{a}_k^\dagger
+
\frac{1}{2}\kappa_{jk}^*\hat{a}_j \hat{a}_k
\right),
\end{equation}
with $h=h^\dagger$ and $\kappa=\kappa^T$.
The linear terms generate displacements, the number-preserving bilinear terms generate passive linear optics, and the pair-creation/annihilation terms generate squeezing; in all cases the resulting unitary evolution preserves Gaussianity \cite{RevModPhys.77.513}.

\medskip

\textbf{Non-Gaussianity} is a necessary ingredient for universal continuous-variable (CV) quantum computation. Gaussian states under Gaussian gates—i.e., evolutions generated by Hamiltonians at most quadratic in the canonical operators—remain Gaussian and are fully characterized by their first and second moments \eqref{eq:gaussian state in coordinate}. Thus, Gaussian CV circuits are not only non-universal \cite{Lloyd_PhysRevLett.82.1784}, but also efficiently classically simulable in the standard Gaussian setting, so one does not expect quantum computational advantage from such circuits alone \cite{PhysRevLett.88.097904}. A standard route to universality is therefore to supplement Gaussian operations with at least one genuinely non-Gaussian gate generated by a higher-than-quadratic Hamiltonian.

Two widely used single-mode non-Gaussian gates are the Kerr gate and the cubic phase gate. The Kerr gate is generated by the quartic Hamiltonian $\hat n^2$ (with $\hat n=\hat a^\dagger \hat a$), and is defined as \cite{Kerr_Milburn}
\begin{equation}\label{eq:kerr_gate}
\hat K_\kappa := \exp\!\left(i\kappa \hat n^{2}\right),
\end{equation}
where $\kappa\in\mathbb{R}$ is the interaction strength. The cubic phase gate is generated by $\hat X^3$ and is defined as \cite{qubic_GU_mile}
\begin{equation}
\hat V^{(3)}(\theta) := \exp\!\left(i\theta \hat X^{3}\right),
\end{equation}
where $\theta\in\mathbb{R}$ and $\hat X$ denotes the position quadrature.

In this work we focus on circuits augmented with the Kerr gate. A key consideration is the resulting energy growth, since it controls both the effective truncation size in classical simulation and the physical resources required for implementation. Here the Kerr gate differs qualitatively from the cubic phase gate: the Kerr generator $\hat{n}^2$ commutes with the number operator, so it does not by itself change photon number, whereas the cubic generator $\hat{X}^3$ does not commute with $\hat{n}$ and can directly drive support to higher-photon-number sectors. Rapid energy growth therefore depends on the surrounding Gaussian layers, but when it occurs it forces large Hilbert-space cutoffs and broad Fock support, amplifying loss, dephasing, and control errors \cite{MIRRAHIMI2016778,PhysRevX.15.011070}. From the algorithmic side, it also makes methods based on finite Hilbert-space truncation or discretization, including tensor-network approaches, increasingly unsuitable.

For the cubic phase gate, there exist explicit alternating Gaussian/non-Gaussian circuits in which the energy grows doubly exponentially with the number of alternations. Let $\hat F$ denote the Fourier transform (a Gaussian unitary), and define
\begin{equation}
\hat U_t := \bigl(\hat F \hat V^{(3)}(\theta)\bigr)^{t},
\qquad
|\psi_t\rangle := \hat U_t |0\rangle .
\end{equation}

\begin{proposition}[Doubly exponential energy growth for Gaussian$+\hat V^{(3)}$ circuits (Lemma~3.4 in \cite{chabaud2025energybosonscomputationalcomplexity})]
For $\hat U_t$ and $|\psi_t\rangle$ defined above, the expected energy satisfies
\begin{equation}
\langle \psi_t | \hat N | \psi_t \rangle
\ge
\frac{1}{4e}\left(\frac{\theta^{2}}{2e}\right)^{2^{\,t-1}} 2^{\,t2^{t}}
-\frac{1}{2}.
\end{equation}
\end{proposition}

By contrast, Kerr-augmented Gaussian circuits admit an exponential upper bound on the energy cost.

\begin{proposition}[Exponential upper bound for Gaussian$+\hat n^2$ circuits (Proposition~3.1 in \cite{chabaud2025energybosonscomputationalcomplexity})]
Let $\hat U$ be a single-mode circuit of size at most $T$ composed of Gaussian gates and Kerr gates generated by $\hat n^2$ (with each gate specified with constant bits of precision). Then there exists a constant $c>0$ such that
\begin{equation}
\langle 0 | \hat U^\dagger \hat N \hat U |0\rangle \le e^{cT}.
\end{equation}
\end{proposition}

\medskip

\textbf{Universality} in the CV setting can be formulated in terms of the Lie closure of implementable Hamiltonians \cite{Lloyd_PhysRevLett.82.1784}. Working with the canonical operators $\hat X$ and $\hat P$ (with $[\hat X,\hat P]=i \mathbb{I}$), commutators allow one to generate effective evolutions under new Hamiltonians from old ones. A basic tool is the group commutator: for sufficiently small $\tau$,
\begin{equation}
e^{iA\tau}e^{iB\tau}e^{-iA\tau}e^{-iB\tau}
=\exp\!\big(-[A,B]\tau^2+O(\tau^3)\big),
\end{equation}
so nested commutators can be synthesized by concatenating such short-time sequences (together with standard product-formula arguments).

In this work we consider displaced linear optics (displacements and passive linear optics) augmented by a single-mode Kerr nonlinearity. In the energy-cutoff sense used throughout, this gate set is universal.

\begin{proposition}
\label{prop:universal-kerr-dlo-main}
Fix an energy cutoff $n_{\max}$ and let $\Pi_{n_{\max}}$ be the projector onto the multimode Fock subspace of total
photon number at most $n_{\max}$. The gate set consisting of displacements, passive linear optics, and a single-mode
Kerr gate is universal on $\Pi_{n_{\max}}\mathcal{H}$, i.e., for any unitary $U$ acting on $\Pi_{n_{\max}}\mathcal{H}$
and any $\varepsilon>0$, there exists a finite circuit $V$ from this gate set such that
\begin{equation}
\|\Pi_{n_{\max}}(U-V)\Pi_{n_{\max}}\|\le \varepsilon.
\end{equation}
\end{proposition}

\begin{proof}
See Appendix~\ref{sec:cv-universality-kerr}.
\end{proof}

\section{Coherent-State Propagation for Simulating Bosonic Circuits}
\label{sec:coherent-state-calculus}

We consider $m$-mode bosonic circuits built from repeated layers that combine Gaussian operations with a Kerr-type nonlinearity. Concretely, each layer is composed of two conceptually distinct parts: first, a displaced passive linear-optics (DPLO) sublayer, generated by Eq.~\eqref{eq:H_DPLO}, and second, a Kerr sublayer built from single-mode Kerr gates of the form in Eq.~\eqref{eq:kerr_gate}. For schematic convenience, we hereafter apply the Kerr action as occurring on the first mode only. This is not a restriction: in addition to the Hamiltonian evolution generated by Eq.~\eqref{eq:H_DPLO}, the DPLO block may also include arbitrary mode-swap operations, so any physical mode can be routed to the first wire, acted on by the Kerr gate. The conceptual circuit shown in Fig.~\ref{fig:circuit_scheme} should be understood as a general multimode architecture rather than a model in which one distinguished mode is intrinsically special. In this sense, together with Appendix~\ref{sec:cv-universality-kerr}, the circuit architecture in Eq.~\ref{eq:citrcuit_architecture_intro} already captures a universal CV circuit structure.

Our goal is to simulate the resulting non-Gaussian dynamics in the Schr\"odinger picture by tracking an explicit representation of the evolving pure state. Because the circuit may end with an arbitrary Gaussian unitary, it already contains purely Gaussian CV computations as a special case; once $\kappa \neq 0$, the Kerr gate adds a genuinely non-Gaussian resource. Hence the model is strictly more expressive than Gaussian CV computing alone.

Specifically, we propose a forward-propagation algorithm in which the state at each step is represented as a finite superposition of multimode product coherent states
\begin{equation}\label{eq:cs-ansatz}
|\psi\rangle=\sum_{k=1}^{N} C_k|\alpha^{(1)}_k\rangle\otimes|\alpha^{(2)}_k\rangle\otimes\cdots\otimes|\alpha^{(m)}_k\rangle,
\end{equation}
where $N$ is the truncation parameter (so the expansion contains $N$ coherent-product terms), the superscript $(j)$ labels the mode index $j\in\{1,\dots,m\}$, and the coefficients $C_k\in\mathbb{C}$ are complex scalars.

It is convenient to collect the parameters into a coefficient vector and a coherent-amplitude matrix. We define
\begin{equation}\label{eq:state_coefs}
\bm C := (C_1,C_2,\dots,C_N)^T\in\mathbb{C}^{N},
\end{equation}
and
\begin{equation}\label{eq:state_matrix}
A :=
\begin{pmatrix}
\alpha^{(1)}_1 & \alpha^{(2)}_1 & \cdots & \alpha^{(m)}_1\\
\alpha^{(1)}_2 & \alpha^{(2)}_2 & \cdots & \alpha^{(m)}_2\\
\vdots & \vdots & \ddots & \vdots\\
\alpha^{(1)}_N & \alpha^{(2)}_N & \cdots & \alpha^{(m)}_N
\end{pmatrix}
\in\mathbb{C}^{(N)\times m},
\end{equation}
so that row $k$ stores the multimode coherent amplitude vector $(\alpha^{(1)}_k,\dots,\alpha^{(m)}_k)$ associated with the $k$-th term in the superposition. In this representation, the current state of the simulation is fully specified by the pair $(\bm C,A)$.

An emerging approach for representing non-Gaussian states is to approximate them by finite linear combinations of Gaussian states, often described as a \emph{superposition of Gaussian states} \cite{hahn2025classical}. In this framework, one works directly in the Schr\"odinger picture with an explicit state description that keeps track of a small number of Gaussian components and their complex weights, enabling a compact description of states beyond the Gaussian sector.

A related Schr\"odinger-picture propagation strategy is developed in Ref.~\cite{upreti2025boundingcomputationalpowerbosonic} for circuits that include cubic phase gates. There, the authors propose to represent the evolving pure state as a finite superposition of phase-space states, and to update this representation forward through alternating Gaussian and non-Gaussian layers. In particular, their circuit architecture interleaves Gaussian transformations with cubic non-Gaussian gates, and the state is maintained throughout as an explicit finite sum of coherent-state-like components. Our starting point in Eq.~\eqref{eq:cs-ansatz} follows the same high-level philosophy, specialized to Kerr-augmented layers.

In the remainder of this section, we describe how the state $|\psi\rangle$ propagates through the Kerr sublayer and the DPLO sublayer, and how the representation $(\bm C,A)$ is updated under each type of operation.

\subsection{Displaced Passive Linear Optics}
\label{subsec:displaced-linear-optics}

In this subsection we consider a DPLO layer, i.e., passive linear optics together with a linear displacement drive. We
show that if the input state is written in the coherent-product form of Eq.~\eqref{eq:cs-ansatz}, then after a DPLO
step it can be written in the same form. Moreover, the number $N$ of coherent-product terms does not increase: DPLO
updates each term independently, without generating additional components.

We fix the DPLO Hamiltonian
\begin{equation}
\label{eq:H_DPLO}
H_{\rm DPLO}=\sum_{j,k=1}^{m}\hat a_j^\dagger K_{jk}\hat a_k+\sum_{j=1}^{m}\big(\eta_j \hat a_j^\dagger+\eta_j^*\hat a_j\big),
\end{equation}
where the first term is passive linear optics (mode mixing) with $K=K^\dagger\in\mathbb{C}^{m\times m}$, and the
second term is a displacement drive with $\bm{\eta}=(\eta_1,\dots,\eta_m)^{\mathsf T}\in\mathbb{C}^m$. The corresponding
unitary is
\begin{equation}
\label{eq:U_DPLO}
U_{\rm DPLO}(t):=\exp\!\big(-it H_{\rm DPLO}\big).
\end{equation}

\begin{proposition}[DPLO does not increase $N$]
\label{prop:DPLO_preserves_N}
Fix $t>0$, a Hermitian matrix $K\in\mathbb{C}^{m\times m}$, and $\bm{\eta}\in\mathbb{C}^m$, and let $U_{\rm DPLO}(t)$ be
defined by Eqs.~\eqref{eq:H_DPLO}--\eqref{eq:U_DPLO}. If $|\psi\rangle$ is written as in Eq.~\eqref{eq:cs-ansatz} with
$N$ coherent-product terms, then $U_{\rm DPLO}(t)|\psi\rangle$ admits the same coherent-product structure with the same
number $N$ of terms.
\end{proposition}

\begin{proof}
Write $\bm{\hat a}=(\hat a_1,\dots,\hat a_m)^{\mathsf T}$. The Heisenberg equation generated by $H_{\rm DPLO}$ is the
forced linear system
\begin{equation}
\label{eq:DPLO_heis}
\frac{d}{ds}\bm{\hat a}(s)=i[H_{\rm DPLO},\bm{\hat a}(s)]=-iK\,\bm{\hat a}(s)-i\bm{\eta}.
\end{equation}
Solving it gives an affine Heisenberg map
\begin{equation}
\label{eq:DPLO_affine}
U_{\rm DPLO}^\dagger(t)\,\bm{\hat a}\,U_{\rm DPLO}(t)=S(t)\bm{\hat a}+\bm{\gamma}(t),
\qquad
S(t)=e^{-iKt},
\qquad
\bm{\gamma}(t)=-i\int_{0}^{t} e^{-iK(t-s)}\,\bm{\eta}\,ds.
\end{equation}

Now fix a product coherent state
\begin{equation}
\ket{\bm{\alpha}}:=\ket{\alpha^{(1)}}\otimes\cdots\otimes\ket{\alpha^{(m)}},
\qquad
\bm{\alpha}=(\alpha^{(1)},\dots,\alpha^{(m)})^{\mathsf T},
\end{equation}
so that $\bm{\hat a}\ket{\bm{\alpha}}=\bm{\alpha}\ket{\bm{\alpha}}$. Using Eq.~\eqref{eq:DPLO_affine},
\begin{equation}
\bm{\hat a}\,U_{\rm DPLO}(t)\ket{\bm{\alpha}}
=
U_{\rm DPLO}(t)\,\bigl(S(t)\bm{\hat a}+\bm{\gamma}(t)\bigr)\ket{\bm{\alpha}}
=
\bigl(S(t)\bm{\alpha}+\bm{\gamma}(t)\bigr)\,U_{\rm DPLO}(t)\ket{\bm{\alpha}}.
\end{equation}
Thus $U_{\rm DPLO}(t)\ket{\bm{\alpha}}$ is a joint eigenvector of $\bm{\hat a}$ with eigenvalue
$\bm{\alpha}'=S(t)\bm{\alpha}+\bm{\gamma}(t)$, hence
\begin{equation}
\label{eq:DPLO_coherent_map}
U_{\rm DPLO}(t)\ket{\bm{\alpha}}=e^{i\phi(\bm{\alpha};t)}\ket{\bm{\alpha}'}
\end{equation}
for some phase $\phi(\bm{\alpha};t)$. Writing $\bm{\beta}:=S(t)\bm{\alpha}$, the $\bm{\alpha}$-dependent part of this
phase can be taken in the standard displacement form,
\begin{equation}
\label{eq:DPLO_phase}
e^{i\phi(\bm{\alpha};t)}
=
e^{i\chi(t)}\,
\exp\!\Big(\frac{1}{2}\big(\bm{\gamma}(t)\!\cdot\!\bm{\beta}^*-\bm{\gamma}(t)^*\!\cdot\!\bm{\beta}\big)\Big),
\end{equation}
where $\chi(t)$ is independent of $\bm{\alpha}$ (and therefore contributes only a global phase when acting on a
superposition).

Finally, apply Eq.~\eqref{eq:DPLO_coherent_map} to each coherent-product term in Eq.~\eqref{eq:cs-ansatz}. Each input
term produces exactly one output coherent-product term, with updated amplitudes
$\bm{\alpha}_k' = S(t)\bm{\alpha}_k+\bm{\gamma}(t)$ and updated coefficient obtained by multiplying $C_k$ by the
$k$-dependent phase factor in Eq.~\eqref{eq:DPLO_phase} (the common factor $e^{i\chi(t)}$ is a global phase). Hence the number
of terms $N$ is unchanged.
\end{proof}

\subsection{Applying the Kerr gate via the Fock basis and a coherent-grid back-transform}
\label{sec:kerr_via_fock}

Now we consider how the superposition of coherent products propagates through the Kerr gate $K_\kappa=\exp(i\kappa \hat n^2)$. Let us first consider the fundamental block of the state \eqref{eq:cs-ansatz}; the Kerr gate $K_\kappa=\exp(i\kappa \hat n^2)$ is diagonal in the Fock basis, so for a coherent input $\ket{\alpha}$ we apply it by
(i) expanding $\ket{\alpha}$ in the Fock basis, (ii) multiplying each Fock amplitude by the Kerr phase $e^{i\kappa n^2}$,
(iii) truncating to $\mathrm{span}\{\ket{0},\dots,\ket{N}\}$ with controlled error (Lemma~\ref{thm:kerr-coherent-truncation}), and
(iv) re-expressing the truncated Kerr state as a finite coherent-state superposition on a simple $(N{+}1)$-point grid
(Lemma~\ref{lem:fock-to-coherent-switch}). The workflow is
\begin{equation}\label{eq:kerr-schematic-coherent}
\ket{\alpha}
\xrightarrow{\ \scriptsize\shortstack{\text{Fock}\\\text{exp.}}\! }
e^{-\frac{|\alpha|^2}{2}}\!\sum_{n=0}^{\infty}\frac{\alpha^n}{\sqrt{n!}}\ket{n}
\xrightarrow{K_\kappa}
e^{-\frac{|\alpha|^2}{2}}\!\sum_{n=0}^{\infty}\frac{e^{i\kappa n^2}\alpha^n}{\sqrt{n!}}\ket{n}
\xrightarrow[\scriptsize \Pi_{\le N_F}\ ]{\ \text{\scriptsize Lemma~\ref{thm:kerr-coherent-truncation}}\! }
\mathcal{N}\!\sum_{n=0}^{N_F} \frac{e^{i\kappa n^2}\alpha^n}{\sqrt{n!}}\ket{n}
\xrightarrow[\scriptsize\shortstack{\text{Coherent}\\\text{basis}}]{\ \text{\scriptsize Lemma~\ref{lem:fock-to-coherent-switch}}\! }
\sum_{k=0}^{N_F} c_k \ket{\beta_k},
\end{equation}
where $\mathcal{N}$ denotes the renormalization constant after the projection to the Hilbert state with restricted photon number $\Pi_{\le N_F}$, i.e.\ it ensures that the
truncated Fock superposition has unit norm after discarding number states $\ket{n}$ with $n>N_F$.
Theorem~\ref{thm:kerr-coherent-expansion} formalizes this procedure: the sequence of maps in
\eqref{eq:kerr-schematic-coherent} is exactly the construction underlying the theorem, producing an $(N_F{+}1)$-term coherent-state
approximation of $K_\kappa\ket{\alpha}$ by first controlling the Fock truncation error
(Lemma~\ref{thm:kerr-coherent-truncation}) and then converting the truncated Kerr state back to a coherent-state expansion on the
grid $\{\beta_k\}_{k=0}^{N_F}$ (Lemma~\ref{lem:fock-to-coherent-switch}).

\begin{theorem}[Kerr acting on a coherent state admits a finite coherent-state expansion]
\label{thm:kerr-coherent-expansion}
Let $\kappa>0$ and let $K_\kappa$ be the Kerr gate defined in \eqref{eq:kerr_gate}. For any coherent state $\ket{\alpha}$ with
$\alpha\in\mathbb{C}$, define $\ket{\psi_{\alpha,\kappa}}:=K_\kappa\ket{\alpha}$.
For any integer $N_F\ge 0$ and any $\tilde\epsilon>0$, define
\begin{equation}
\label{eq:kerr-grid}
\omega:=e^{2\pi i/(N_F+1)},\qquad \beta_k:=\tilde\epsilon\,\omega^k,\qquad k=0,\dots,N_F,
\end{equation}
and coefficients
\begin{equation}
\label{eq:kerr-coeff-final}
c_k
:=
\frac{e^{\tilde\epsilon^2/2}}{N_F+1}\cdot
\frac{1}{\sqrt{\sum_{m=0}^{N_F}\frac{|\alpha|^{2m}}{m!}}}
\sum_{n=0}^{N_F} e^{i\kappa n^2}\left(\frac{\alpha}{\tilde\epsilon}\right)^{\!n}\omega^{-kn},
\qquad k=0,\dots,N_F.
\end{equation}
Let
\begin{equation}
\label{eq:phi-final}
\ket{\phi_{\alpha,\kappa}^{(N_F)}}:=
\frac{\sum_{k=0}^{N_F} c_k \ket{\beta_k}}
{\left\|\sum_{k=0}^{N_F} c_k \ket{\beta_k}\right\|_2}.
\end{equation}
Then $\ket{\phi_{\alpha,\kappa}^{(N_F)}}$ is an $(N_F{+}1)$-term coherent-state approximation of $K_\kappa\ket{\alpha}$ with error
\begin{equation}
\label{eq:error-final-scaling}
\big\|\ket{\psi_{\alpha,\kappa}}-\ket{\phi_{\alpha,\kappa}^{(N_F)}}\big\|_2
=
\mathcal{O}\!\left(
\exp\!\Big(-\Theta\!\big(N_F-|\alpha|^2\big)\Big)
\;+\;
\frac{\tilde\epsilon^{\,N_F+1}}{\sqrt{(N_F+1)!}}
\right),
\end{equation}
where the first term is the (inverted) Fock-truncation error from Lemma~\ref{thm:kerr-coherent-truncation} and the second term
is the coherent-grid leakage from Lemma~\ref{lem:fock-to-coherent-switch}.
\end{theorem}
\begin{proof}
By the triangle inequality,
\begin{equation}\label{eq:tri-kerr}
\big\|\ket{\psi_{\alpha,\kappa}}-\ket{\phi_{\alpha,\kappa}^{(N_F)}}\big\|_2
\le
\big\|\ket{\psi_{\alpha,\kappa}}-\ket{\psi_{\alpha,\kappa}^{(\le N_F)}}\big\|_2
+
\big\|\ket{\psi_{\alpha,\kappa}^{(\le N_F)}}-\ket{\phi_{\alpha,\kappa}^{(N_F)}}\big\|_2 .
\end{equation}
The first term is the Fock-truncation error bounded in Lemma~\ref{thm:kerr-coherent-truncation}.
The second term is the coherent-grid conversion error bounded in Corollary~\ref{cor:kerr-trunc-coherent-grid}, which follows
from Lemma~\ref{lem:fock-to-coherent-switch}. Combining these bounds yields \eqref{eq:error-final-scaling}.
\end{proof}

\begin{lemma}[Fock truncation error for a Kerr-evolved coherent state]
\label{thm:kerr-coherent-truncation}
Let $K_\kappa := \exp(i\kappa \hat n^2)$ be the single-mode Kerr unitary, where $\hat n := \hat a^\dagger \hat a$.
Fix $\alpha\in\mathbb{C}$ and define the Kerr-evolved coherent state $\ket{\psi_{\alpha,\kappa}} := K_\kappa\ket{\alpha}$.
For an integer $N_F\ge 0$, let $\Pi_{\le N_F}:=\sum_{n=0}^{N_F} \ket{n}\!\bra{n}$ and define the normalized truncation
\begin{equation}
\label{eq:psi-trunc}
\ket{\psi_{\alpha,\kappa}^{(\le N_F)}} :=
\frac{\Pi_{\le N_F}\ket{\psi_{\alpha,\kappa}}}{\sqrt{\bra{\psi_{\alpha,\kappa}}\Pi_{\le N_F}\ket{\psi_{\alpha,\kappa}}}}.
\end{equation}
Writing $\lambda:=|\alpha|^2$, the fidelity satisfies the exact identity
\begin{equation}
\label{eq:fidelity-exact}
F_{N_F}
:= \big|\braket{\psi_{\alpha,\kappa}}{\psi_{\alpha,\kappa}^{(\le N_F)}}\big|^2
= \sum_{n=0}^{N_F} e^{-\lambda}\frac{\lambda^n}{n!},
\qquad
1-F_{N_F} = \sum_{n=N_F+1}^\infty e^{-\lambda}\frac{\lambda^n}{n!}.
\end{equation}
Moreover, for any $\varepsilon\in(0,1)$, if $N_F$ satisfies
\begin{equation}
\label{eq:N-choice}
N_F \;\ge\;
\lambda + \sqrt{2\lambda \ln(1/\varepsilon)} + \frac{2}{3}\ln(1/\varepsilon) - 1,
\end{equation}
then $1-F_{N_F}\le \varepsilon$, and consequently
\begin{equation}
\label{eq:trunc-norm-bound}
\big\|\ket{\psi_{\alpha,\kappa}}-\ket{\psi_{\alpha,\kappa}^{(\le N_F)}}\big\|_2
\le \sqrt{2\varepsilon}=\mathcal{O}\!\left(
\exp\!\Big(-\Theta\!\big(N_F-|\alpha|^2\big)\Big)
\right).
\end{equation}
\end{lemma}

\begin{proof}
Using $K_\kappa\ket{n}=e^{i\kappa n^2}\ket{n}$ and
$\ket{\alpha}=e^{-|\alpha|^2/2}\sum_{n\ge 0}\alpha^n/\sqrt{n!}\ket{n}$, we obtain
$|\braket{n}{\psi_{\alpha,\kappa}}|^2=|\braket{n}{\alpha}|^2=e^{-\lambda}\lambda^n/n!$ with $\lambda=|\alpha|^2$.
By definition of \eqref{eq:psi-trunc},
\begin{equation}
F_{N_F}
=
\big|\braket{\psi_{\alpha,\kappa}}{\psi_{\alpha,\kappa}^{(\le N_F)}}\big|^2
=
\bra{\psi_{\alpha,\kappa}}\Pi_{\le N_F}\ket{\psi_{\alpha,\kappa}}
=
\sum_{n=0}^{N_F}e^{-\lambda}\frac{\lambda^n}{n!},
\end{equation}
which is \eqref{eq:fidelity-exact}, and $1-F_{N_F}$ is the complementary tail.
Let $X\sim\mathrm{Pois}(\lambda)$ so that $1-F_{N_F}=\Pr[X\ge N_F+1]$. A Chernoff bound gives
$\Pr[X\ge N_F+1]\le \varepsilon$ under the condition \eqref{eq:N-choice}, hence $1-F_{N_F}\le \varepsilon$.
Finally, for pure states $\|\ket{x}-\ket{y}\|_2^2=2(1-\Re\langle x|y\rangle)\le 2(1-|\langle x|y\rangle|^2)$, so
\begin{equation}
\big\|\ket{\psi_{\alpha,\kappa}}-\ket{\psi_{\alpha,\kappa}^{(\le N_F)}}\big\|_2
\le \sqrt{2\big(1-F_{N_F}\big)} \le \sqrt{2\varepsilon},
\end{equation}
which is \eqref{eq:trunc-norm-bound}.
\end{proof}

\begin{lemma}[Fock-to-coherent switch on a $(N_F{+}1)$-point coherent grid (Theorem $1$ from \cite{Marshall:23})]
\label{lem:fock-to-coherent-switch}
Let $\ket{\psi}=\sum_{n=0}^{N_F}a_n\ket{n}$ with $\sum_{n=0}^{N_F}|a_n|^2=1$ be a decomposition of state $\ket{\psi}$ in Fock basis.
Fix $\tilde\epsilon>0$ and define the uniform coherent grid
\begin{equation}
\label{eq:grid-def}
\omega:=e^{2\pi i/(N_F+1)},\qquad
\beta_k:=\tilde\epsilon\,\omega^k,\qquad k=0,\dots,N_F .
\end{equation}
Define the coefficients
\begin{equation}
\label{eq:ck-def-compact}
c_k
:=
\frac{e^{\tilde\epsilon^2/2}}{N_F+1}\sum_{n=0}^{N_F}\sqrt{n!}\,
\frac{a_n}{\tilde\epsilon^{\,n}}\;\omega^{-kn},
\qquad k=0,\dots,N_F,
\end{equation}
and the associated coherent superposition
\begin{equation}
\label{eq:psi-tilde-and-phi}
\ket{\widetilde\psi}:=\sum_{k=0}^{N_F}c_k\ket{\beta_k},
\qquad
\ket{\phi}:=\frac{\ket{\widetilde\psi}}{\|\ket{\widetilde\psi}\|_2}.
\end{equation}
Then $\Pi_{\le N_F}\ket{\widetilde\psi}=\ket{\psi}$, where $\Pi_{\le N_F}:=\sum_{n=0}^{N_F}\ket{n}\!\bra{n}$.
Moreover, letting $\ket{\chi}:=(I-\Pi_{\le N_F})\ket{\widetilde\psi}$, Theorem~1 of \cite{Marshall:23} implies
\begin{equation}
\label{eq:leakage-scaling}
\|\ket{\chi}\|_2 \le
\delta_{N_F}(\tilde\epsilon)
:=\mathcal{O}\!\left(\frac{\tilde\epsilon^{\,N_F+1}}{\sqrt{(N_F+1)!}}\right),
\end{equation}
and consequently (normalization only changes the state by the leakage scale)
\begin{equation}
\label{eq:approx-scaling}
\|\ket{\phi}-\ket{\psi}\|_2 \le 2\,\delta_{N_F}(\tilde\epsilon).
\end{equation}
\end{lemma}

We now extract the cutoff scaling used in the main text,
see Eq.~\eqref{eq:NF_gap_scaling_main}. For fixed finite grid radius
$\tilde\epsilon$, independent of $N_F$ and of the target error, the second term
in Eq.~\eqref{eq:error-final-scaling} decays faster than exponentially in
$N_F$ by Stirling's approximation. Hence the cutoff scaling is governed by the
Poisson-tail term. Requiring
\begin{equation}
\exp\!\Big(-\Theta\!\big(N_F-|\alpha|^2\big)\Big)
\lesssim
\varepsilon_{\mathrm{kerr}}
\end{equation}
gives
\begin{equation}
N_F
=
|\alpha|^2
+
\mathcal{O}\!\left(
\log\frac{1}{\varepsilon_{\mathrm{kerr}}}
\right).
\end{equation}

\begin{corollary}[Coherent-grid coefficients for the truncated Kerr-evolved coherent state]
\label{cor:kerr-trunc-coherent-grid}
Let $\ket{\psi_{\alpha,\kappa}}:=e^{i\kappa \hat n^2}\ket{\alpha}$ and let
$\ket{\psi_{\alpha,\kappa}^{(\le N_F)}}=\sum_{n=0}^{N_F}a_n\ket{n}$
be the normalized Fock truncation from Lemma~\ref{thm:kerr-coherent-truncation}. Then
\begin{equation}
\label{eq:an-kerr-short}
a_n=\frac{e^{i\kappa n^2}\,\alpha^n/\sqrt{n!}}{\sqrt{\sum_{m=0}^{N_F}\frac{|\alpha|^{2m}}{m!}}}.
\end{equation}
Applying Lemma~\ref{lem:fock-to-coherent-switch} with $\omega=e^{2\pi i/(N_F+1)}$ and $\beta_k=\tilde\epsilon\,\omega^k$ gives
\begin{equation}
\label{eq:ck-kerr-short}
c_k
=
\frac{e^{\tilde\epsilon^2/2}}{N_F+1}\cdot
\frac{1}{\sqrt{\sum_{m=0}^{N_F}\frac{|\alpha|^{2m}}{m!}}}
\sum_{n=0}^{N_F} e^{i\kappa n^2}\left(\frac{\alpha}{\tilde\epsilon}\right)^{\!n}\omega^{-kn},
\qquad k=0,\dots,N_F.
\end{equation}
Hence $\ket{\widetilde\psi}:=\sum_{k=0}^{N_F}c_k\ket{\beta_k}$ satisfies
$\Pi_{\le N_F}\ket{\widetilde\psi}=\ket{\psi_{\alpha,\kappa}^{(\le N_F)}}$
and the approximation error scales as in \eqref{eq:leakage-scaling}--\eqref{eq:approx-scaling}.
\end{corollary}

The grid radius $\tilde\epsilon>0$ in Lemma~\ref{lem:fock-to-coherent-switch}
controls the leakage term $\delta_{N_F}(\tilde\epsilon)$ and also the numerical conditioning of evaluating the sum containing Discrete Fourier Transform (DFT) in
\eqref{eq:kerr-coeff-final}. When $|\alpha|/\tilde\epsilon\neq 1$, the magnitudes $|(\alpha/\tilde\epsilon)^n|$ vary exponentially
in $n$, which can cause overflow/underflow and catastrophic cancellation for large $N_F$ \cite{high:ASNA2}. The following lemma gives a principled
choice that balances these magnitudes.

\begin{lemma}[Numerically stable radius choice for Kerr--coherent expansion]
\label{lem:kerr-stable-radius}
Fix $N_F\ge 0$ and let $\omega:=e^{2\pi i/(N_F+1)}$ be the root of unity used in the coherent grid
\eqref{eq:kerr-grid}. Consider the coherent-grid coefficients $\{c_k\}_{k=0}^{N_F}$ defined in
Theorem~\ref{thm:kerr-coherent-expansion} (equivalently, in \eqref{eq:kerr-coeff-final}) as functions of the grid radius
$\tilde\epsilon>0$. Define the term-magnitude dynamic range
\begin{equation}\label{eq:dyn-range}
D(\tilde\epsilon):=\frac{\max_{0\le n\le N_F}\big|(\alpha/\tilde\epsilon)^n\big|}{\min_{0\le n\le N_F}\big|(\alpha/\tilde\epsilon)^n\big|}.
\end{equation}
Then $D(\tilde\epsilon)$ is uniquely minimized at $\tilde\epsilon=|\alpha|$, where $D(\tilde\epsilon)=1$.
For any $\tilde\epsilon\neq|\alpha|$, one has
\begin{equation}\label{eq:dyn-range-exp}
D(\tilde\epsilon)=\max\!\Big\{\Big(\frac{|\alpha|}{\tilde\epsilon}\Big)^{N_F},\Big(\frac{\tilde\epsilon}{|\alpha|}\Big)^{N_F}\Big\}>1,
\end{equation}
so the dynamic range grows exponentially in $N_F$. Consequently, the choice $\tilde\epsilon=|\alpha|$ is optimal for numerical
stability of evaluating the coefficient formula in Theorem~\ref{thm:kerr-coherent-expansion} in finite precision arithmetic
\cite{high:ASNA2}.
\end{lemma}

\begin{proof}
Let $r:=|\alpha|/\tilde\epsilon>0$. Since $\big|(\alpha/\tilde\epsilon)^n\big|=r^n$, we have
$\max_{0\le n\le N_F} r^n = \max\{1,r^{N_F}\}$ and $\min_{0\le n\le N_F} r^n = \min\{1,r^{N_F}\}$, hence
\begin{equation}
D(\tilde\epsilon)=\frac{\max\{1,r^{N_F}\}}{\min\{1,r^{N_F}\}}=\max\{r^{N_F},r^{-N_F}\},
\end{equation}
which is minimized iff $r=1$, i.e.\ $\tilde\epsilon=|\alpha|$, giving $D(\tilde\epsilon)=1$.
If $r\neq 1$ then $D(\tilde\epsilon)>1$ and equals \eqref{eq:dyn-range-exp}.
Balancing magnitudes across summands minimizes loss of significance in floating-point summation.
\end{proof}

\subsection{Forward propagation through the layered circuit}
\label{subsec:forward-propagation-layered}

We now specialize to the circuit architecture shown in
Fig.~\ref{fig:circuit_scheme}. The input state is written in the
coherent-product form of Eq.~\eqref{eq:cs-ansatz}. Each layer
$\ell=1,\dots,L$ consists of a DPLO unitary
$U_{\mathrm{DPLO}}^{(\ell)}$ followed by a Kerr gate acting on the first
mode only,
\begin{equation}
\label{eq:csp_first_mode_kerr}
K_{\kappa_\ell}^{(1)}
:=
\exp\!\left(i\kappa_\ell \hat n_1^2\right),
\qquad
\hat n_1=\hat a_1^\dagger \hat a_1 .
\end{equation}
The circuit may end with a final Gaussian unitary $U_G$. Thus the
propagation proceeds layer by layer: the DPLO step updates each
coherent-product component, while the Kerr step branches the
representation on the first mode.

Proposition~\ref{prop:DPLO_preserves_N} shows that the DPLO sublayer
preserves the coherent-product structure and does not increase the
number of terms. For one DPLO step with parameters $(K,\bm{\eta},t)$,
the displacement vector is
\begin{equation}
\label{eq:csp_gamma_closed}
\bm{\gamma}(t)
=
K^{-1}\!\big(e^{-iKt}-I\big)\bm{\eta}
=
-\,K^{-1}\!\big(I-e^{-iKt}\big)\bm{\eta},
\end{equation}
where $K^{-1}\bm{v}$ is understood as the solution of
$K\bm{x}=\bm{v}$, and the $K\to0$ limit gives
$\bm{\gamma}(t)=-it\,\bm{\eta}$. For a product coherent state
\begin{equation}
\ket{\bm{\alpha}}
=
\ket{\alpha^{(1)}}\otimes\cdots\otimes\ket{\alpha^{(m)}},
\qquad
\bm{\alpha}=(\alpha^{(1)},\dots,\alpha^{(m)})^{\mathsf T},
\end{equation}
we define
\begin{equation}
\label{eq:csp_beta_def}
\bm{\beta}:=e^{-iKt}\bm{\alpha}.
\end{equation}
Using the affine Heisenberg update, one obtains
\begin{equation}
\label{eq:csp_dplo_on_coherent}
U_{\rm DPLO}(t)\ket{\bm{\alpha}}
=
\exp\!\Big(
i\chi(t)
+
\frac{1}{2}
\big(
\bm{\gamma}(t)\cdot\bm{\beta}^*
-
\bm{\gamma}(t)^*\cdot\bm{\beta}
\big)
\Big)
\ket{\bm{\alpha}'},
\qquad
\bm{\alpha}'=\bm{\beta}+\bm{\gamma}(t).
\end{equation}
After absorbing the common phase $e^{i\chi(t)}$, each term in
Eq.~\eqref{eq:cs-ansatz} is updated as
\begin{equation}
\label{eq:csp_dplo_update_rule}
\bm{\alpha}_k
\mapsto
\bm{\alpha}_k'
=
e^{-iKt}\bm{\alpha}_k+\bm{\gamma}(t),
\qquad
C_k
\mapsto
\widetilde C_k
=
C_k
\exp\!\Big(
\frac{1}{2}
\big(
\bm{\gamma}(t)\cdot\bm{\beta}_k^*
-
\bm{\gamma}(t)^*\cdot\bm{\beta}_k
\big)
\Big),
\end{equation}
where $\bm{\beta}_k=e^{-iKt}\bm{\alpha}_k$.

After the DPLO step, the state still has the form
\begin{equation}
\label{eq:csp_state_after_dplo}
\ket{\psi_{\ell,\mathrm{DPLO}}}
=
\sum_{k=1}^{N_\ell}
\widetilde C_k
\ket{\alpha_k^{(1)}}\otimes\ket{\alpha_k^{(2)}}
\otimes\cdots\otimes\ket{\alpha_k^{(m)}} ,
\end{equation}
with the same number $N_\ell$ of coherent-product components. The Kerr
gate then acts only on the first factor. Using the coherent-state
reconstruction of Eq.~\eqref{eq:kerr-schematic-coherent}, we replace,
for each $k$,
\begin{equation}
\label{eq:csp_kerr_first_mode_update}
e^{i\kappa_\ell \hat n_1^2}\ket{\alpha_k^{(1)}}
\approx
\sum_{r=0}^{N_F^{(k,\ell)}}
c_r^{(k,\ell)}
\ket{\beta_r^{(k,\ell)}} ,
\end{equation}
where the cutoff $N_F^{(k,\ell)}$ is chosen using
Lemma~\ref{thm:kerr-coherent-truncation} and
Lemma~\ref{lem:fock-to-coherent-switch}. Hence the Kerr sublayer
produces
\begin{equation}
\label{eq:csp_kerr_branching_update}
\begin{aligned}
K_{\kappa_\ell}^{(1)}
\ket{\psi_{\ell,\mathrm{DPLO}}}
\approx
\sum_{k=1}^{N_\ell}
\sum_{r=0}^{N_F^{(k,\ell)}}
\widetilde C_k\,c_r^{(k,\ell)}
\ket{\beta_r^{(k,\ell)}}\otimes
\ket{\alpha_k^{(2)}}\otimes\cdots\otimes\ket{\alpha_k^{(m)}} .
\end{aligned}
\end{equation}
This is the only step that increases the number of coherent-product
terms.

Iterating Eqs.~\eqref{eq:csp_dplo_update_rule} and
\eqref{eq:csp_kerr_branching_update} through all $L$ layers yields,
before the final Gaussian transformation,
\begin{equation}
\label{eq:csp_state_before_gaussian}
\ket{\psi_L}
=
\sum_{j=1}^{N_L}
\widetilde C_j
\ket{\widetilde{\bm{\alpha}}_j},
\end{equation}
where $N_L$ is the number of propagated coherent-product terms after the
last Kerr layer. The final Gaussian unitary acts termwise and maps each
coherent-product component to a Gaussian state,
\begin{equation}
\label{eq:csp_final_gaussian_output}
\ket{\psi'}
=
U_G\ket{\psi_L}
=
\sum_{j=1}^{N_L} C_j'\ket{G_j}
=
\sum_{j=1}^{N_L} C_j'\ket{\Sigma_j,\vec{\mu}_j}.
\end{equation}

We now state the general local-error and cost estimates. Let
\begin{equation}
\label{eq:alpha-max-def}
\alpha_{\max}
:=
\max_{\ell,k}
\big|\alpha_{k,\ell}^{(1)}\big|
\end{equation}
be the largest first-mode coherent amplitude encountered immediately
before a Kerr step, and let
\begin{equation}
    N_F:=\max_{\ell,k}N_F^{(k,\ell)}
\end{equation}
be a uniform cutoff used throughout the propagation. We also denote by
$N_0$ the number of coherent-product components in the input
decomposition.

Fix a local accuracy parameter $\eta\in(0,1)$. We choose the Kerr
reconstruction parameters so that, for every $|\alpha|\le\alpha_{\max}$,
\begin{equation}
\label{eq:csp_local_kerr_error_eta}
\big\|
K_{\kappa_\ell}\ket{\alpha}
-
\ket{\phi_{\alpha,\kappa_\ell}^{(N_F)}}
\big\|_2
\le
\eta .
\end{equation}
By Theorem~\ref{thm:kerr-coherent-expansion},
Lemma~\ref{thm:kerr-coherent-truncation}, and
Lemma~\ref{lem:fock-to-coherent-switch}, it is sufficient to impose
\begin{equation}
\label{eq:csp_tail_plus_leakage_eta}
\sqrt{2\,\tau_{N_F}(\alpha_{\max}^2)}
+
2\,\delta_{N_F}(\alpha_{\max})
\le
\eta,
\end{equation}
where
\begin{equation}
\label{eq:poisson-tail-def}
\tau_N(\lambda)
:=
\sum_{r=N+1}^{\infty}
e^{-\lambda}\frac{\lambda^r}{r!}
\end{equation}
is the Poisson tail, and
\begin{equation}
\label{eq:grid-leakage-def}
\delta_N(\rho)
=
\mathcal{O}\!\left(
\frac{\rho^{N+1}}{\sqrt{(N+1)!}}
\right)
\end{equation}
is the coherent-grid leakage term. Here we used the stable grid-radius
choice $\tilde\epsilon=|\alpha|$ from
Lemma~\ref{lem:kerr-stable-radius}, and then bounded
$\tilde\epsilon\le\alpha_{\max}$ uniformly.

A convenient sufficient choice is
\begin{equation}
\tau_{N_F}(\alpha_{\max}^2)
\le
\frac{\eta^2}{8},
\qquad
\delta_{N_F}(\alpha_{\max})
\le
\frac{\eta}{4}.
\end{equation}
By Lemma~\ref{thm:kerr-coherent-truncation}, the first condition is
satisfied whenever
\begin{equation}
N_F
\ge
\alpha_{\max}^2
+
\sqrt{
2\alpha_{\max}^2
\log\!\left(\frac{8}{\eta^2}\right)
}
+
\frac{2}{3}
\log\!\left(\frac{8}{\eta^2}\right)
-1 .
\end{equation}
The coherent-grid leakage is controlled by the same order of cutoff.
Thus, for local accuracy $\eta$, we may take
\begin{equation}
\label{eq:csp_cutoff_scaling_eta}
N_F
=
\mathcal{O}\!\left(
\alpha_{\max}^2
+
\sqrt{\alpha_{\max}^2\log(1/\eta)}
+
\log(1/\eta)
\right).
\end{equation}
The cutoff estimate is independent of Kerr strength $\kappa_\ell$, since the Kerr gate
only changes Fock phases and leaves the photon-number distribution of a
coherent state unchanged.

Since each Kerr reconstruction maps one coherent-product component to at
most $N_F+1$ coherent-product components, after $L$ layers
\begin{equation}
\label{eq:csp_term_growth_eta}
N_L
\le
N_0(N_F+1)^L .
\end{equation}
For the DPLO sublayer, one must update every existing branch. At layer
$\ell$, computing the matrix exponential and the displacement vector
costs $\mathcal{O}(m^3)$, while applying the affine update
\eqref{eq:csp_dplo_update_rule} to all propagated coherent-product terms
costs $\mathcal{O}(m^2N_{\ell-1})$. Therefore
\begin{equation}
\label{eq:csp_forward_total_cost_eta}
\mathrm{FLOPs}_{\mathrm{forward}}
=
\mathcal{O}\!\left(
Lm^3
+
m^2N_0(N_F+1)^L
\right),
\end{equation}
up to factors polynomial in $N_F$ for evaluating the coherent-grid
coefficients. Combining this with
Eq.~\eqref{eq:csp_cutoff_scaling_eta}, we obtain
\begin{equation}
\label{eq:csp_forward_total_cost_explicit_eta}
\mathrm{FLOPs}_{\mathrm{forward}}
=
\mathcal{O}\!\left(
Lm^3
+
m^2N_0
\left[
\alpha_{\max}^2
+
\sqrt{\alpha_{\max}^2\log(1/\eta)}
+
\log(1/\eta)
\right]^L
\right),
\end{equation}
again up to factors polynomial in $N_F$.

\begin{theorem}[Coherent-state propagation for layered Kerr--DPLO circuits]
\label{thm:csp-layered-runtime}
Consider the $m$-mode circuit in Fig.~\ref{fig:circuit_scheme},
\begin{equation}
    U
    :=
    U_G V_L \cdots V_1,
    \qquad
    V_\ell
    :=
    K_{\kappa_\ell}^{(1)} U_{\mathrm{DPLO}}^{(\ell)},
    \qquad
    K_{\kappa_\ell}^{(1)}
    :=
    \exp(i\kappa_\ell \hat n_1^2),
\end{equation}
where $U_G$ is a final Gaussian unitary, $ K_{\kappa}$ is Kerr gate, $U_{\mathrm{DPLO}}$ is DPLO layer.

Assume the input state admits the coherent-product expansion
\begin{equation}
    \ket{\psi_0}
    =
    \sum_{r=1}^{N_0}
    C_r^{(0)} \ket{\bm{\alpha}_r^{(0)}} .
\end{equation}
For each layer $\ell$, let the DPLO update on coherent amplitudes be
\begin{equation}
    \bm{\alpha}
    \mapsto
    S_\ell \bm{\alpha} + \bm{\gamma}^{(\ell)},
    \qquad
    S_\ell^\dagger S_\ell = I,
\end{equation}
and define the input amplitude mass and displacement budget
\begin{equation}
    A_{\rm in}
    :=
    \max_{1\le r\le N_0}
    \|\bm{\alpha}_r^{(0)}\|_2,
    \qquad
    \Gamma
    :=
    \sum_{\ell=1}^{L}
    \|\bm{\gamma}^{(\ell)}\|_2 .
\end{equation}
Then every coherent amplitude vector appearing immediately before a Kerr
update obeys
\begin{equation}
    \|\bm{\alpha}\|_2
    \le
    A_{\rm in}+\Gamma.
\end{equation}
Use the same Kerr cutoff $N_F$ at every layer, chosen from the
worst-case pre-Kerr amplitude bound $A_{\rm in}+\Gamma$ and the local
error budget $\varepsilon/L$. A sufficient uniform choice is
\begin{equation}
    N_F
    =
    \mathcal{O}\!\left(
        (A_{\rm in}+\Gamma)^2
        +
        (A_{\rm in}+\Gamma)\sqrt{\log(L/\varepsilon)}
        +
        \log(L/\varepsilon)
    \right).
\end{equation}
With this choice, coherent-state propagation outputs a state
$\ket{\widetilde\psi}$ such that
\begin{equation}
    \bigl\|
        U\ket{\psi_0}-\ket{\widetilde\psi}
    \bigr\|_2
    \le
    \varepsilon .
\end{equation}

The number of propagated branches entering the final Gaussian unitary
satisfies
\begin{equation}
    N_L
    \le
    N_0 (N_F+1)^L .
\end{equation}
Consequently, the forward propagation cost is
\begin{equation}
    \mathrm{FLOPs}_{\mathrm{forward}}
    =
    \mathcal{O}\!\left(
        Lm^3
        +
        m^2 N_0 (N_F+1)^L
    \right).
\end{equation}
\end{theorem}
\begin{proof}
The DPLO update has the affine form
\begin{equation}
    \bm{\alpha}\mapsto S_\ell\bm{\alpha}+\bm{\gamma}^{(\ell)},
    \qquad
    S_\ell^\dagger S_\ell=I .
\end{equation}
Therefore
\begin{equation}
    \|S_\ell\bm{\alpha}+\bm{\gamma}^{(\ell)}\|_2
    \le
    \|\bm{\alpha}\|_2+\|\bm{\gamma}^{(\ell)}\|_2 .
\end{equation}
Iterating this inequality over all layers gives
\begin{equation}
    \|\bm{\alpha}\|_2
    \le
    A_{\rm in}
    +
    \sum_{\ell=1}^L
    \|\bm{\gamma}^{(\ell)}\|_2
    =
    A_{\rm in}+\Gamma .
\end{equation}
Thus every first-mode amplitude encountered before a Kerr update obeys
\begin{equation}
    |\alpha^{(1)}|
    \le
    A_{\rm in}+\Gamma .
\end{equation}

We choose the local Kerr reconstruction error to be
\begin{equation}
    \eta:=\frac{\varepsilon}{L}.
\end{equation}
By Theorem~\ref{thm:kerr-coherent-expansion},
Lemma~\ref{thm:kerr-coherent-truncation}, and
Lemma~\ref{lem:fock-to-coherent-switch}, it is sufficient to choose the
cutoff so that, uniformly for all
$|\alpha|\le A_{\rm in}+\Gamma$,
\begin{equation}
    \big\|
        K_{\kappa_\ell}\ket{\alpha}
        -
        \ket{\phi_{\alpha,\kappa_\ell}^{(N_F)}}
    \big\|_2
    \le
    \eta,
\end{equation}
where $ \ket{\phi_{\alpha,\kappa_\ell}^{(N_F)}}$ is given in \eqref{eq:phi-final}. Using the Poisson-tail bound and the coherent-grid leakage bound, a
sufficient choice is
\begin{equation}
    N_F
    =
    O\!\left(
        (A_{\rm in}+\Gamma)^2
        +
        \sqrt{(A_{\rm in}+\Gamma)^2\log(1/\eta)}
        +
        \log(1/\eta)
    \right).
\end{equation}
Since $\eta=\varepsilon/L$, this becomes
\begin{equation}
    N_F
    =
    O\!\left(
        (A_{\rm in}+\Gamma)^2
        +
        \sqrt{(A_{\rm in}+\Gamma)^2\log(L/\varepsilon)}
        +
        \log(L/\varepsilon)
    \right).
\end{equation}

It remains to compose the local errors. Let
\begin{equation}
    V_\ell
    :=
    K_{\kappa_\ell}^{(1)}
    U_{\mathrm{DPLO}}^{(\ell)},
    \qquad
    \widetilde V_\ell
    :=
    \widetilde K_{\kappa_\ell,N_F}^{(1)}
    U_{\mathrm{DPLO}}^{(\ell)}
\end{equation}
be the exact and approximate one-layer maps. The telescopic identity gives
\begin{equation}
    V_L\cdots V_1
    -
    \widetilde V_L\cdots \widetilde V_1
    =
    \sum_{\ell=1}^L
    V_L\cdots V_{\ell+1}
    \big(
        V_\ell-\widetilde V_\ell
    \big)
    \widetilde V_{\ell-1}\cdots \widetilde V_1 .
\end{equation}
Applying this identity to $\ket{\psi_0}$ and taking the Hilbert-space
norm,
\begin{equation}
\begin{aligned}
    \big\|
        V_L\cdots V_1\ket{\psi_0}
        -
        \widetilde V_L\cdots\widetilde V_1\ket{\psi_0}
    \big\|_2
    &\le
    \sum_{\ell=1}^L
    \left\|
        V_L\cdots V_{\ell+1}
        \big(
            V_\ell-\widetilde V_\ell
        \big)
        \widetilde V_{\ell-1}\cdots\widetilde V_1
        \ket{\psi_0}
    \right\|_2 .
\end{aligned}
\end{equation}
The exact layers $V_L,\dots,V_{\ell+1}$ are unitary, so they do not
change the norm. By construction, each remaining summand is at most
$\eta=\varepsilon/L$. Hence
\begin{equation}
    \big\|
        V_L\cdots V_1\ket{\psi_0}
        -
        \widetilde V_L\cdots\widetilde V_1\ket{\psi_0}
    \big\|_2
    \le
    \sum_{\ell=1}^L\frac{\varepsilon}{L}
    =
    \varepsilon .
\end{equation}
The final Gaussian unitary $U_G$ is common to both evolutions and is
unitary, so it preserves the distance.

Finally, each Kerr reconstruction replaces one coherent-product
component by at most $N_F+1$ coherent-product components. Thus
\begin{equation}
    N_L
    \le
    N_0(N_F+1)^L .
\end{equation}
The DPLO update costs $O(m^3)$ per layer to compute the matrix
exponential and displacement vector, and $O(m^2)$ per propagated
coherent-product component. Therefore
\begin{equation}
    \mathrm{FLOPs}_{\mathrm{forward}}
    =
    O\!\left(
        Lm^3
        +
        m^2N_0(N_F+1)^L
    \right),
\end{equation}
up to factors polynomial in $N_F$ for evaluating the coherent-grid
coefficients.
\end{proof}

\begin{corollary}[Two log-depth regimes]
\label{cor:log-depth-special-regimes}
Consider the setting of Theorem~\ref{thm:csp-layered-runtime}. Assume
that the depth satisfies
\begin{equation}
    L=O(\log m),
\end{equation}
that the input decomposition has size
\begin{equation}
    N_0=\mathrm{poly}(m),
\end{equation}
and that each DPLO layer has constant displacement mass,
\begin{equation}
    \|\bm{\gamma}^{(\ell)}\|_2=O(1),
    \qquad \ell=1,\dots,L.
\end{equation}
Equivalently, the total displacement budget satisfies
\begin{equation}
    \Gamma
    :=
    \sum_{\ell=1}^L
    \|\bm{\gamma}^{(\ell)}\|_2
    =
    O(L)
    =
    O(\log m).
\end{equation}
Then the following two regimes hold.

\begin{enumerate}
    \item If
    \begin{equation}
        A_{\rm in}=O(\log m),
        \qquad
        \varepsilon=\frac{1}{\mathrm{poly}(m)},
    \end{equation}
    then coherent-state propagation runs in time
    \begin{equation}
        m^{O(\log\log m)}.
    \end{equation}

\item If
\begin{equation}
    A_{\rm in}=O(m^a),
    \qquad
    \varepsilon \ge \exp(-m^b)
\end{equation}
for constants $a,b>0$, then coherent-state propagation runs in time
\begin{equation}
    m^{O(\log m)}.
\end{equation}
\end{enumerate}
\end{corollary}

\begin{proof}
By Theorem~\ref{thm:csp-layered-runtime}, the required Kerr cutoff is
\begin{equation}
    N_F
    =
    O\!\left(
        (A_{\rm in}+\Gamma)^2
        +
        \sqrt{(A_{\rm in}+\Gamma)^2\log(L/\varepsilon)}
        +
        \log(L/\varepsilon)
    \right),
\end{equation}
and the propagation cost is
\begin{equation}
    O\!\left(
        Lm^3
        +
        m^2N_0(N_F+1)^L
    \right),
\end{equation}
up to factors polynomial in $N_F$.

Since $\|\bm{\gamma}^{(\ell)}\|_2=O(1)$ and
$L=O(\log m)$, we have
\begin{equation}
    \Gamma=O(\log m).
\end{equation}

For the first regime, $A_{\rm in}=O(\log m)$, hence
\begin{equation}
    A_{\rm in}+\Gamma=O(\log m).
\end{equation}
Moreover, for $\varepsilon=1/\mathrm{poly}(m)$,
\begin{equation}
    \log(L/\varepsilon)=O(\log m).
\end{equation}
Therefore
\begin{equation}
    N_F
    =
    O\!\left(
        \log^2 m
        +
        \sqrt{\log^2 m\cdot \log m}
        +
        \log m
    \right)
    =
    O(\log^2 m).
\end{equation}
Since $L=O(\log m)$,
\begin{equation}
    (N_F+1)^L
    =
    O(\log^2 m)^{O(\log m)}
    =
    \exp\!\big(O(\log m\log\log m)\big)
    =
    m^{O(\log\log m)}.
\end{equation}
Using $N_0=\mathrm{poly}(m)$, the total runtime is therefore
\begin{equation}
    m^{O(\log\log m)}.
\end{equation}

For the second regime, let $A_{\rm in}=m^a$ for some constant $a>0$,
without loss of generality up to polynomial factors. Since
$\Gamma=O(\log m)$, we have
\begin{equation}
    A_{\rm in}+\Gamma=\mathrm{poly}(m).
\end{equation}
For $\varepsilon=e^{-m}$,
\begin{equation}
    \log(L/\varepsilon)=O(m).
\end{equation}
Hence
\begin{equation}
    N_F
    =
    O\!\left(
        (A_{\rm in}+\Gamma)^2
        +
        (A_{\rm in}+\Gamma)\sqrt{m}
        +
        m
    \right)
    =
    \mathrm{poly}(m).
\end{equation}
Therefore, since $L=O(\log m)$,
\begin{equation}
    (N_F+1)^L
    =
    \mathrm{poly}(m)^{O(\log m)}
    =
    m^{O(\log m)}.
\end{equation}
Again using $N_0=\mathrm{poly}(m)$, the total runtime is
\begin{equation}
    m^{O(\log m)}.
\end{equation}
\end{proof}

\section{Approximations for small non-Gaussianity}
\label{sec:small}

Section~\ref{sec:kerr_via_fock} provides a general coherent-state expansion for \(e^{i\kappa \hat{n}^2}\ket{\alpha}\) with controlled error, valid uniformly for arbitrary Kerr strength \(\kappa\) and arbitrary coherent amplitude \(\alpha\) (arbitrary input energy scale $\langle \hat n_1\rangle=|\alpha|^2$). This generality comes at a computational cost: the required representation size is governed by a Poisson distribution cutoff that scales with \(|\alpha|^2\) (up to accuracy-dependent factors, see \eqref{eq:N-choice}), and repeated Kerr applications can therefore lead to rapid exponential growth in the number of coherent components $N$ (see Eq.\eqref{eq:csp_forward_total_cost_eta}).

In this section we study an approximation regime, namely \emph{small non-Gaussianity}, where the Kerr strength $\kappa$ is sufficiently small so that the simulation cost does not grow exponentially with the number \(L\) of Kerr-gate applications. Our objective is to identify this regime of polynomial cost $N\sim \mathrm{poly}(L)$ and derive an explicit approximation scheme with bounded errors.

We first isolate the elementary Kerr update on a single coherent state \(\ket{\alpha}\). This is the relevant primitive because, in the forward propagation scheme of Section~\ref{sec:coherent-state-calculus}, the evolving multimode state is represented as a finite linear combination of tensor products of single-mode coherent states (see Eq.~\eqref{eq:cs-ansatz}), while each Kerr gate acts on one mode at a time. Therefore, the action of the Kerr sublayer reduces to applying a single-mode map \(e^{i\kappa \hat n^2}\ket{\alpha}\) componentwise within the representation \eqref{eq:cs-ansatz}. After deriving a controlled approximation for this elementary step in the small-\(\kappa\) regime, we will use it as the local update rule in the full forward propagation algorithm.

\subsection{Finite Fourier approximation for Kerr updates}
\label{subsec:finite-fourier-cutoff}

Section~\ref{sec:kerr_via_fock} gives a general coherent-state expansion of \(K_\kappa\ket{\alpha}\) for arbitrary Kerr strength \(\kappa\). In particular, Corollary~\ref{cor:kerr-trunc-coherent-grid} provides explicit coherent-grid coefficients after a chosen Fock cutoff, with bounded error. For the numerical tests in Section~\ref{sec:numerics}, we use a complementary construction in which the number of coherent terms is fixed directly. The parameter \(M\in\mathbb N\) below yields exactly \(2M\) coherent components.

We start from the phase orbit
\begin{equation}
f_\alpha(\phi):=\ket{\alpha e^{i\phi}}.
\end{equation}
Using the Fock expansion of a coherent state,
\begin{equation}
f_\alpha(\phi)
=
e^{-|\alpha|^2/2}\sum_{n=0}^{\infty}\frac{\alpha^n e^{in\phi}}{\sqrt{n!}}\ket n,
\end{equation}
a direct differentiation gives
\begin{equation}
-\partial_\phi^2 f_\alpha(\phi)
=
e^{-|\alpha|^2/2}\sum_{n=0}^{\infty}\frac{n^2\alpha^n e^{in\phi}}{\sqrt{n!}}\ket n
=
\hat n^2 f_\alpha(\phi).
\end{equation}
Thus, on the phase orbit, the Kerr gate \(K_\kappa=e^{i\kappa \hat n^2}\) is represented by the propagator \(e^{-i\kappa\partial_\phi^2}\).

To obtain a finite coherent-state expansion, we discretize the orbit parameter \(\phi\) and replace \(\partial_\phi^2\) by the central second-order finite difference
\begin{equation}
(\delta_h^2u)(\phi):=\frac{u(\phi+h)-2u(\phi)+u(\phi-h)}{h^2}.
\end{equation}
On the Fourier modes \(e^{in\phi}\), one has
\begin{equation}
-\delta_h^2 e^{in\phi}
=
\mu_h(n)e^{in\phi},
\qquad
\mu_h(n):=\frac{4}{h^2}\sin^2\!\Big(\frac{nh}{2}\Big).
\end{equation}
Hence the finite-difference discretization replaces the exact symbol \(n^2\) by its discrete counterpart \(\mu_h(n)\).

This suggests the finite-difference propagator
\begin{equation}
\widetilde K_{\kappa,h}:=\exp(-i\kappa\,\delta_h^2),
\end{equation}
understood as acting only on the \(\phi\)-dependence of \(f_\alpha\). We now choose \(h=\pi/M\) and restrict to the periodic grid
\begin{equation}
\phi_j:=j\pi/M,\qquad j=0,\dots,2M-1.
\end{equation}
Let \(T\) denote the cyclic shift on this grid,
\begin{equation}
(Tu)(\phi_j):=u(\phi_{j+1\!\!\!\!\pmod{2M}}).
\end{equation}
Since \(T^{2M}=I\) and \(\delta_h^2=(T+T^{-1}-2)/h^2\), the operator \(\widetilde K_{\kappa,\pi/M}\) is a function of \(T\), hence circulant. Therefore it admits a unique expansion
\begin{equation}
\label{eq:finite-fourier-shift-expansion}
\widetilde K_{\kappa,\pi/M}
=
\sum_{j=0}^{2M-1} c_j^{(M)}\,T^j .
\end{equation}
Applying this to the orbit \(f_\alpha(\phi)=\ket{\alpha e^{i\phi}}\) and evaluating at \(\phi=0\), we obtain
\begin{equation}
\label{eq:finite-fourier-kerr}
\widetilde K_{\kappa,\pi/M}\ket{\alpha}
=
\sum_{j=0}^{2M-1} c_j^{(M)}\ket{\alpha e^{i\pi j/M}}.
\end{equation}

The coefficients $c_j^{(M)}$ are obtained by diagonalizing \(\widetilde K_{\kappa,\pi/M}\) in the discrete Fourier basis
\begin{equation}
g_n(\phi_j):=e^{in\phi_j},
\qquad n=0,\dots,2M-1.
\end{equation}
Since \(Tg_n=e^{i\pi n/M}g_n\), we have
\begin{equation}
\widetilde K_{\kappa,\pi/M}g_n=e^{i\kappa \Lambda_n} g_n,
\qquad
\Lambda_n
:=
\frac{4M^2}{\pi^2}\sin^2\!\Big(\frac{n\pi}{2M}\Big).
\end{equation}
On the other hand, Eq.~\eqref{eq:finite-fourier-shift-expansion} gives
\begin{equation}
\widetilde K_{\kappa,\pi/M}g_n
=
\left(\sum_{j=0}^{2M-1} c_j^{(M)}e^{i\pi nj/M}\right)g_n.
\end{equation}
Therefore \(\{c_j^{(M)}\}\) is the discrete Fourier transform of \(\{\exp \left(i\kappa \Lambda_n\right)\}\), namely
\begin{equation}
\label{eq:finite-fourier-coeffs}
c_j^{(M)}
=
\frac{1}{2M}\sum_{n=0}^{2M-1}
\exp\!\left(
i\kappa\,\frac{4M^2}{\pi^2}\sin^2\!\Big(\frac{n\pi}{2M}\Big)
\right)
e^{-i\pi nj/M},
\qquad j=0,\dots,2M-1.
\end{equation}

A short normalization check follows from the same diagonalization. The eigenvalues \(e^{i\kappa\Lambda_n}\) are pure phases, so \(|e^{i\kappa\Lambda_n}|=1\) for all \(n\), and therefore \(\widetilde K_{\kappa,\pi/M}\) is unitary on the \(2M\)-point phase grid. Equivalently,
\begin{equation}
\widetilde K_{\kappa,\pi/M}\ket{\alpha}
=
e^{-|\alpha|^2/2}\sum_{m=0}^{\infty}\frac{\alpha^m}{\sqrt{m!}}
\exp\!\left(
i\kappa\,\frac{4M^2}{\pi^2}\sin^2\!\Big(\frac{(m\bmod 2M)\pi}{2M}\Big)
\right)\ket m,
\end{equation}
so the finite-Fourier cutoff changes only the Fock phases and therefore preserves the norm.

We do not give a general prescription for choosing $M$. In
numerical implementations, however, these parameters can be tuned by monitoring
a local one-step discrepancy. Namely, one compares the weak-Kerr finite-Fourier
update in Eq.~\eqref{eq:finite-fourier-kerr}, with coefficients from
Eq.~\eqref{eq:finite-fourier-coeffs}, against the general Kerr reconstruction
described in Eq.~\eqref{eq:kerr-schematic-coherent}.

Suppose that the input to a Kerr sublayer contains $\mathcal S$ coherent-product
branches. The weak-Kerr update produces $2M\mathcal S$ branches, while the
general reconstruction in Eq.~\eqref{eq:csp_kerr_branching_update} produces
$\mathcal S(N_F+1)$ branches:
\begin{equation}
\begin{aligned}
    \ket{\psi_{\mathrm{wk}}^{(\ell)}}
    &=
    \sum_{j=1}^{2M\mathcal S}
    D_j^{(\ell)}
    \ket{\widetilde{\bm\alpha}_j^{(\ell)}},
    \\
    \ket{\psi_{\mathrm{gen}}^{(\ell)}}
    &=
    \sum_{i=1}^{\mathcal S(N_F+1)}
    C_i^{(\ell)}
    \ket{\bm\alpha_i^{(\ell)}} .
\end{aligned}
\end{equation}
After normalizing both states, the local Hilbert-space discrepancy is
\begin{equation}
\begin{aligned}
    \left\|
    \ket{\psi_{\mathrm{gen}}^{(\ell)}}
    -
    \ket{\psi_{\mathrm{wk}}^{(\ell)}}
    \right\|_2
    =
    \sqrt{
    2
    -
    2\,\Re
    \braket{
    \psi_{\mathrm{gen}}^{(\ell)}
    }{
    \psi_{\mathrm{wk}}^{(\ell)}
    }
    } .
\end{aligned}
\end{equation}
The overlap is evaluated analytically using
\begin{equation}
    \braket{\bm\alpha}{\widetilde{\bm\alpha}}
    =
    \prod_{p=1}^{m}
    \exp\!\left(
    -\frac{|\alpha^{(p)}|^2}{2}
    -\frac{|\widetilde\alpha^{(p)}|^2}{2}
    +
    (\alpha^{(p)})^*\widetilde\alpha^{(p)}
    \right).
\end{equation}
Thus, evaluating the full overlap costs
$\mathcal O(m\,M\,\mathcal S^2 N_F)$ FLOPs. This gives a practical local
monitor for fine-tuning $M$. Any accumulated multi-layer
estimate obtained from these local quantities is generally loose for a global error evaluation, but still
useful as a conservative diagnostic during numerical simulations.

\subsection{Single Kerr-gate approximation}\label{subsec: two term}

Motivated by the finite Fourier form \eqref{eq:finite-fourier-kerr} we now introduce a two-component coherent-state approximation for the elementary Kerr update on a single coherent input for small non-Gaussianity $\kappa$. We write
\begin{equation}
\label{eq:psi2-theta-def}
\ket{\psi_{2}(\theta)}
:=
\frac{1+e^{i\theta}}{2}\ket{\alpha}
+
\frac{1-e^{i\theta}}{2}\ket{-\alpha},
\qquad \theta\in\mathbb{R},
\end{equation}
where the phase parameter $\theta$ will be chosen later.

Let us first show that the state in Eq.~\eqref{eq:psi2-theta-def} is normalized. One has
\begin{equation}
\begin{aligned}
\braket{\psi_2(\theta)}{\psi_2(\theta)}
={}&
\frac{1+\cos\theta}{2}
+
\frac{1-\cos\theta}{2}
-\frac{i}{2}\sin\theta\,\braket{\alpha}{-\alpha}
+\frac{i}{2}\sin\theta\,\braket{-\alpha}{\alpha}.
\end{aligned}
\end{equation}
Since
\begin{equation}
\braket{\alpha}{-\alpha}=\braket{-\alpha}{\alpha}=e^{-2|\alpha|^2}\in\mathbb{R},
\end{equation}
the last two terms cancel. Therefore
\begin{equation}
\braket{\psi_2(\theta)}{\psi_2(\theta)}=1.
\end{equation}

Next, we quantify the approximation error of the normalized two-component state \eqref{eq:psi2-theta-def} and derive an explicit upper bound valid for arbitrary $\theta\in\mathbb{R}$.

\begin{theorem}[Two-component coherent-state approximation]
\label{thm:L1-kerr-coherent-general-theta}
Let $K_\kappa$ be the Kerr gate defined in Eq.~\eqref{eq:kerr_gate}, let $\ket{\alpha}$ be a coherent state with $\alpha\in\mathbb C$, and write $\lambda:=|\alpha|^2$. For any $\theta\in\mathbb R$, define $\ket{\psi_{2}(\theta)}$ as in Eq.~\eqref{eq:psi2-theta-def}. Then
\begin{equation}
\label{eq:L1-general-theta-bound}
\begin{aligned}
\left\|K_\kappa\ket{\alpha}-\ket{\psi_2(\theta)}\right\|_2^2
\le\;&
\kappa^2\bigl(\lambda^4+6\lambda^3+7\lambda^2+\lambda\bigr)
\\
&\;
-\kappa\theta\Bigl((\lambda^2+\lambda)-e^{-2\lambda}(\lambda^2-\lambda)\Bigr)
+\frac{\theta^2}{2}\bigl(1-e^{-2\lambda}\bigr).
\end{aligned}
\end{equation}
\end{theorem}

\begin{proof}
Expanding Eq.~\eqref{eq:psi2-theta-def} in the Fock basis, we obtain
\begin{equation}
\label{eq:psi2-fock-expansion}
\ket{\psi_2(\theta)}
=
e^{-\lambda/2}\sum_{n=0}^{\infty}\frac{\alpha^n}{\sqrt{n!}}
\begin{cases}
\ket n, & n \text{ even},\\[2mm]
e^{i\theta}\ket n, & n \text{ odd},
\end{cases}
\qquad \lambda:=|\alpha|^2.
\end{equation}

On the other hand, the exact Kerr-evolved coherent state is
\begin{equation}
\label{eq:kerr-coherent}
K_\kappa\ket{\alpha}
=
e^{-\lambda/2}\sum_{n=0}^{\infty}\frac{\alpha^n}{\sqrt{n!}}\,e^{i\kappa n^2}\ket n.
\end{equation}

Therefore, by orthogonality of the Fock basis,
\begin{equation}
\label{eq:norm-split}
\begin{aligned}
\left\|K_\kappa\ket{\alpha}-\ket{\psi_2(\theta)}\right\|_2^2
={}&
e^{-\lambda}\sum_{n\,\mathrm{even}}\frac{\lambda^n}{n!}\,\bigl|e^{i\kappa n^2}-1\bigr|^2
\\
&+
e^{-\lambda}\sum_{n\,\mathrm{odd}}\frac{\lambda^n}{n!}\,\bigl|e^{i\kappa n^2}-e^{i\theta}\bigr|^2.
\end{aligned}
\end{equation}

Using
\begin{equation}
\label{eq:trig-ineq}
|e^{ix}-e^{iy}|^2
=
4\sin^2\!\left(\frac{x-y}{2}\right)
\le (x-y)^2,
\qquad x,y\in\mathbb R,
\end{equation}
we obtain
\begin{equation}
\label{eq:after-bound}
\left\|K_\kappa\ket{\alpha}-\ket{\psi_2(\theta)}\right\|_2^2
\le
e^{-\lambda}\sum_{n\,\mathrm{even}}\frac{\lambda^n}{n!}\,\kappa^2 n^4
+
e^{-\lambda}\sum_{n\,\mathrm{odd}}\frac{\lambda^n}{n!}\,(\kappa n^2-\theta)^2.
\end{equation}

Expanding the square in the odd sector yields
\begin{equation}
\label{eq:expand}
\begin{aligned}
\left\|K_\kappa\ket{\alpha}-\ket{\psi_2(\theta)}\right\|_2^2
\le\;&
e^{-\lambda}\sum_{n\,\mathrm{even}}\frac{\lambda^n}{n!}\,\kappa^2 n^4
\\
&+
e^{-\lambda}\sum_{n\,\mathrm{odd}}\frac{\lambda^n}{n!}\,
\bigl(\kappa^2 n^4-2\kappa\theta n^2+\theta^2\bigr)
\\[2mm]
=\;&
\kappa^2 e^{-\lambda}\sum_{n=0}^{\infty}\frac{\lambda^n}{n!}\,n^4
-2\kappa\theta\, e^{-\lambda}\sum_{n\,\mathrm{odd}}\frac{\lambda^n}{n!}\,n^2
+\theta^2 e^{-\lambda}\sum_{n\,\mathrm{odd}}\frac{\lambda^n}{n!}.
\end{aligned}
\end{equation}

Now let \(M\) be distributed according to the Poisson distribution with parameter \(\lambda\), i.e.
$M\sim\mathrm{Pois}(\lambda)$.
Then Eq.~\eqref{eq:expand} can be rewritten as
\begin{equation}
\label{eq:prob-form}
\left\|K_\kappa\ket{\alpha}-\ket{\psi_2(\theta)}\right\|_2^2
\le
\kappa^2\,\mathbb E[M^4]
-2\kappa\theta\,\mathbb E\!\left[M^2\mathbf 1_{\mathrm{odd}}(M)\right]
+\theta^2\,\mathbb P(M\text{ odd}).
\end{equation}

Using the standard Poisson identities
\begin{equation}
\label{eq:m4}
\mathbb E[M^4]=\lambda^4+6\lambda^3+7\lambda^2+\lambda,
\end{equation}
\begin{equation}
\label{eq:odd-prob}
\mathbb P(M\text{ odd})=\frac{1-e^{-2\lambda}}{2},
\end{equation}
and
\begin{equation}
\label{eq:odd-moment}
\mathbb E\!\left[M^2\mathbf 1_{\mathrm{odd}}(M)\right]
=
\frac12\Bigl((\lambda^2+\lambda)-e^{-2\lambda}(\lambda^2-\lambda)\Bigr),
\end{equation}
we obtain
\begin{equation}
\label{eq:final}
\begin{aligned}
\left\|K_\kappa\ket{\alpha}-\ket{\psi_2(\theta)}\right\|_2^2
\le\;&
\kappa^2\bigl(\lambda^4+6\lambda^3+7\lambda^2+\lambda\bigr)
\\
&\;
-\kappa\theta\Bigl((\lambda^2+\lambda)-e^{-2\lambda}(\lambda^2-\lambda)\Bigr)
+\frac{\theta^2}{2}\bigl(1-e^{-2\lambda}\bigr),
\end{aligned}
\end{equation}
which is exactly Eq.~\eqref{eq:L1-general-theta-bound}.
\end{proof}

Let us now find the optimal choice of \(\theta\) in the upper bound of Theorem~\ref{thm:L1-kerr-coherent-general-theta}.

\begin{corollary}[Optimal choice of \(\theta\) for Theorem~\ref{thm:L1-kerr-coherent-general-theta}]
\label{cor:L1-kerr-coherent-optimal-theta}
Assume \(\lambda>0\), and define
\begin{equation}
\label{eq:theta-opt-def}
\theta_{\mathrm{opt}}
:=
\kappa\,
\frac{(\lambda^2+\lambda)-e^{-2\lambda}(\lambda^2-\lambda)}
{1-e^{-2\lambda}}.
\end{equation}
Then \(\theta_{\mathrm{opt}}\) minimizes the quadratic upper bound
\eqref{eq:L1-general-theta-bound} over all \(\theta\in\mathbb R\), and
\begin{equation}
\label{eq:L1-optimal-theta-bound}
\begin{aligned}
\left\|K_\kappa\ket{\alpha}-\ket{\psi_2(\theta_{\mathrm{opt}})}\right\|_2^2
\le\;&
\kappa^2\Bigg[
\lambda^4+6\lambda^3+7\lambda^2+\lambda
\\
&\hspace{1.3cm}
-\frac{\Bigl((\lambda^2+\lambda)-e^{-2\lambda}(\lambda^2-\lambda)\Bigr)^2}
{2\bigl(1-e^{-2\lambda}\bigr)}
\Bigg].
\end{aligned}
\end{equation}
For \(\lambda=0\), one has \(K_\kappa\ket{\alpha}=\ket{\psi_2(\theta)}=\ket{0}\) for every \(\theta\), so the error is exactly zero.
\end{corollary}

\begin{proof}
For \(\lambda>0\), the right-hand side of Eq.~\eqref{eq:L1-general-theta-bound} is a quadratic polynomial in \(\theta\) with positive leading coefficient
\begin{equation}
\frac{1-e^{-2\lambda}}{2}>0.
\end{equation}
Differentiating with respect to \(\theta\) and setting the derivative equal to zero gives
\begin{equation}
-\kappa\Bigl((\lambda^2+\lambda)-e^{-2\lambda}(\lambda^2-\lambda)\Bigr)
+\theta\bigl(1-e^{-2\lambda}\bigr)=0,
\end{equation}
which yields Eq.~\eqref{eq:theta-opt-def}. Substituting \(\theta=\theta_{\mathrm{opt}}\) into Eq.~\eqref{eq:L1-general-theta-bound} gives Eq.~\eqref{eq:L1-optimal-theta-bound}. The case \(\lambda=0\) is immediate.
\end{proof}

Next, consider the asymptotic behavior of the optimized bound \eqref{eq:L1-optimal-theta-bound} in the regimes of small and large mean photon number \(\lambda=\langle \hat n\rangle\).

For \(\lambda\to 0\), the coefficient in Eq.~\eqref{eq:L1-optimal-theta-bound} admits the Maclaurin expansion
\begin{equation}
\label{eq:L1-optimal-small-lambda}
\begin{aligned}
&\lambda^4+6\lambda^3+7\lambda^2+\lambda
-\frac{\Bigl((\lambda^2+\lambda)-e^{-2\lambda}(\lambda^2-\lambda)\Bigr)^2}
{2\bigl(1-e^{-2\lambda}\bigr)}
\\
&\qquad\qquad
=
8\lambda^2+\frac{8}{3}\lambda^3+4\lambda^4+\mathcal{O}(\lambda^5).
\end{aligned}
\end{equation}
Hence Eq.~\eqref{eq:L1-optimal-theta-bound} yields
\begin{equation}
\left\|K_\kappa\ket{\alpha}-\ket{\psi_2(\theta_{\mathrm{opt}})}\right\|_2^2
\le
\kappa^2\Bigl(8\lambda^2+\frac{8}{3}\lambda^3+4\lambda^4+\mathcal{O}(\lambda^5)\Bigr),
\qquad \lambda\to 0,
\end{equation}
and therefore
\begin{equation}
\label{eq:L1-optimal-small-lambda-norm}
\left\|K_\kappa\ket{\alpha}-\ket{\psi_2(\theta_{\mathrm{opt}})}\right\|_2
=\mathcal{O}(\kappa\lambda),
\qquad \lambda\to 0.
\end{equation}

For \(\lambda\to+\infty\), since \(e^{-2\lambda}\to 0\), one has
\begin{equation}
\lambda^4+6\lambda^3+7\lambda^2+\lambda
-\frac{\Bigl((\lambda^2+\lambda)-e^{-2\lambda}(\lambda^2-\lambda)\Bigr)^2}
{2\bigl(1-e^{-2\lambda}\bigr)}
=
\frac{1}{2}\lambda^4+5\lambda^3+\frac{13}{2}\lambda^2+\lambda+\mathcal{O}(e^{-2\lambda}\lambda^4).
\end{equation}
Consequently,
\begin{equation}
\left\|K_\kappa\ket{\alpha}-\ket{\psi_2(\theta_{\mathrm{opt}})}\right\|_2^2
\le
\kappa^2\left(\frac{1}{2}\lambda^4+5\lambda^3+\frac{13}{2}\lambda^2+\lambda+\mathcal{O}(e^{-2\lambda}\lambda^4)\right),
\qquad \lambda\to+\infty,
\end{equation}
so that
\begin{equation}
\label{eq:L1-optimal-large-lambda-norm}
\left\|K_\kappa\ket{\alpha}-\ket{\psi_2(\theta_{\mathrm{opt}})}\right\|_2
=\mathcal{O}(\kappa\lambda^2),
\qquad \lambda\to+\infty.
\end{equation}

\subsection{The tree structure with online truncation}\label{subsec: tree structure}

We now specialize the small-nonlinearity approximation regime to an $m$-mode circuit architecture in which a
\emph{fixed} mode undergoes repeated Kerr updates, interleaved with displaced passive linear-optics
(DPLO) layers. A schematic of the circuit is shown in Fig.~\ref{fig:circuit_scheme}.
Here $U^{(q)}_{\mathrm{DPLO}}$ denotes an arbitrary unitary generated by the DPLO Hamiltonian $U^{(q)}_{\mathrm{DPLO}} = e^{it_q H_{\mathrm{DPLO}}}$ with $t_q\in\mathbb{R}$, see Eq.~\eqref{eq:H_DPLO}, possibly composed with mode-swap operations.
As discussed in Section~\ref{subsec:displaced-linear-optics}, such transformations preserve the
factorized coherent-state structure \eqref{eq:cs-ansatz} and do not increase the number of coherent-product terms $N$.

We take as input an $m$-mode state of the form \eqref{eq:coherent_superposition_background}. Our goal is to propagate
this representation through $L$ interleaved Kerr/DPLO layers while avoiding the exponential growth $N \sim 2^L$ that would arise from naive iteration of the coherent-state calculus
(Section~\ref{sec:coherent-state-calculus}).

To this end, we introduce an \emph{approximating circuit} in which each exact Kerr gate is replaced by
the two-component approximation from Eq.~\eqref{eq:psi2-theta-def}, with the phase chosen according to
Corollary~\ref{cor:L1-kerr-coherent-optimal-theta}, while all DPLO layers are kept unchanged. More precisely, for a coherent input \(\ket{\alpha}\) with \(\lambda:=|\alpha|^2\), we define the local approximating Kerr map by
\begin{equation}
\label{eq:Ktilde-kappa-opt}
\widetilde{K}_{\kappa}\ket{\alpha}
:=
\ket{\psi_2(\theta_{\mathrm{opt}})}
=
A\ket{\alpha}
+
B\ket{-\alpha},
\end{equation}
with
\begin{equation}
\label{eq:AB-kappa-opt}
A:=\frac{1+e^{i\theta_{\mathrm{opt}}}}{2},
\qquad
B:=\frac{1-e^{i\theta_{\mathrm{opt}}}}{2},
\end{equation}
where \(\theta_{\mathrm{opt}}\) is given by Eq.~\eqref{eq:theta-opt-def}. Thus, the exact and approximating circuits differ only in the local Kerr update. In layer
$q\in\{1,\dots,L\}$ we first apply a DPLO unitary $U^{(q)}_{\mathrm{DPLO}}$, and then apply either the exact Kerr
gate $K_\kappa$ or its two-component approximation $\widetilde{K}_\kappa$ on mode~$1$, see Eqs.~\eqref{eq:small-kappa-layered-circuit-approx}, \eqref{eq:small-kappa-layered-circuit-true}.

Since our error metric will be the Hilbert-space $2$-norm distance between the final states, any common
terminal unitary may be omitted from the loss. Thus, if both output states are
post-composed with the same Gaussian unitary $G$, then unitarity implies
\begin{equation}
\big\|G\ket{\phi}-G\ket{\varphi}\big\|_2
=
\big\|\ket{\phi}-\ket{\varphi}\big\|_2.
\end{equation}
Therefore, for the purpose of the approximation analysis, it is sufficient to compare the two circuits
before the final Gaussian gate $G$; see Fig.~\ref{fig:circuit_scheme}.

We define the depth-$L$ approximating evolution by
\begin{equation}
\label{eq:small-kappa-layered-circuit-approx}
\ket{\psi_{\mathrm{approx}}}
:=
\Bigg[\prod_{q=1}^{L}
\Big(
U^{(q)}_{\mathrm{DPLO}}\big(\widetilde{K}_{\kappa_q}\otimes I^{\otimes(m-1)}\big)\,
\Big)\Bigg]\ket{\psi_0},
\end{equation}
where $U^{(q)}_{\mathrm{DPLO}}$ is given in Eq.~\eqref{eq:U_DPLO}, and the local approximating Kerr map
\(\widetilde K_{\kappa_q}\) is defined by Eqs.~\eqref{eq:Ktilde-kappa-opt}--\eqref{eq:AB-kappa-opt}. Note that the parameters may vary from layer to layer. Moreover, the construction in Eq.~\eqref{eq:small-kappa-layered-circuit-approx} directly corresponds to the circuit shown in Fig.~\ref{fig:circuit_scheme}, since the initial DPLO layer $U_{\mathrm{DPLO}}$ simply maps the input coherent state to another coherent state and therefore introduces no additional conceptual ingredient.

For comparison, we define the corresponding exact evolution by replacing each approximate Kerr update
with the exact Kerr gate \eqref{eq:kerr_gate}:
\begin{equation}
\label{eq:small-kappa-layered-circuit-true}
\ket{\psi_{\mathrm{true}}}
:=
\Bigg[\prod_{q=1}^{L}
\Big(
U^{(q)}_{\mathrm{DPLO}}
\big(K_{\kappa_q}\otimes I^{\otimes(m-1)}\big)\,\Big)\Bigg]\ket{\psi_0}.
\end{equation}
In Section~\ref{sec:coherent-state-calculus} we demonstrated that such a circuit corresponds to the universal bosonic CV computation.

To quantify the approximation quality, we consider the Hilbert-space $2$-norm distance
\begin{equation}
\label{eq:small-kappa-global-error-def}
\varepsilon
:=
\big\|\ket{\psi_{\mathrm{true}}}-\ket{\psi_{\mathrm{approx}}}\big\|_2
=
\sqrt{2-2\Re\!\left(\braket{\psi_{\mathrm{true}}}{\psi_{\mathrm{approx}}}\right)}.
\end{equation}
The next step is to control this global error by propagating the optimized single-gate estimate from Corollary~\ref{cor:L1-kerr-coherent-optimal-theta} through the branching structure generated by repeated Kerr applications.

The error bound depends on the largest coherent amplitude encountered on the Kerr mode immediately
\emph{before} a Kerr gate is applied. Let $\alpha^{(1)}_{i,q}$ denote the coherent amplitude on mode~$1$
of branch $i$ after applying $U^{(q)}_{\mathrm{DPLO}}$ (and any intermediate swaps), but before the Kerr
update in layer $q$. We define
\begin{equation}
\lambda_{\max}
:=
\max_{q\in\{1,\dots,L\}}
\ \max_{i\in\{1,\dots,N_{q-1}\}}
\big|\alpha^{(1)}_{i,q}\big|^2,
\label{eq:lambda-max}
\end{equation}
that is, the maximal mean photon number on the Kerr mode over all branches appearing immediately
prior to any Kerr application.

For the Kerr strengths, we similarly define
\begin{equation}\label{eq:kappa max}
\kappa_{\max}:=\max_{q}\kappa_q.
\end{equation}
Using Corollary~\ref{cor:L1-kerr-coherent-optimal-theta}, the global error
\eqref{eq:small-kappa-global-error-def} will be bounded by evaluating the optimized single-gate estimate
at the worst-case choice obtained by replacing every $\kappa_q$ in
Eqs.~\eqref{eq:small-kappa-layered-circuit-approx} and \eqref{eq:small-kappa-layered-circuit-true}
by $\kappa_{\max}$ and every pre-Kerr value of $\lambda$ by $\lambda_{\max}$.

For simplicity of notation, we first present the construction for a single input coherent-product term, i.e., we assume $N_0=1$. The extension to a general input of the form \eqref{eq:coherent_superposition_background} with arbitrary $N$ can be constructed using the linearity of the propagation through the unitaries in \eqref{eq:small-kappa-layered-circuit-approx}, \eqref{eq:small-kappa-layered-circuit-true}.

Under the worst-case substitution described above, repeated application of the local Kerr update
\eqref{eq:Ktilde-kappa-opt} within the approximating circuit \eqref{eq:small-kappa-layered-circuit-approx}
shows that, after \(\ell\) Kerr--DPLO layers, the state takes the form
\begin{equation}\label{eq:lth layer state approx}
\ket{\psi_{\mathrm{approx}}^{(\ell)}}
=
\sum_{q=0}^{\ell}
\sum_{k=1}^{\binom{\ell}{q}}
A^{\ell-q}B^q
\ket{\vec{\alpha}_{k,q}},
\qquad \ell=0,\dots,L,
\end{equation}
where \(A\) and \(B\) are given in Eq.~\eqref{eq:AB-kappa-opt}, and each
\(\ket{\vec{\alpha}_{k,q}}\) is a coherent-product state of the form
\begin{equation}\label{eq:coherent-product-vec-alpha}
\ket{\vec{\alpha}_{k,q}}
:=
\ket{\alpha^{(1)}_{k,q}}
\otimes
\cdots
\otimes
\ket{\alpha^{(m)}_{k,q}},
\qquad
\vec{\alpha}_{k,q}
:=
\big(\alpha^{(1)}_{k,q},\dots,\alpha^{(m)}_{k,q}\big)\in\mathbb{C}^m.
\end{equation}
Here the index \(q\) counts the number of times the \(B\)-branch is chosen among the first \(\ell\) Kerr updates, while \(k\) labels the resulting branches within that sector.

\begin{remark}[Independence of $\vec\alpha$]\label{remark:independet alpha}
In Eq.~\eqref{eq:lth layer state approx} we treat the coherent-amplitude vectors $\vec{\alpha}_{k,q}$ as mutually independent labels and do not assume any additional structure relating different branches. This is natural in the present framework, since for the DPLO Hamiltonian $H_{\mathrm{DPLO}}$ in Eq.~\eqref{eq:H_DPLO} no special assumptions are imposed on its energy coefficients, and in addition the circuit may include mode-swap operations. Therefore, after successive Kerr--DPLO layers, different branches may generate arbitrary coherent-amplitude vectors, and it is both convenient and sufficiently general to regard each $\vec{\alpha}_{k,q}$ in Eq.~\eqref{eq:lth layer state approx} as an independent parameter.
\end{remark}

One can note that using the two-component Kerr approximation \eqref{eq:Ktilde-kappa-opt} alone does not remove the exponential growth of the tree. Each application of \(\widetilde K\) still splits every incoming branch into two descendants, so naive propagation of the approximating circuit \eqref{eq:small-kappa-layered-circuit-approx} yields a binary branching at each Kerr layer. Therefore, after \(L\) layers the representation \eqref{eq:cs-ansatz} contains
\begin{equation}\label{eq:naive-term-count}
N = 2^L
\end{equation}
coherent-product terms. This leads to exponential computational cost in \(L\) for state propagation and, consequently, for the readout procedure (Section~\ref{sec:readout}).

For this reason, on top of the two-component approximation \eqref{eq:Ktilde-kappa-opt}, we introduce an online truncation strategy, illustrated in Fig.~\ref{fig:propogation_small_kerr_intro}. Let \(s\in\mathbb{N}\) be a fixed truncation parameter. During the propagation, whenever a branch acquires a factor \(B^{s+1}\), that branch is discarded and is no longer propagated. Equivalently, after each layer we retain only the sectors with \(q\le s\). The resulting truncated, generally unnormalized, state after \(\ell\) layers is
\begin{equation}\label{eq:lth layer state online}
\psi_{\mathrm{online}}^{(\ell)}
=
\sum_{q=0}^{\min(\ell,s)}
\sum_{k=1}^{\binom{\ell}{q}}
A^{\ell-q}B^q
\ket{\vec{\alpha}_{k,q}},
\qquad \ell=0,\dots,L,
\end{equation}
in contrast to the unpruned approximating state \(\ket{\psi_{\mathrm{approx}}^{(\ell)}}\) in Eq.~\eqref{eq:lth layer state approx}. Since only the first \(s+1\) sectors in the \(q\)-grading are kept, the number of propagated terms grows polynomially with the depth \(L\), i.e., \(N\sim \mathrm{poly}(L)\). The discarded residual is therefore
\begin{equation}\label{eq:approx-minus-online}
\ket{\psi_{\mathrm{approx}}^{(\ell)}}-\psi_{\mathrm{online}}^{(\ell)}
=
\sum_{q=s+1}^{\ell}
\sum_{k=1}^{\binom{\ell}{q}}
A^{\ell-q}B^q
\ket{\vec{\alpha}_{k,q}},
\qquad \ell=s+1,\dots,L.
\end{equation}

Let us now quantify the total error accumulated over $L$ layers when using both the two-component Kerr approximation \eqref{eq:Ktilde-kappa-opt} and the online-truncation strategy. To keep the notation compact, we define the exact and approximate one-layer update operators
\begin{equation}\label{eq:Uq-Uqtilde-def}
U_q := U^{(q)}_{\mathrm{DPLO}} \big(K_{\kappa_q} \otimes I^{\otimes(m-1)}\big),
\qquad
\widetilde U_q := U^{(q)}_{\mathrm{DPLO}} \big(\widetilde K_{\kappa_q} \otimes I^{\otimes(m-1)}\big),
\qquad q=1,\dots,L,
\end{equation}
where \(K_{\kappa_q}\) and \(\widetilde K_{\kappa_q}\) are defined in Eqs.~\eqref{eq:kerr_gate} and \eqref{eq:Ktilde-kappa-opt}, respectively. Using Corollary~\ref{cor:L1-kerr-coherent-optimal-theta}, together with the worst-case substitutions \(\kappa_{\max}\) and \(\lambda_{\max}\), we define the single-branch Kerr error
\begin{equation}\label{eq:delta-def-final}
\delta
:=
\kappa_{\max}
\sqrt{
\lambda_{\max}^4+6\lambda_{\max}^3+7\lambda_{\max}^2+\lambda_{\max}
-\frac{\Bigl((\lambda_{\max}^2+\lambda_{\max})-e^{-2\lambda_{\max}}(\lambda_{\max}^2-\lambda_{\max})\Bigr)^2}
{2\bigl(1-e^{-2\lambda_{\max}}\bigr)}
}.
\end{equation}

Next, we define the truncation operator \(R_s\) that implements the online truncation strategy. Let a state after \(\ell\) approximate layers be written as
\begin{equation}\label{eq:phi-layers-full}
\phi^{(\ell)} = \sum_{q=0}^{\ell} \sum_{k=1}^{\binom{\ell}{q}} A^{\ell-q} B^q \ket{\vec\alpha_{k,q}},
\end{equation}
where \(A,B\) are given in Eq.~\eqref{eq:AB-kappa-opt}. The operator \(R_s\) acts by removing all sectors with \(q \ge s+1\) while leaving the remaining sectors unchanged:
\begin{equation}\label{eq:Rs-def-full}
R_s[\phi^{(\ell)}] := \sum_{q=0}^{\min(\ell,s)} \sum_{k=1}^{\binom{\ell}{q}} A^{\ell-q} B^q \ket{\vec\alpha_{k,q}}.
\end{equation}
Equivalently, \(R_s\) drops all branches carrying a factor \(B^{s+1}\) and keeps all terms with \(q\le s\) untouched.

The truncated state after \(\ell\) layers is therefore
\begin{equation}\label{eq:psi-online-recursive-full}
\psi_{\mathrm{online}}^{(\ell)} 
:=
\Pi_{q=1}^{\ell} R_s \widetilde U_q \ket{\psi_0},
\qquad \ell = 0,\dots,L,
\end{equation}
where \(\widetilde U_q\) is defined in Eq.~\eqref{eq:Uq-Uqtilde-def}.

We now define the total error after \(L\) layers as the Hilbert-space distance between the exact state and the truncated state,
\begin{equation}\label{eq:epsilon-online-start}
\varepsilon_{\mathrm{online}}
:=
\left\|
\ket{\psi_{\mathrm{true}}^{(L)}}-\psi_{\mathrm{online}}^{(L)}
\right\|_2
=
\left\|
\prod_{q=1}^{L} U_q \ket{\psi_0}
-
\Pi_{q=1}^{L} R_s \widetilde U_q \ket{\psi_0}
\right\|_2,
\end{equation}
where \(U_q\) and \(\widetilde U_q\) are defined in Eq.~\eqref{eq:Uq-Uqtilde-def}, while \(\psi_{\mathrm{online}}^{(L)}\) is given by Eq.~\eqref{eq:psi-online-recursive-full}. By adding and subtracting intermediate terms, we obtain the telescopic decomposition
\begin{equation}\label{eq:online-telescopic-full-final}
\begin{aligned}
\prod_{q=1}^{L} U_q \ket{\psi_0}
-
\Pi_{q=1}^{L} R_s \widetilde U_q \ket{\psi_0}
&=
\left(\prod_{q=2}^{L} U_q\right)\big(U_1 - R_s\widetilde U_1\big)\ket{\psi_0}
\\
&\quad+
\left(\prod_{q=3}^{L} U_q\right)\big(U_2 - R_s\widetilde U_2\big)R_s\widetilde U_1\ket{\psi_0}
+\cdots
\\
&\quad+
\big(U_L - R_s\widetilde U_L\big)\Pi_{q=1}^{L-1} R_s\widetilde U_q \ket{\psi_0}
\\
&=
\sum_{\ell=0}^{L-1}
\left( \prod_{q=\ell+2}^{L} U_q \right)
\left( U_{\ell+1} - R_s \widetilde U_{\ell+1} \right)
\Pi_{q=1}^{\ell} R_s \widetilde U_q \ket{\psi_0}.
\end{aligned}
\end{equation}

Taking the Hilbert-space norm and using the triangle inequality gives
\begin{equation}\label{eq:online-telescopic-norm-pre-unitary}
\varepsilon_{\mathrm{online}}
\le
\sum_{\ell=0}^{L-1}
\left\|
\left( \prod_{q=\ell+2}^{L} U_q \right)
\left( U_{\ell+1} - R_s \widetilde U_{\ell+1} \right)
\Pi_{q=1}^{\ell} R_s \widetilde U_q \ket{\psi_0}
\right\|_2.
\end{equation}
Since each \(U_q\) is unitary, the left factor does not change the norm, and therefore
\begin{equation}\label{eq:online-telescopic-norm-final-full}
\begin{aligned}
\varepsilon_{\mathrm{online}}
&\le
\sum_{\ell=0}^{L-1}
\left\|
\left( U_{\ell+1} - R_s \widetilde U_{\ell+1} \right)
\Pi_{q=1}^{\ell} R_s \widetilde U_q \ket{\psi_0}
\right\|_2
\\
&=
\sum_{\ell=0}^{L-1}
\left\|
\left( U_{\ell+1} - R_s \widetilde U_{\ell+1} \right)
\psi_{\mathrm{online}}^{(\ell)}
\right\|_2
\\
&=
\sum_{\ell=0}^{L-1}
\left\|
\left(
\big(U_{\ell+1}-\widetilde U_{\ell+1}\big)
+
\big(I-R_s\big)\widetilde U_{\ell+1}
\right)
\psi_{\mathrm{online}}^{(\ell)}
\right\|_2
\\
&\le
\sum_{\ell=0}^{L-1}
\left\|
\big(U_{\ell+1}-\widetilde U_{\ell+1}\big)
\psi_{\mathrm{online}}^{(\ell)}
\right\|_2
+
\sum_{\ell=0}^{L-1}
\left\|
\big(I-R_s\big)\widetilde U_{\ell+1}
\psi_{\mathrm{online}}^{(\ell)}
\right\|_2.
\end{aligned}
\end{equation}

To estimate the first sum in Eq.~\eqref{eq:online-telescopic-norm-final-full}, we substitute the layer-\(\ell\) truncated state from Eq.~\eqref{eq:lth layer state online}:
\begin{equation}\label{eq:online-first-sum-substitution}
\begin{aligned}
&\sum_{\ell=0}^{L-1}
\left\|
\big(U_{\ell+1}-\widetilde U_{\ell+1}\big)
\psi_{\mathrm{online}}^{(\ell)}
\right\|_2
\\
&=
\sum_{\ell=0}^{L-1}
\left\|
\big(U_{\ell+1}-\widetilde U_{\ell+1}\big)
\sum_{q=0}^{\min(\ell,s)}
\sum_{k=1}^{\binom{\ell}{q}}
A^{\ell-q}B^q
\ket{\vec{\alpha}_{k,q}}
\right\|_2
\\
&\le
\sum_{\ell=0}^{L-1}
\sum_{q=0}^{\min(\ell,s)}
\sum_{k=1}^{\binom{\ell}{q}}
|A|^{\ell-q}|B|^q
\left\|
\big(U_{\ell+1}-\widetilde U_{\ell+1}\big)
\ket{\vec{\alpha}_{k,q}}
\right\|_2.
\end{aligned}
\end{equation}
By the definitions in Eq.~\eqref{eq:Uq-Uqtilde-def} and unitarity of \(U_{\mathrm{DPLO}}^{(\ell+1)}\), each norm reduces to the corresponding single-mode Kerr approximation error on the first mode:
\begin{equation}\label{eq:online-first-sum-kerr-reduction}
\left\|
\big(U_{\ell+1}-\widetilde U_{\ell+1}\big)
\ket{\vec{\alpha}_{k,q}}
\right\|_2
=
\left\|
\big(K_{\kappa_{\ell+1}}-\widetilde K_{\kappa_{\ell+1}}\big)
\ket{\alpha^{(1)}_{k,q}}
\right\|_2.
\end{equation}
Using the definition of \(\delta\) in Eq.~\eqref{eq:delta-def-final}, together with the worst-case substitutions \(\kappa_{\max}\) and \(\lambda_{\max}\) from Eqs.~\eqref{eq:lambda-max} and \eqref{eq:kappa max}, we obtain
\begin{equation}\label{eq:online-first-sum-delta}
\left\|
\big(K_{\kappa_{\ell+1}}-\widetilde K_{\kappa_{\ell+1}}\big)
\ket{\alpha^{(1)}_{k,q}}
\right\|_2
\le
\delta.
\end{equation}
Therefore,
\begin{equation}\label{eq:online-first-sum-bound}
\sum_{\ell=0}^{L-1}
\left\|
\big(U_{\ell+1}-\widetilde U_{\ell+1}\big)
\psi_{\mathrm{online}}^{(\ell)}
\right\|_2
\le
\delta
\sum_{\ell=0}^{L-1}
\sum_{q=0}^{\min(\ell,s)}
\binom{\ell}{q}|A|^{\ell-q}|B|^q.
\end{equation}

We now estimate the second sum in Eq.~\eqref{eq:online-telescopic-norm-final-full}. For \(\ell<s\), the state \(\psi_{\mathrm{online}}^{(\ell)}\) contains no sector with \(q=s\), and hence after one further application of \(\widetilde U_{\ell+1}\) no term with power \(B^{s+1}\) can be generated. Therefore,
\begin{equation}\label{eq:online-second-sum-zero}
\big(I-R_s\big)\widetilde U_{\ell+1}\psi_{\mathrm{online}}^{(\ell)}=0,
\qquad \ell=0,\dots,s-1.
\end{equation}
For every \(\ell\ge s\), only the sector \(q=s\) in Eq.~\eqref{eq:lth layer state online} can generate discarded terms. By Eq.~\eqref{eq:Ktilde-kappa-opt}, each branch in the sector \(q=s\) produces one descendant proportional to \(A\), which remains in the retained sector \(q=s\), and one descendant proportional to \(B\), which enters the discarded sector \(q=s+1\). Hence,
\begin{equation}\label{eq:online-second-sum-explicit}
\big(I-R_s\big)\widetilde U_{\ell+1}\psi_{\mathrm{online}}^{(\ell)}
=
\sum_{k=1}^{\binom{\ell}{s}}
A^{\ell-s}B^{s+1}
\ket{\vec{\beta}_{k,s}^{(\ell+1)}},
\qquad \ell=s,\dots,L-1,
\end{equation}
for suitable coherent-product states \(\ket{\vec{\beta}_{k,s}^{(\ell+1)}}\), each of unit norm. Therefore, by the triangle inequality,
\begin{equation}\label{eq:online-second-sum-bound-step}
\begin{aligned}
\left\|
\big(I-R_s\big)\widetilde U_{\ell+1}\psi_{\mathrm{online}}^{(\ell)}
\right\|_2
&=
\left\|
\sum_{k=1}^{\binom{\ell}{s}}
A^{\ell-s}B^{s+1}
\ket{\vec{\beta}_{k,s}^{(\ell+1)}}
\right\|_2
\\
&\le
\sum_{k=1}^{\binom{\ell}{s}}
|A|^{\ell-s}|B|^{s+1}
\left\|
\ket{\vec{\beta}_{k,s}^{(\ell+1)}}
\right\|_2
\\
&=
\binom{\ell}{s}|A|^{\ell-s}|B|^{s+1}.
\end{aligned}
\end{equation}
Summing over \(\ell\) and using Eq.~\eqref{eq:online-second-sum-zero}, we find
\begin{equation}\label{eq:online-second-sum-bound}
\sum_{\ell=0}^{L-1}
\left\|
\big(I-R_s\big)\widetilde U_{\ell+1}\psi_{\mathrm{online}}^{(\ell)}
\right\|_2
\le
\sum_{\ell=s}^{L-1}
\binom{\ell}{s}|A|^{\ell-s}|B|^{s+1}.
\end{equation}

Combining Eqs.~\eqref{eq:online-telescopic-norm-final-full}, \eqref{eq:online-first-sum-bound}, and \eqref{eq:online-second-sum-bound}, we have
\begin{equation}\label{eq:true-vs-onlinepruned-final-expression}
\left\|
\ket{\psi^{(L)}_{\mathrm{true}}}-\psi_{\mathrm{online}}^{(L)}
\right\|_2
\le
\delta
\sum_{\ell=0}^{L-1}
\sum_{q=0}^{\min(\ell,s)}
\binom{\ell}{q}|A|^{\ell-q}|B|^q
+
\sum_{\ell=s}^{L-1}\binom{\ell}{s}|A|^{\ell-s}|B|^{s+1}.
\end{equation}

Equation~\eqref{eq:true-vs-onlinepruned-final-expression} cleanly separates the two sources of error in the online scheme. The first term is the accumulated \emph{local Kerr-approximation error}: at each layer, every retained branch with $q\le s$ is propagated with $\widetilde K_{\kappa}$ instead of the exact Kerr gate $K_{\kappa}$, and the resulting mismatch is bounded branchwise by $\delta$. This explains the factor
\begin{equation}
\sum_{\ell=0}^{L-1}\sum_{q=0}^{\min(\ell,s)}\binom{\ell}{q}|A|^{\ell-q}|B|^q,
\end{equation}
which counts all retained branches weighted by their amplitudes. The second term is purely a \emph{truncation error}: it comes from discarding, during propagation, the descendants of the sector $q=s$ that acquire one additional factor $B$, namely the branches proportional to $B^{s+1}$.

We now collect the preceding ingredients into the main result of this section. Theorem~\ref{thm:online-pruned-small-kappa} combines the optimized single-gate approximation from Corollary~\ref{cor:L1-kerr-coherent-optimal-theta}, the asymptotic estimates \eqref{eq:L1-optimal-small-lambda-norm} and \eqref{eq:L1-optimal-large-lambda-norm}, and the truncation estimate in Eq.~\eqref{eq:true-vs-onlinepruned-final-expression}. It identifies an explicit small-non-Gaussianity regime in which coherent-state propagation through $L$ interleaved Kerr--DPLO layers remains polynomial, with representation size $N=\mathcal{O}(L^s)$ and controlled Hilbert-space error.

\begin{theorem}[Polynomial-time coherent-state propagation in the small non-Gaussianity regime]
\label{thm:online-pruned-small-kappa}
Consider the depth-$L$ Kerr--DPLO circuit of Fig.~\ref{fig:circuit_scheme} with $m$ bosonic modes, with Kerr gate $K_{\kappa}$ and DPLO unitary $U_{\mathrm{DPLO}}$ as in Eqs.~\eqref{eq:kerr_gate}, \eqref{eq:U_DPLO}, and let the exact output state $\ket{\psi_\text{true}}$ be defined by Eq.~\eqref{eq:small-kappa-layered-circuit-true}. Let $\lambda$ denote the maximal mean photon number in the first mode and let $\kappa$ denote the maximal per-gate Kerr strength, with these quantities defined precisely as in Eqs.~\eqref{eq:lambda-max} and \eqref{eq:kappa max}. Define
\begin{equation}
\label{eq:f-lambda-def}
f(\lambda)
:=
\frac{(\lambda^2+\lambda)-e^{-2\lambda}(\lambda^2-\lambda)}
{1-e^{-2\lambda}}.
\end{equation}
Fix truncation cutoff parameter $s\in\mathbb{N}$ with $s\leq L$, and assume that the Kerr strength lies in the small-nonlinearity regime
\begin{equation}
\label{eq:thm-online-small-kappa-scaling}
\kappa \sim
\begin{cases}
\dfrac{1}{L\lambda^{2}}, & \lambda>1,\\[8pt]
\dfrac{1}{L\lambda}, & 0< \lambda\le 1.
\end{cases}
\end{equation}
Assume moreover that
\begin{equation}
\label{eq:thm-online-small-kappa-denominator}
eL\left|\tan\!\left(\frac{\kappa f(\lambda)}{2}\right)\right|<1.
\end{equation}
Then runtime
\begin{equation}
\label{eq:thm-online-small-kappa-runtime}
t\in\mathcal{O}(Lm^3+m^2L^{s+1}),
\end{equation}
suffices to compute a state of the coherent-product form \eqref{eq:cs-ansatz},
\begin{equation}
\label{eq:thm-online-small-kappa-output}
\psi_{\mathrm{online}}^{(L)}
=
\sum_{k=1}^{N} C_k
\ket{\alpha_k^{(1)}}\otimes\ket{\alpha_k^{(2)}}\otimes\cdots\otimes\ket{\alpha_k^{(m)}},
\end{equation}
with
\begin{equation}
\label{eq:thm-online-small-kappa-term-count}
N\le \sum_{q=0}^{s}\binom{L}{q}=\mathcal{O}(L^s).
\end{equation}
Such a state satisfies the approximation error bound
\begin{equation}
\label{eq:thm-online-small-kappa-error}
\epsilon
:=
\left\|
\ket{\psi_{\mathrm{true}}^{(L)}}-\psi_{\mathrm{online}}^{(L)}
\right\|_2
\le
E_0
+
\frac{e^s}{s^s}\left(\frac{\kappa f(\lambda)L}{2}\right)^{s+1}
\left(1-\frac{s}{L}\right),
\end{equation}
where
\begin{equation}
\label{eq:thm-online-small-kappa-E0}
E_0
:=
\frac{
\kappa L\,
\sqrt{
\lambda^4+6\lambda^3+7\lambda^2+\lambda
-\dfrac{1-e^{-2\lambda}}{2}\,f(\lambda)^2
}
}{
1-eL\left|\tan\!\left(\frac{\kappa f(\lambda)}{2}\right)\right|
}.
\end{equation}
Moreover, for any target accuracy $\epsilon>E_0$, it is sufficient to choose the truncation cutoff $s$ such that
\begin{equation}
\label{eq:s-loglog-scaling-final-eps}
s
=
\Theta\left(
\frac{\log\left(\frac{1}{\epsilon-E_0}\right)}
{\log\log\!\big(\frac{1}{\epsilon-E_0}\big)}\right).
\end{equation}
\end{theorem}
\begin{proof}
We first bound the Kerr-replacement contribution in Eq.~\eqref{eq:true-vs-onlinepruned-final-expression}. For the combinatorial estimates we use the standard binomial inequalities collected in Appendix~A.4 of Ref.~\cite{lerch2024efficient}. We obtain
\begin{equation}
\label{eq:proof-thm-online-double-sum-geometric}
\begin{aligned}
\sum_{\ell=0}^{L-1}\sum_{q=0}^{\min(\ell,s)}
\binom{\ell}{q}|A|^{\ell-q}|B|^q
&\le
\sum_{\ell=0}^{L-1}|A|^\ell
\sum_{q=0}^{s}\binom{\ell}{q}\left(\frac{|B|}{|A|}\right)^q \\
&\le
\sum_{\ell=0}^{L-1}|A|^\ell
\sum_{q=0}^{s}\left(\frac{eL|B|}{q|A|}\right)^q \\
&\le
\sum_{\ell=0}^{L-1}|A|^\ell
\sum_{q=0}^{s}\left(\frac{eL|B|}{|A|}\right)^q .
\end{aligned}
\end{equation}
Here, in the first line we used $\min(\ell,s)\le s$ and factored out $|A|^\ell$. In the second line we applied the binomial estimate $\binom{\ell}{q}\le (e\ell/q)^q\le (eL/q)^q$. In the third line we used $q^{-q}\le 1$ for $q\ge 1$, while the $q=0$ term is understood as $1$.

Under the assumption \eqref{eq:thm-online-small-kappa-denominator}, the ratio
\begin{equation}
\frac{eL|B|}{|A|}
=
eL\left|\tan\!\left(\frac{\kappa f(\lambda)}{2}\right)\right|
\end{equation}
is strictly smaller than $1$, so
\begin{equation}
\label{eq:proof-thm-online-double-sum-final}
\sum_{\ell=0}^{L-1}\sum_{q=0}^{\min(\ell,s)}
\binom{\ell}{q}|A|^{\ell-q}|B|^q
\le
\sum_{\ell=0}^{L-1}|A|^\ell\,
\frac{1}{1-\dfrac{eL|B|}{|A|}}
\le
\frac{L}{1-\dfrac{eL|B|}{|A|}},
\end{equation}
where in the last step we used $|A|\le 1$.

We next bound the truncation contribution in Eq.~\eqref{eq:true-vs-onlinepruned-final-expression}
\begin{equation}
\label{eq:proof-thm-online-pruning-bound}
\begin{aligned}
\sum_{\ell=s}^{L-1}\binom{\ell}{s}|A|^{\ell-s}|B|^{s+1}
&=
|B|^{s+1}\sum_{\ell=s}^{L-1}\binom{\ell}{s}|A|^{\ell-s} \\
&\le
|B|^{s+1}\sum_{\ell=s}^{L-1}\left(\frac{e\ell}{s}\right)^s |A|^{\ell-s} \\
&\le
|B|^{s+1}\left(\frac{eL}{s}\right)^s\sum_{\ell=s}^{L-1}|A|^{\ell-s} \\
&\le
\frac{e^s}{s^s}\,|B|^{s+1}L^{s+1}\left(1-\frac{s}{L}\right).
\end{aligned}
\end{equation}
Here, in the second line we used $\binom{\ell}{s}\le (e\ell/s)^s$. In the third line we used $\ell\le L$ for all $\ell=s,\dots,L-1$. In the last line we used $|A|\le 1$, so that $\sum_{\ell=s}^{L-1}|A|^{\ell-s}\le L-s$.

Substituting Eqs.~\eqref{eq:proof-thm-online-double-sum-final} and \eqref{eq:proof-thm-online-pruning-bound} into Eq.~\eqref{eq:true-vs-onlinepruned-final-expression}, then using Eqs.~\eqref{eq:delta-def-final} and \eqref{eq:AB-kappa-opt}, together with
\begin{equation}
\theta_{\mathrm{opt}}=\kappa f(\lambda),
\qquad
|B|=\left|\sin\!\left(\frac{\kappa f(\lambda)}{2}\right)\right|,
\qquad
\frac{|B|}{|A|}=\left|\tan\!\left(\frac{\kappa f(\lambda)}{2}\right)\right|,
\end{equation}
and finally applying $\sin x\le x$ for $x\ge 0$, we obtain
\begin{equation}
\label{eq:proof-thm-online-combined-explicit}
\begin{aligned}
\left\|
\ket{\psi_{\mathrm{true}}^{(L)}}-\psi_{\mathrm{online}}^{(L)}
\right\|_2
\le\;&
\delta\,
\frac{L}{1-\dfrac{eL|B|}{|A|}}
+
\frac{e^s}{s^s}\,|B|^{s+1}L^{s+1}\left(1-\frac{s}{L}\right)
\\[4pt]
=\;&
\kappa
\sqrt{
\lambda^4+6\lambda^3+7\lambda^2+\lambda
-\dfrac{1-e^{-2\lambda}}{2}\,f(\lambda)^2
}
\,
\frac{L}{1-eL\left|\tan\!\left(\frac{\kappa f(\lambda)}{2}\right)\right|}
\\
&\;+
\frac{e^s}{s^s}\left|\sin\!\left(\frac{\kappa f(\lambda)}{2}\right)\right|^{s+1}L^{s+1}\left(1-\frac{s}{L}\right)
\\[4pt]
\le\;&
\frac{
\kappa L\,
\sqrt{
\lambda^4+6\lambda^3+7\lambda^2+\lambda
-\dfrac{1-e^{-2\lambda}}{2}\,f(\lambda)^2
}
}{
1-eL\left|\tan\!\left(\frac{\kappa f(\lambda)}{2}\right)\right|
}
+
\frac{e^s}{s^s}\left(\frac{\kappa f(\lambda)L}{2}\right)^{s+1}\left(1-\frac{s}{L}\right).
\end{aligned}
\end{equation}
This is exactly Eq.~\eqref{eq:thm-online-small-kappa-error}.

It remains to justify the scaling in Eq.~\eqref{eq:thm-online-small-kappa-scaling}. Here we use the asymptotic estimates \eqref{eq:L1-optimal-small-lambda-norm} and \eqref{eq:L1-optimal-large-lambda-norm}. The key point is that in both terms of Eq.~\eqref{eq:thm-online-small-kappa-error} the depth $L$ appears together with the Kerr strength through the combination $\kappa L$.

More precisely, for $\lambda>1$, Eq.~\eqref{eq:L1-optimal-large-lambda-norm} shows that the relevant single-gate contribution is governed by a factor of order $\lambda^2$, so choosing $\kappa\sim 1/(L\lambda^2)$ keeps the expression bounded uniformly in $L$. For $0<\lambda\le 1$ Eq.~\eqref{eq:L1-optimal-small-lambda-norm} yields $\kappa\sim1/(L\lambda)$. These are exactly the regimes stated in Eq.~\eqref{eq:thm-online-small-kappa-scaling}.

We next analyze the dependence of the truncation cutoff $s$ on the target accuracy. From Eq.~\eqref{eq:thm-online-small-kappa-error}, for any prescribed tolerance $\epsilon>E_0$ it is sufficient to require
\begin{equation}
\label{eq:proof-thm-online-s-condition}
\frac{e^s}{s^s}\left(\frac{\kappa f(\lambda)L}{2}\right)^{s+1}\left(1-\frac{s}{L}\right)
\le
\epsilon-E_0.
\end{equation}
Using $1-\frac{s}{L}\le 1$ we have
\begin{equation}
\label{eq:proof-thm-online-s-condition-simplified}
\frac{\kappa f(\lambda)L}{2}
\left(\frac{e\kappa f(\lambda)L}{2s}\right)^s
\le
\epsilon-E_0.
\end{equation}
Set
\begin{equation}
a:=\frac{e\kappa f(\lambda)L}{2},
\qquad
c:=\frac{\kappa f(\lambda)L}{2},
\qquad
\Delta_\epsilon:=\epsilon-E_0.
\end{equation}
Then it is enough to solve
\begin{equation}
\label{eq:proof-thm-online-s-lambert-start}
c\left(\frac{a}{s}\right)^s\le \Delta_\epsilon,
\end{equation}
or equivalently
\begin{equation}
\label{eq:proof-thm-online-s-lambert-log}
s\log\!\left(\frac{s}{a}\right)\ge \log\!\left(\frac{c}{\Delta_\epsilon}\right).
\end{equation}
Writing $M:=\log\!\left(\frac{c}{\Delta_\epsilon}\right)$ and setting $s=a e^u$, we obtain $aue^u=M$, hence
$u=W(M/a)$, where $W$ denotes the Lambert $W$ function \cite{Corless1996}. Therefore
\begin{equation}
\label{eq:proof-thm-online-s-lambert-explicit}
s
=
a e^{W(M/a)}
=
\frac{M}{W(M/a)}.
\end{equation}
Since $\kappa f(\lambda)L$ is fixed in the regime of Eq.~\eqref{eq:thm-online-small-kappa-scaling}, both $a$ and $c$ are constants independent of $\epsilon$, so $M=\mathcal{O}\!\left(\log\!\bigl(\frac{1}{\epsilon-E_0}\bigr)\right)$. Using the standard asymptotic $W(x)\sim \log x$ as $x\to\infty$, we conclude that
\begin{equation}
\label{eq:proof-thm-online-s-asymptotic}
s
=
\Theta\left(
\frac{\log\left(\frac{1}{\epsilon-E_0}\right)}
{\log\log\!\big(\frac{1}{\epsilon-E_0}\big)}
\right),
\end{equation}
which is exactly Eq.~\eqref{eq:s-loglog-scaling-final-eps}.

Finally, for the runtime estimate, we use the DPLO update from Eqs.~\eqref{eq:H_DPLO}--\eqref{eq:U_DPLO}. The passive part is specified by the matrix $K\in\mathbb{C}^{m\times m}$, so for one DPLO layer the affine coherent-state update in Eq.~\eqref{eq:DPLO_affine} requires first computing the matrix exponential $e^{-it_qK^{(q)}}$, which costs $\mathcal{O}(m^3)$, and then applying the resulting affine map to each retained coherent-amplitude vector $\vec\alpha_{k,q}\in\mathbb{C}^m$, which costs $\mathcal{O}(m^2)$ per branch. Hence, if $N_q$ denotes the number of retained branches at depth $q$, the cost of one DPLO layer is
\begin{equation}
\label{eq:proof-thm-online-runtime-one-layer}
\mathcal{O}(m^3+m^2N_q).
\end{equation}
By Eq.~\eqref{eq:thm-online-small-kappa-term-count}, we have
\begin{equation}
\label{eq:proof-thm-online-runtime-branches}
N_q\le \sum_{r=0}^{s}\binom{q}{r}=\mathcal{O}(q^s),
\qquad q=0,\dots,L.
\end{equation}
Therefore the total propagation cost through all $L$ layers is
\begin{equation}
\label{eq:proof-thm-online-runtime-total}
\sum_{q=1}^{L}\mathcal{O}(m^3+m^2N_q)
=
\mathcal{O}\!\left(Lm^3+m^2\sum_{q=1}^{L}q^s\right)
=
\mathcal{O}(Lm^3+m^2L^{s+1}).
\end{equation}
If the DPLO matrix exponentials are precomputed once in advance and not counted in the online propagation cost, this reduces to $\mathcal{O}(m^2L^{s+1})$.
\end{proof}

In the following, we provide a compact corollary of Theorem~\ref{thm:online-pruned-small-kappa}. In particular, it isolates the two sources of approximation error: an $\mathcal{O}(\delta)$ term coming from the small-Kerr approximation itself, and an $\varepsilon$ term controlled algorithmically by the truncation threshold.

\begin{corollary}[Simple small-nonlinearity regime]
\label{cor:simple-small-kappa}
Assume the hypotheses of Theorem~\ref{thm:online-pruned-small-kappa}.
Suppose moreover that
\begin{equation}
\label{eq:simple-small-kappa-assumption}
\kappa
\leq \frac{\delta}{L(1+\lambda^2)},
\end{equation}
where the right-hand side is understood up to a sufficiently small absolute prefactor.

Then, for every $\varepsilon\in(0,1)$, there exists a truncation cutoff
\begin{equation}
s = \mathcal{O}(\log(1/\varepsilon))
\end{equation}
such that coherent-state propagation with truncation outputs a state
of the form \eqref{eq:thm-online-small-kappa-output} satisfying
\begin{equation}
\label{eq:simple-small-kappa-error}
\left\|
\ket{\psi_{\mathrm{true}}^{(L)}}-\ket{\psi_{\mathrm{online}}^{(L)}}
\right\|_2
\le
\mathcal{O}(\delta)+\varepsilon.
\end{equation}

Moreover, the number of retained coherent-product terms satisfies
\begin{equation}
\label{eq:simple-small-kappa-term-count}
N = L^{\,\mathcal{O}(\log(1/\varepsilon))},
\end{equation}
and the runtime satisfies
\begin{equation}
\label{eq:simple-small-kappa-runtime}
t
=
\mathcal{O}\!\left(
Lm^3 + m^2 L^{\,1+\mathcal{O}(\log(1/\varepsilon))}
\right).
\end{equation}

In particular, for fixed $\delta$ and constant $\varepsilon$, the runtime is
polynomial in $m$ and $L$.
\end{corollary}

\begin{proof}
We start from Theorem~\ref{thm:online-pruned-small-kappa}, which gives
\begin{equation}
\label{eq:cor-proof-starting-error}
\left\|
\ket{\psi_{\mathrm{true}}^{(L)}}-\ket{\psi_{\mathrm{online}}^{(L)}}
\right\|_2
\le
E_0
+
\frac{e^s}{s^s}\left(\frac{\kappa f(\lambda)L}{2}\right)^{s+1}
\left(1-\frac{s}{L}\right),
\end{equation}
together with the bounds
\begin{equation}
\label{eq:cor-proof-starting-runtime}
t\in\mathcal{O}(Lm^3+m^2L^{s+1}),
\qquad
N\le \sum_{q=0}^{s}\binom{L}{q}
=
\mathcal{O}(L^s).
\end{equation}

We first simplify the auxiliary function $f(\lambda)$ from
\eqref{eq:f-lambda-def}. A direct rearrangement gives
\begin{equation}
f(\lambda)
=
\lambda^2+\lambda\,\frac{1+e^{-2\lambda}}{1-e^{-2\lambda}}
=
\lambda^2+\lambda \coth \lambda.
\end{equation}
Using $e^{2\lambda}-1\ge 2\lambda$ for $\lambda>0$, we obtain
\begin{equation}
\coth\lambda
=
1+\frac{2}{e^{2\lambda}-1}
\le
1+\frac{1}{\lambda},
\end{equation}
and therefore
\begin{equation}
\label{eq:cor-proof-f-bound}
f(\lambda)\le \lambda^2+\lambda+1 \le 2(1+\lambda^2).
\end{equation}

Now let $c>0$ denote the sufficiently small absolute prefactor hidden in
\eqref{eq:simple-small-kappa-assumption}, i.e.
\begin{equation}
\kappa \le \frac{c\,\delta}{L(1+\lambda^2)}.
\end{equation}
Set
\begin{equation}
a:=\frac{\kappa f(\lambda)L}{2}.
\end{equation}
By \eqref{eq:cor-proof-f-bound},
\begin{equation}
a
\le
\frac{\kappa L}{2}\,2(1+\lambda^2)
=
\kappa L(1+\lambda^2)
\le
c\,\delta
\le c,
\end{equation}
since $\delta\in(0,1)$.

If $c$ is chosen sufficiently small, then $a\le 1/(4e)$. Since
$a=\kappa f(\lambda)L/2<1/2$, we may use $|\tan x|\le 2|x|$ for
$|x|\le 1/2$ and obtain
\begin{equation}
eL\left|\tan\!\left(\frac{\kappa f(\lambda)}{2}\right)\right|
\le
2e\,a
\le
\frac12.
\end{equation}
Hence the denominator in \eqref{eq:thm-online-small-kappa-E0} is bounded
from below by $1/2$.

Using now \eqref{eq:thm-online-small-kappa-E0}, together with the fact that
the subtracted term under the square root is nonnegative, we get
\begin{equation}
E_0
\le
\frac{
\kappa L\sqrt{\lambda^4+6\lambda^3+7\lambda^2+\lambda}
}{
1-eL\left|\tan\!\left(\frac{\kappa f(\lambda)}{2}\right)\right|
}.
\end{equation}
Moreover,
\begin{equation}
6\lambda^3\le 3\lambda^4+3\lambda^2,
\qquad
\lambda\le 1+\lambda^2,
\end{equation}
so
\begin{equation}
\lambda^4+6\lambda^3+7\lambda^2+\lambda
\le
4\lambda^4+11\lambda^2+1
\le
16(1+\lambda^2)^2.
\end{equation}
Therefore
\begin{equation}
\sqrt{\lambda^4+6\lambda^3+7\lambda^2+\lambda}
\le
4(1+\lambda^2),
\end{equation}
and hence
\begin{equation}
E_0
\le
8\,\kappa L(1+\lambda^2)
\le
8c\,\delta
=
\mathcal{O}(\delta).
\end{equation}
This proves that the first contribution to the total error is
$\mathcal{O}(\delta)$.

We next bound the truncation contribution in
\eqref{eq:cor-proof-starting-error}. Choose
\begin{equation}
s:=\min\!\left\{L,\left\lceil \log_4(1/\varepsilon)\right\rceil\right\}.
\end{equation}
If $s=L$, then
\begin{equation}
1-\frac{s}{L}=0,
\end{equation}
so the truncation contribution vanishes identically.

If instead $s<L$, then by the definition of $a$ and the bound $a\le 1/(4e)$,
\begin{equation}
\frac{e^s}{s^s}\left(\frac{\kappa f(\lambda)L}{2}\right)^{s+1}
\left(1-\frac{s}{L}\right)
\le
\frac{e^s}{s^s}a^{s+1}
\le
a(ea)^s
\le
\left(\frac14\right)^s
\le
\varepsilon.
\end{equation}
Thus, in all cases, the truncation contribution is at most $\varepsilon$.

Substituting these two estimates into
\eqref{eq:cor-proof-starting-error} yields
\begin{equation}
\left\|
\ket{\psi_{\mathrm{true}}^{(L)}}-\ket{\psi_{\mathrm{online}}^{(L)}}
\right\|_2
\le
\mathcal{O}(\delta)+\varepsilon,
\end{equation}
which proves \eqref{eq:simple-small-kappa-error}.

Finally, since
\begin{equation}
s\le \left\lceil \log_4(1/\varepsilon)\right\rceil
=
\mathcal{O}(\log(1/\varepsilon)),
\end{equation}
the bounds in \eqref{eq:cor-proof-starting-runtime} give
\begin{equation}
N
=
\mathcal{O}(L^s)
=
L^{\,\mathcal{O}(\log(1/\varepsilon))}
\end{equation}
and
\begin{equation}
t
=
\mathcal{O}(Lm^3+m^2L^{s+1})
=
\mathcal{O}\!\left(
Lm^3+m^2L^{\,1+\mathcal{O}(\log(1/\varepsilon))}
\right),
\end{equation}
which prove \eqref{eq:simple-small-kappa-term-count} and
\eqref{eq:simple-small-kappa-runtime}.
\end{proof}

\section{Readout from the superposition of Gaussian states}
\label{sec:readout}

Let us now consider the readout stage for the proposed circuit design in Fig.~\ref{fig:circuit_scheme} yielding the superposition of Gaussian states as an output. Since quadrature detection is the basic measurement primitive in continuous-variable platforms \cite{RevModPhys.77.513}, a final Gaussian basis change allows one to express many relevant readout schemes in this language. For instance, one may measure rotated quadratures such as $\hat{X}_{\theta}=\hat{X}\cos\theta+\hat{P}\sin\theta$. This is natural in our setting, because the final Gaussian layer is able to change the measurement basis \cite{RevModPhys.84.621}. Bell-type measurements also fit into this framework, since they can be implemented by adding auxiliary modes, applying Gaussian interference, and then performing quadrature measurements \cite{PhysRevLett.80.869}.

Besides sampling from output quadrature distributions, one is often interested in expectation values of observables. Typical examples include moments and correlations such as $\langle \hat{X}_j\rangle$, $\langle \hat{X}_j^2\rangle$, $\langle \hat{X}_j\hat{X}_k\rangle$, $\langle \hat{P}_j\hat{P}_k\rangle$, and more general polynomial observables $O(\hat{\bm X},\hat{\bm P})$. These quantities are central in theoretical analysis, since they directly characterize the output state and provide the main objects compared across simulation methods. Accordingly, in this section we consider two readout tasks for the final coherent-state superposition: sampling from the output distribution, and computing expectation values of observables. In both cases we give explicit algorithms and the resulting complexity scalings.

\subsection{Sampling from Superpositions of Gaussian States}
\label{subsec:readout-sampling}

After the final Gaussian unitary, the output state is a finite superposition of Gaussian states,
\begin{equation}\label{eq:psi_superposition_gaussians}
|\Psi\rangle=\sum_{k=1}^{N} c_k |G_k\rangle,
\end{equation}
where each $|G_k\rangle$ is an $m$-mode Gaussian state in the notation of Section~\ref{sec:bosonic-cv-qc}. Any componentwise phase can be absorbed into $c_k$.

To describe homodyne-coordinate sampling, we work in the coordinate basis $|\vec{x}\rangle$ with $\vec{x}\in\mathbb{R}^m$ and denote the wavefunction of each Gaussian component by
\begin{equation}\label{eq:psi_k_def}
\psi_k(\vec{x})=\langle \vec{x}|G_k\rangle,
\end{equation}
which in our parametrization is given by Eq.~\eqref{eq:gaussian state in coordinate} and uses the block decomposition
\begin{equation}\label{eq:Sigma_blocks_again}
\Sigma_k=
\begin{pmatrix}
(\Sigma_k)_{XX} & (\Sigma_k)_{XP}\\
(\Sigma_k)_{PX} & (\Sigma_k)_{PP}
\end{pmatrix},
\qquad
(\Sigma_k)_{PX}=(\Sigma_k)_{XP}^{\mathsf T},
\end{equation}
where $(\Sigma_k)_{XX}$ is the $m\times m$ covariance block for the position quadratures and $(\Sigma_k)_{XP}$ is the $m\times m$ position-momentum correlation block.

The (Born) probability density of obtaining the coordinate outcome $\vec{x}$ is
\begin{equation}\label{eq:born_prob_Px_superposition}
P(\vec{x})=\frac{|\langle \vec{x}|\Psi\rangle|^2}{\langle \Psi|\Psi\rangle}
=\frac{\left|\sum_{k=1}^{N} c_k\,\psi_k(\vec{x})\right|^2}{\sum_{i,j=1}^{N} c_i c_j^* \langle G_j|G_i\rangle}
=\frac{\sum_{i,j=1}^{N} c_i c_j^*\,\psi_i(\vec{x})\,\psi_j(\vec{x})^*}{\langle \Psi|\Psi\rangle}.
\end{equation}
This expression makes explicit the interference terms $i\neq j$ through the products $\psi_i(\vec{x})\psi_j(\vec{x})^*$, while each $\psi_k(\vec{x})$ is evaluated via Eq.~\eqref{eq:gaussian state in coordinate} in the coordinate basis.

\textbf{Cost of evaluating $P(\vec{x})$.}
Given $\vec{x}\in\mathbb{R}^m$, evaluating $P(\vec{x})$ in Eq.~\eqref{eq:born_prob_Px_superposition} amounts to computing
$\psi_k(\vec{x})=\langle \vec{x}|G_k\rangle$ for $k=1,\dots,N$ (via Eq.~\eqref{eq:gaussian state in coordinate}),
forming $A(\vec{x})=\sum_{k=1}^N c_k\psi_k(\vec{x})$, and returning $P(\vec{x})=|A(\vec{x})|^2$.
Assuming $\Sigma_{k,XX}^{-1}$ and $\Sigma_{k,XP}\Sigma_{k,XX}^{-1}$ are precomputed, the dominant work per component is evaluating
dense $m$-dimensional quadratic forms, costing $O(m^2)$ floating-point operations (FLOP). Hence
\begin{equation}\label{eq:cost_Px}
\mathrm{FLOPs}\big(P(\vec{x})\big)=O(N m^2),
\end{equation}
up to lower-order $O(Nm)$ vector operations and $O(N)$ scalar exponentials (and $O(Nm^3)$ preprocessing if the matrix products are not precomputed).

Instead of evaluating $P(\vec{x})$ exactly (which requires summing all $N$ amplitudes), Ref.~\cite{hahn2025classical} proposes an
approximate scheme based on \emph{sparsifying} the Gaussian decomposition via Monte Carlo sampling. Writing
$|\Psi\rangle=\sum_{k=1}^N c_k|G_k\rangle$, they sample indices from $p(k)=|c_k|/\|c\|_1$ and absorb the phase of $c_k$ into a
rephased Gaussian state $|\widetilde{G}_k\rangle$, so that $|\Psi\rangle=\|c\|_1\,\mathbb{E}_{k\sim p}[|\widetilde{G}_k\rangle]$.
An empirical average over $L$ i.i.d.\ draws yields a sparsified state $|\widetilde{\Psi}_L\rangle$ supported on only $L$ Gaussian
components, and $\widetilde{P}(\vec{x})$ is obtained by applying the Born rule to $|\widetilde{\Psi}_L\rangle$, with each amplitude
$\langle \vec{x}|\widetilde{G}_{k}\rangle$ evaluated via Eq.~\eqref{eq:gaussian state in coordinate}.

The estimator satisfies $\mathbb{E}\!\left[\|\,|\Psi\rangle-|\widetilde{\Psi}_L\rangle\|^2\right]\le \|c\|_1^2/L$, hence achieving
state error $\varepsilon$ is ensured by $L=O(\|c\|_1^2/\varepsilon^2)$ \cite{hahn2025classical}. Since one evaluation of
$\langle \vec{x}|G\rangle$ costs $O(m^2)$ FLOPs (dense quadratic forms), the resulting per-point cost scales as
\begin{equation}\label{eq:cost_Px_approx_eps}
\mathrm{FLOPs}\big(\widetilde{P}(\vec{x})\big)
=O\!\left(\frac{\|c\|_1^2}{\varepsilon^2}\,m^2\right),
\end{equation}
up to lower-order vector operations and exponentials.

\textbf{Sampling algorithm.} We aim to sample $\vec{x}\in\mathbb{R}^m$ from the target density $P(\vec{x})$ in Eq.~\eqref{eq:born_prob_Px_superposition}
without evaluating it on a full grid. We use an \emph{independence} Metropolis--Hastings (MH) sampler: at each step we propose a
fresh $\vec{x}'$ from a fixed proposal $q(\vec{x}')$ (no local random walk), and accept/reject so that the stationary distribution
is $P$.

A natural choice of proposal is the incoherent Gaussian mixture obtained by discarding interference,
\begin{equation}\label{eq:q_indep_choice}
q(\vec{x})
:=\sum_{k=1}^{N} w_k\,|\psi_k(\vec{x})|^2,
\qquad
\sum_{k=1}^{N} w_k=1,
\end{equation}
where $\psi_k(\vec{x})=\langle \vec{x}|G_k\rangle$ is given by Eq.~\eqref{eq:gaussian state in coordinate}.
Sampling from \eqref{eq:q_indep_choice} is efficient: draw $k\sim\{w_k\}$ and then draw $\vec{x}\sim\mathcal{N}((\vec{\mu}_k)_X,(\Sigma_k)_{XX})$.

Given the current state $\vec{x}^{(t)}$, propose $\vec{x}'\sim q(\cdot)$ and accept with probability
\begin{equation}\label{eq:indep_MH_accept_vec}
\alpha(\vec{x}^{(t)},\vec{x}')
=\min\left\{1,\frac{P(\vec{x}')\,q(\vec{x}^{(t)})}{P(\vec{x}^{(t)})\,q(\vec{x}')}\right\}.
\end{equation}
Each MH step requires evaluating $P(\cdot)$ and $q(\cdot)$ at two points; with precomputed $(\Sigma_k)_{XX}^{-1}$ and
$(\Sigma_k)_{XP}(\Sigma_k)_{XX}^{-1}$, this costs $O(Nm^2)$ FLOPs per evaluation (hence $O(Nm^2)$ per step up to constants).

Finally, Theorem~\ref{thm:interference-boost-bound} gives an \emph{envelope} relating $P$ and the incoherent mixture. In particular,
for the choice $w_k\propto |c_k|^2$ (so that $q$ is proportional to $\sum_k |c_k|^2|\psi_k|^2$), one obtains a pointwise bound
$P(\vec{x})\le M\,q(\vec{x})$ with $M=O(N)$ (and $M=N$ under the normalization $\sum_k|c_k|^2=1$ and $\langle\Psi|\Psi\rangle=1$).
Therefore, by Theorem~\ref{thm:MT-independence-O(M)}, the independence MH chain mixes in
\begin{equation}\label{eq:mixing_steps_O_N}
t=O\!\big(M\log(1/\varepsilon)\big)=O\!\big(N\log(1/\varepsilon)\big)
\end{equation}
steps to reach total-variation error at most $\varepsilon$.

Thus, the total cost to obtain an $\varepsilon$-mixed sample satisfies
\begin{equation}\label{eq:flops_one_eps_mixed_sample}
\mathrm{FLOPs}\big(\text{one sample } \vec{x}\sim P\big)
=O\!\big(N^2 m^2 \log(1/\varepsilon)\big).
\end{equation}

Independence Metropolis--Hastings produces a \emph{single} Markov chain whose stationary distribution is the target $P$.
After an initial burn-in, we do not restart the algorithm from scratch for each new sample. Instead, we keep advancing the
same chain and record states that are sufficiently separated in time. The only requirement is to skip roughly one
\emph{correlation time} between recorded samples so that successive outputs are approximately independent. Hence, producing
additional samples is cheaper: it requires only a modest number of extra MH steps per sample, rather than re-paying the
burn-in cost each time.

\begin{theorem}[Uniform ergodicity and $O(M\log(1/\varepsilon))$ mixing for independence MH. Theorem 2.1 from {\cite{MengersenTweedie1996Rates}}]\label{thm:MT-independence-O(M)}
Let $p$ be a target density on a measurable space $(\mathsf{X},\mathcal{B})$, and let $q$ be a proposal density
(independent of the current state). Consider the independence Metropolis--Hastings chain with transition kernel $P$
and acceptance rule
\begin{equation}\label{eq:indep-MH-accept}
\alpha(x,y) \;=\; \min\!\left\{1,\ \frac{p(y)\,q(x)}{p(x)\,q(y)}\right\}.
\end{equation}
Assume there exists a constant $b>0$ such that
\begin{equation}\label{eq:MT-lower-bound}
\frac{q(y)}{p(y)} \;\ge\; b \qquad \text{for all } y\in\mathsf{X}.
\end{equation}
Then for all $x\in\mathsf{X}$ and all $t\in\mathbb{Z}_{\ge 0}$,
\begin{equation}\label{eq:MT-TV-bound}
\bigl\|P^{t}(x,\cdot)-p(\cdot)\bigr\|_{\mathrm{TV}}
\;\le\; (1-b)^{t}.
\end{equation}
In particular, if there exists $M\ge 1$ such that $p(y)\le M q(y)$ for all $y\in\mathsf{X}$ (equivalently, $b\ge 1/M$),
then
\begin{equation}\label{eq:MT-TV-bound-M}
\bigl\|P^{t}(x,\cdot)-p(\cdot)\bigr\|_{\mathrm{TV}}
\;\le\; \left(1-\frac{1}{M}\right)^{t}
\;\le\; \exp\!\left(-\frac{t}{M}\right),
\end{equation}
and to ensure $\|P^{t}(x,\cdot)-p\|_{\mathrm{TV}}\le \varepsilon$ it suffices that
\begin{equation}\label{eq:MT-steps}
t \;\ge\; M \ln(1/\varepsilon),
\end{equation}
i.e.\ $t = O\!\bigl(M\log(1/\varepsilon)\bigr)$.
\end{theorem}

\begin{theorem}[Interference boost bound for a superposition of $N$ components]\label{thm:interference-boost-bound}
Let $\{\psi_i(x)\}_{i=1}^N$ be complex-valued functions on some domain (e.g., $x\in\mathbb{R}^m$), and let
$c_i\in\mathbb{C}$ be coefficients. Define
\begin{equation}\label{eq:psi-superposition}
\psi(x) := \sum_{i=1}^N c_i\,\psi_i(x),
\qquad
P(x) := |\psi(x)|^2.
\end{equation}
Define also the incoherent mixture
\begin{equation}\label{eq:q-mixture}
q(x) := \sum_{i=1}^N |c_i|^2\,|\psi_i(x)|^2.
\end{equation}
Then for every $x$,
\begin{equation}\label{eq:P-leq-Nq}
P(x)\ \le\ N\, q(x).
\end{equation}
Equivalently, whenever $q(x)>0$,
\begin{equation}\label{eq:ratio-bound}
\frac{P(x)}{q(x)} \le N .
\end{equation}
\end{theorem}

\begin{proof}
Fix any $x$. Define the vector $u(x)\in\mathbb{C}^N$ by $u_i(x):=c_i\,\psi_i(x)$, and let
$\mathbf{1}=(1,1,\dots,1)\in\mathbb{C}^N$. Then
\begin{equation}\label{eq:psi-as-inner-product}
\psi(x)=\sum_{i=1}^N c_i\psi_i(x)=\mathbf{1}^{\mathsf T}u(x).
\end{equation}
By the Cauchy--Schwarz inequality,
\begin{equation}\label{eq:cs}
|\mathbf{1}^{\mathsf T}u(x)|^2 \le \|\mathbf{1}\|_2^2\,\|u(x)\|_2^2.
\end{equation}
Since $\|\mathbf{1}\|_2^2=N$, and
\begin{equation}\label{eq:u-norm}
\|u(x)\|_2^2=\sum_{i=1}^N |u_i(x)|^2=\sum_{i=1}^N |c_i|^2\,|\psi_i(x)|^2=q(x),
\end{equation}
we obtain
\begin{equation}\label{eq:conclude}
P(x)=|\psi(x)|^2 = |\mathbf{1}^{\mathsf T}u(x)|^2 \le N\,q(x),
\end{equation}
which proves \eqref{eq:P-leq-Nq}. The ratio bound \eqref{eq:ratio-bound} follows immediately for $q(x)>0$.
\end{proof}

\subsubsection{Scaling improvement for coherent-product components}\label{sec:strong-coherent-scaling}

If each component $|G_k\rangle$ is a factorized product coherent state, then the homodyne amplitude
$\psi_k(\vec{x})=\langle \vec{x}|G_k\rangle$ factorizes over modes, and evaluating $\psi_k(\vec{x})$ costs $O(m)$ rather than
$O(m^2)$ (no dense quadratic forms). Consequently, one evaluation of the unnormalized target
$\widetilde{P}(\vec{x})=\big|\sum_{k=1}^N c_k \psi_k(\vec{x})\big|^2$ scales as $O(Nm)$, and an independence-MH step has the
same order (up to constants for evaluating $q$). Combining this with the envelope bound $P\le M q$ with $M=O(N)$ yields an
overall strong-sampling cost $O(N^2 m \log(1/\varepsilon))$ for an $\varepsilon$-mixed sample.

\subsection{Computing Expectation Values for Superpositions of Gaussian States}
\label{subsec:readout-expectation-values}

In this section we consider \emph{weak simulation}: given a non-Gaussian pure state expressed as a finite superposition
of $m$-mode pure Gaussian states \eqref{eq:psi_superposition_gaussians},
we compute expectation values of polynomial observables in the canonical quadratures.  Concretely, let
$\hat{\bm X}=(\hat X_1,\ldots,\hat X_m)^{\mathsf T}$ and $\hat{\bm P}=(\hat P_1,\ldots,\hat P_m)^{\mathsf T}$ with
$[\hat X_j,\hat P_k]=i\,\delta_{jk}$. We consider polynomial observables in the canonical quadratures of the form
\begin{equation}
\label{eq:poly-observable}
\hat O
=
\sum_{k=1}^{\eta} o_k\;
\hat{\bm X}^{\bm q^{(k)}}\,\hat{\bm P}^{\bm s^{(k)}}+o_k^*\;\hat{\bm P}^{\bm s^{(k)}}\,\hat{\bm X}^{\bm q^{(k)}},
\qquad
\hat{\bm X}^{\bm q}:=\prod_{j=1}^m \hat X_j^{q_j},\quad
\hat{\bm P}^{\bm s}:=\prod_{j=1}^m \hat P_j^{s_j},
\end{equation}
where $\bm q^{(k)},\bm s^{(k)}\in\mathbb{Z}_{\ge 0}^{m}$ are arbitrary multi-indices, and $\eta$ is the number of monomials appearing in $\hat O$, the notation $^*$ indicates the complex conjugation.  We assume the coefficients
$\{o_k\}_{k=1}^{\eta}$ are chosen such that $\hat O$ is Hermitian.

Let $\ket{\Psi}=\sum_{k=1}^{N} c_k \ket{G_k}$ be a normalized state, i.e.\ $\braket{\Psi}{\Psi}=1$.  Then its
expectation value reduces to Gaussian cross-moments:
\begin{equation}
\label{eq:weak-expand}
\bra{\Psi}\hat O\ket{\Psi}
=
\sum_{i,j=1}^{N} c_i^* c_j\,
\bra{G_i}\hat O\ket{G_j}
=
\sum_{i,j=1}^{N} c_i^* c_j
\sum_{k=1}^{\eta} 2\Re \left\{o_k\;
\bra{G_i}\hat{\bm X}^{\bm q^{(k)}}\hat{\bm P}^{\bm s^{(k)}}\ket{G_j}\right\}.
\end{equation}
Thus weak simulation is reduced to evaluating the elementary matrix elements
$\bra{G_i}\hat{\bm X}^{\bm q^{(k)}}\hat{\bm P}^{\bm s^{(k)}}\ket{G_j}$ for $i,j\in[N]$ and $k\in[\eta]$.

\begin{theorem}[Closed-form generating function for cross-moments]
\label{thm:cross-moments-generating}
For each Gaussian component $\ket{G_k}=\ket{\Sigma_k,\vec{\mu}_k}$, let $\psi_k(\vec{x})=\braket{\vec{x}}{G_k}$
be the coordinate-space wavefunction given by Eq.~\eqref{eq:gaussian state in coordinate}.  Define the
complex $m\times m$ matrix
\begin{equation}
\label{eq:A_k_def}
A_k
:=
\frac{1}{4}(\Sigma_k)_{XX}^{-1}
+
\frac{i}{2}(\Sigma_k)_{XP}(\Sigma_k)_{XX}^{-1},
\qquad
A_k^\ast
=
\frac{1}{4}(\Sigma_k)_{XX}^{-1}
-
\frac{i}{2}(\Sigma_k)_{XP}(\Sigma_k)_{XX}^{-1},
\end{equation}
and write $(\vec{\mu}_k)_X\in\mathbb{R}^m$ for the position-mean block of $\vec{\mu}_k$ as in
Eq.~\eqref{eq:Sigma_blocks_again}.

For any pair $(i,j)$, introduce a \emph{single} $2m$-component source $\bm\lambda=(\bm\lambda_X,\bm\lambda_P)\in\mathbb{C}^{2m}$
and the generating function
\begin{equation}
\label{eq:Z_def}
Z_{ij}(\bm\lambda)
:=
\bra{G_i}\exp\!\big(\bm\lambda_X^{\mathsf T}\hat{\bm X}\big)\exp\!\big(\bm\lambda_P^{\mathsf T}\hat{\bm P}\big)\ket{G_j}.
\end{equation}
Then $Z_{ij}(\bm\lambda)$ admits the closed form
\begin{equation}
\label{eq:Z_closed_form}
Z_{ij}(\bm\lambda)
=
\mathcal{N}_{ij}\,
\frac{(2\pi)^{m/2}}{\sqrt{\det M_{ij}}}\,
\exp\!\left(
\frac{1}{2}\,J_{ij}(\bm\lambda)^{\mathsf T}M_{ij}^{-1}J_{ij}(\bm\lambda)
+
K_{ij}(\bm\lambda_P)
\right),
\end{equation}
where
\begin{align}
\label{eq:M_def}
M_{ij}
&:=2\,(A_i^\ast + A_j),\\[2pt]
\label{eq:J_def}
J_{ij}(\bm\lambda)
&:=\bm\lambda_X
+2\Big(A_i^\ast(\vec{\mu}_i)_X + A_j\big((\vec{\mu}_j)_X+i\,\bm\lambda_P\big)\Big),\\[2pt]
\label{eq:K_def}
K_{ij}(\bm\lambda_P)
&:= -(\vec{\mu}_i)_X^{\mathsf T}A_i^\ast(\vec{\mu}_i)_X
-\big((\vec{\mu}_j)_X+i\,\bm\lambda_P\big)^{\mathsf T}A_j\big((\vec{\mu}_j)_X+i\,\bm\lambda_P\big),
\end{align}
and the prefactor $\mathcal{N}_{ij}$ collects the normalization constants from
Eq.~\eqref{eq:gaussian state in coordinate}:
\begin{equation}
\label{eq:Nij_def}
\mathcal{N}_{ij}
:=
\left(\frac{1}{2\pi}\right)^{m/2}
\big(\det(\Sigma_i)_{XX}\,\det(\Sigma_j)_{XX}\big)^{-1/4}.
\end{equation}
Moreover, for any multi-indices $\bm q,\bm s\in\mathbb{Z}_{\ge 0}^m$,
\begin{equation}
\label{eq:moment_by_derivatives}
\bra{G_i}\hat{\bm X}^{\bm q}\hat{\bm P}^{\bm s}\ket{G_j}
=
\left.
\left(\prod_{r=1}^m \frac{\partial^{q_r}}{\partial (\lambda_{X,r})^{q_r}}\right)
\left(\prod_{r=1}^m \frac{\partial^{s_r}}{\partial (\lambda_{P,r})^{s_r}}\right)
Z_{ij}(\bm\lambda)
\right|_{\bm\lambda=\bm 0}.
\end{equation}
\end{theorem}

\begin{proof}
Insert resolutions of the identity in the coordinate basis and use $\hat{\bm P}=-i\,\nabla_{\bm x}$ in
the $\ket{\bm x}$ representation.  The operator $\exp(\bm\lambda_P^{\mathsf T}\hat{\bm P})$ acts as a translation
$\big[\exp(\bm\lambda_P^{\mathsf T}\hat{\bm P})f\big](\bm x)=f(\bm x-i\,\bm\lambda_P)$, so
\begin{equation}
Z_{ij}(\bm\lambda)
=\int_{\mathbb{R}^m} d\bm x\;\psi_i(\bm x)^\ast\,
e^{\bm\lambda_X^{\mathsf T}\bm x}\,
\psi_j(\bm x-i\,\bm\lambda_P),
\end{equation}
with $\psi_k$ given by Eq.~\eqref{eq:gaussian state in coordinate}.  Multiplying the two Gaussian
factors yields a quadratic form in $\bm x$ with coefficient matrix $A_i^\ast+A_j$ and linear term determined by
$(\vec{\mu}_i)_X$, $(\vec{\mu}_j)_X$, and the shift $i\,\bm\lambda_P$, together with the normalization and phase
constants \cite{do2008more}.  Completing the square gives the Gaussian integral in \eqref{eq:Z_closed_form} with the definitions
\eqref{eq:M_def}--\eqref{eq:Nij_def}.  Finally, differentiating \eqref{eq:Z_def} at $\bm\lambda=\bm 0$ generates the
ordered moments \eqref{eq:moment_by_derivatives}.
\end{proof}

Combining Theorem~\ref{thm:cross-moments-generating} with Eq.~\eqref{eq:weak-expand}, weak simulation of any polynomial
observable \eqref{eq:poly-observable} reduces to evaluating $Z_{ij}(\bm\lambda)$ and its partial derivatives at
$\bm\lambda=\bm 0$ for all pairs $(i,j)$.

\textbf{Computational cost (FLOPs).}
Eq.~\eqref{eq:weak-expand} requires $N^2\eta$ elementary cross-moments
$\bra{G_i}\hat{\bm X}^{\bm q^{(k)}}\hat{\bm P}^{\bm s^{(k)}}\ket{G_j}$.
For each pair $(i,j)$ we form $M_{ij}$ and calculate its inverse $M_{ij}^{-1}$, costing $O(m^3)$ FLOPs, after which
evaluating each monomial moment via Theorem~\ref{thm:cross-moments-generating} reduces to a constant number of dense
matrix--vector operations/linear solves, costing $O(m^2)$ FLOPs per monomial. Hence
\begin{equation}
\label{eq:weak_flops}
\mathrm{FLOPs}\big(\bra{\Psi}\hat O\ket{\Psi}\big)
=
O\!\big(N^2(m^3+\eta\,m^2)\big).
\end{equation}

\begin{remark}[Coherent-product specialization]
If each Gaussian component is a \emph{factorized coherent product},
$\ket{G_i}=\ket{\bm\alpha_i}:=\bigotimes_{r=1}^m \ket{\alpha_{i,r}}$,
then the cross-generating function factorizes over modes and all cross-moments
$\bra{G_i}\hat{\bm X}^{\bm q}\hat{\bm P}^{\bm s}\ket{G_j}$ reduce to products of
single-mode derivatives (no dense $m\times m$ linear algebra). Consequently, the
weak-simulation cost in Eq.~\eqref{eq:weak_flops} drops from
$O\!\big(N^2(m^3+\eta m^2)\big)$ to
\begin{equation}
\mathrm{FLOPs}\big(\bra{\Psi}\hat O\ket{\Psi}\big)
=
O\!\Big(N^2 m\big(\eta + q_{\max}s_{\max}\big)\Big),
\end{equation}
where $q_{\max}:=\max_{k,r} q^{(k)}_r$ and $s_{\max}:=\max_{k,r} s^{(k)}_r$.
\end{remark}

\section{The Driven Bose--Hubbard Model}
\label{sec:driven-bose-hubbard}

The Bose--Hubbard model is a paradigmatic lattice model for interacting bosons. It captures the competition between coherent delocalization across lattice sites and local interactions, and provides the standard framework for boson localization and the superfluid--Mott-insulator transition \cite{Bose_hubbard2,Bose_hubbard_1}. In the form considered here, it is written as
\begin{equation}
\label{eq:bose_hubbard_model}
H_\text{BH}
=
-J\sum_{\langle i,j\rangle}\left(a_i^\dagger a_j+a_j^\dagger a_i\right)
+\frac{U}{2}\sum_i n_i^2
-\Delta\sum_i n_i,
\qquad
n_i=a_i^\dagger a_i.
\end{equation}

Here, $J>0$ is the tunneling (hopping) strength. The first term in Eq.~\eqref{eq:bose_hubbard_model} describes the coherent transfer of bosons between coupled sites. The notation $\langle i,j\rangle$ denotes the pairs of sites connected by the hopping term, and therefore specifies the underlying geometry or topology on which the bosons live. Common choices include one-dimensional chains, two-dimensional lattices, three-dimensional lattices, and all-to-all coupled graphs. The parameter $U>0$ represents the effective on-site repulsion, penalizing multiple bosons occupying the same site. The parameter $\Delta$ denotes the frequency detuning of the pump with respect to the cavity mode.

\begin{figure}[h!]
    \centering
    \includegraphics[width=0.8\linewidth]{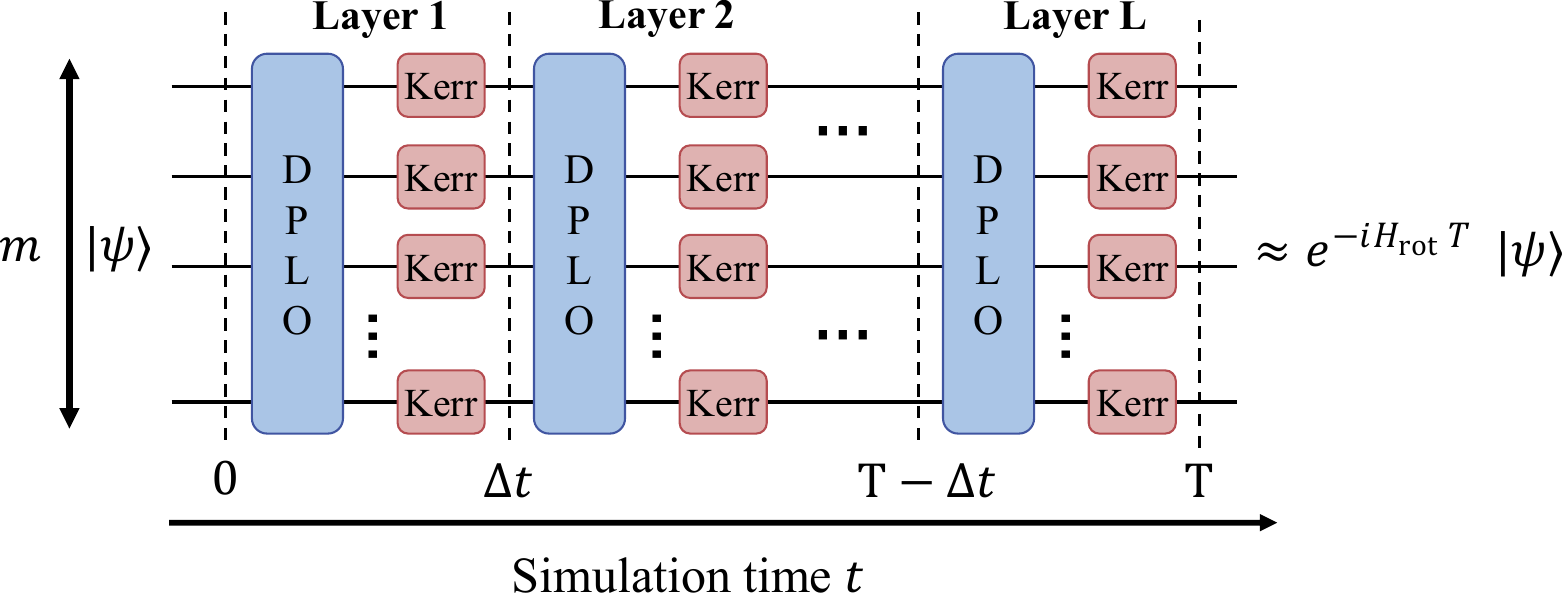}
    \caption{Schematic of the first-order Lie--Trotterized evolution in Eq.~\eqref{eq:trotter_on}. The dynamics acts on $m$ bosonic modes and is decomposed into $L$ Trotter layers of duration $\delta t$. Each layer consists of a displaced passive linear-optical (DPLO) sublayer $e^{-i H_{\rm dPLO}\delta t}$, followed by a Kerr sublayer $e^{-i H_{\rm Kerr}\delta t}$.}
    \label{fig:overview_layer_structure}
\end{figure}

From an algorithmic point of view, the Bose--Hubbard Hamiltonian in Eq.~\eqref{eq:bose_hubbard_model} is number conserving, as discussed in Appendix~\ref{sec:passive-lo-kerr-universality}. Indeed, both the hopping term and the onsite interaction term commute with the total excitation operator $\hat N=\sum_{j=1}^m \hat n_j$, so the evolution preserves the total photon number in the Fock basis. As a result, the dynamics does not explore the full infinite-dimensional multimode Hilbert space at once, but decomposes into independent fixed-excitation sectors $\mathcal H_N^{(m)}$. For example, in the single-excitation sector, the evolution is confined to the subspace spanned by states such as $|10\dots0\rangle$, $|01\dots0\rangle$, \dots, $|00\dots1\rangle$. This block structure is algorithmically important, because it reduces the simulation to separate finite-dimensional invariant subspaces rather than the full bosonic space.

\begin{algorithm}[h!]
\caption{Approximate forward propagation for small non-Gaussianity}
\label{alg:cs-forward}
\DontPrintSemicolon

\KwIn{
Initial $(\bm C,A)$ for $\ket{\psi}=\sum_{r=1}^{N} C_r \bigotimes_{\mathrm{mode}=1}^m \ket{\alpha^{(\mathrm{mode})}_r}$ in \eqref{eq:coherent_superposition_background}.\\
Layer data $\{(K^{(\ell)},\bm\eta^{(\ell)},t_\ell)\}_{\ell=1}^Q$ for DPLO, Eqs.~\eqref{eq:HDPLO_background}.\\
Kerr strengths $\{\kappa^{(\ell)}_{\mathrm{mode}}\}$ acting on mode $\mathrm{mode}$ in layer $\ell$.\\
Finite-Fourier parameter $M\in\mathbb N$ for the Kerr approximation, Eqs.~\eqref{eq:finite-fourier-kerr}--\eqref{eq:finite-fourier-coeffs}.\\
(Optional) coefficient cutoff $\mathcal S$ retaining the $\mathcal S$ largest values of $|C_r|$ after each Kerr update.
}

\KwOut{Approximate final $(\bm C,A)$ representing the propagated state.}

\BlankLine
\For{\textnormal{\textbf{Layer} }$\ell=1$ \KwTo $Q$}{

\BlankLine
\textbf{DPLO sublayer, see Proposition~\ref{prop:DPLO_preserves_N}.}
\begin{equation}
(\bm C,A)\leftarrow \mathrm{DPLO}\big((\bm C,A);K^{(\ell)},\bm\eta^{(\ell)},t_\ell\big)
\end{equation}

\BlankLine
\textbf{Approximate Kerr sublayer via the finite-Fourier update, see Appendix~\ref{subsec:finite-fourier-cutoff}.}
\For{$\mathrm{mode}=1$ \KwTo $m$}{
$N_{\mathrm{curr}}\leftarrow \mathrm{rows}(A)$

\For{$r=1$ \KwTo $N_{\mathrm{curr}}$}{
\begin{equation}
(\bm C,A)\leftarrow
\mathrm{Kerr}_{\mathrm{FF}}\big((\bm C,A);r,\mathrm{mode},\kappa^{(\ell)}_{\mathrm{mode}},M\big)
\end{equation}
}

\If{$\mathcal S$ is used}{
\begin{equation}
(\bm C,A)\leftarrow \mathrm{Truncate}\big((\bm C,A);\mathcal S\big)
\end{equation}
}
}
}

\Return $(\bm C,A)$
\end{algorithm}

Adding coherent pumping gives the driven Bose--Hubbard model, which is standard in photonic lattice and cavity-array settings \cite{driven_bose_hubbard,driven_bose_hubbard_2}. Physically, the pump injects photons with a fixed phase, while algorithmically it breaks the photon-number conservation, since the drive does not commute with $\hat N=\sum_i n_i$. Thus the dynamics is no longer confined to fixed-excitation sectors. In the laboratory frame, we write
\begin{equation}
\label{eq:H_driven_lab}
H(t)
=
H_{\mathrm{BH}}
+
\sum_i\left(\Omega_i e^{-i\omega_d t}a_i^\dagger+\Omega_i^* e^{i\omega_d t}a_i\right),
\end{equation}
where $\Omega_i$ is the amplitude of the incident laser field. Passing to the rotating frame $R(t)=\exp\!\big(i\omega_d t\sum_i n_i\big)$ and absorbing the linear density shift into the effective detuning $\Delta$, one obtains
\begin{equation}
\label{eq:H_rot_compact}
H_{\mathrm{rot}}
=
\underbrace{-J\sum_{\langle ij\rangle}\left(a_i^\dagger a_j+a_j^\dagger a_i\right)
-\Delta\sum_i n_i}_{H_{\rm PLO}}
+
\underbrace{\frac{U}{2}\sum_i n_i^2}_{H_{\rm Kerr}}
+
\underbrace{\sum_i\Big(\Omega_i a_i^\dagger+\Omega_i^* a_i\Big)}_{H_{\rm disp}}.
\end{equation}
The displacement term $H_\text{disp}$ is more than a symmetry-breaking probe: it is a phase-coherent
\emph{control port} that injects amplitude into an interacting lattice with programmable temporal and spatial
structure.

\subsection{Trotterized evolution and simulation algorithm}
\label{subsec:rot_frame}

In this work, we approximate the real-time dynamics by a first-order Lie--Trotter product formula \cite{Trotter1959,SUZUKI1990319}. Writing $\delta t=T_{\rm on}/L$, we split the rotating-frame Hamiltonian into the Kerr part and the displaced passive linear-optical part, and obtain
\begin{equation}
\label{eq:trotter_on}
\exp\!\big(-i H_{\rm rot}^{(i)} T_{\rm on}\big)
\approx
\left[
\exp\!\big(-i H_{\rm Kerr}\,\delta t\big)\;
\exp\!\big(-i H_{\rm DPLO}\,\delta t\big)
\right]^L,
\end{equation}
where $H_\text{DPLO}=H_\text{disp}+H_\text{PLO}$.
For fixed total evolution time, the corresponding global product-formula error scales as $\epsilon\sim\delta t$. Since the Trotter error is not the focus of the present paper, this first-order decomposition is sufficient for our purposes. If needed, one may instead employ higher-order product formulas \cite{SUZUKI1992387,AlvermannFehske2011}, for which the error can be improved to $\epsilon\sim \delta t^{\,s}$ for suitable order $s$.

The first-order Trotterized evolution is shown in Fig.~\ref{fig:overview_layer_structure}. Each Trotter layer has exactly the same DPLO$+$Kerr structure as the universal CV circuit in Fig.~\ref{fig:circuit_scheme}: the linear part $H_{\rm DPLO}=H_{\rm PLO}+H_{\rm disp}$ forms the DPLO block, while the nonlinear part is a product of local Kerr gates generated by $H_{\rm Kerr}=\frac{U}{2}\sum_i n_i^2$. The apparent difference between the two figures is only schematic. Fig.~\ref{fig:circuit_scheme} shows a single representative Kerr wire, but its DPLO block already includes mode permutations (swap networks), so Kerr actions can be routed to arbitrary physical modes. In the Bose--Hubbard Trotterization, the same general architecture is specialized to a layer in which the local Kerr operation is applied on every mode.

From an experimental perspective, the DPLO block corresponds to standard linear-optical ingredients: the hopping terms are realized by beam-splitter networks implementing mode mixing, the detuning term gives local phase shifters, and the coherent pump contributes the displacement \cite{Scully_Zubairy_1997}. The Kerr block represents the onsite nonlinearity; in superconducting-circuit platforms, such effective Kerr interactions can be engineered using Josephson-junction-based nonlinear elements \cite{PhysRevA.86.013814,Kirchmair2013}.

Next, let us consider Algorithm~\ref{alg:cs-forward} which gives the corresponding simulation procedure within our coherent-state propagation framework. Starting from an initial coherent-state superposition \eqref{eq:cs-ansatz}, one propagates the state through the $L$ Trotter layers by first updating the DPLO step and then applying the Kerr update mode by mode, see Appendix~\ref{subsec: tree structure} for more details.

\subsection{Numerical test}\label{subsec:numerics_supplementary}

We now present numerical results for the driven Bose--Hubbard model in order to assess the practical performance of coherent-state propagation in a concrete bosonic many-body setting. Throughout this section, we benchmark the method against two standard reference approaches, namely exact Fock-basis simulation and matrix-product-state (MPS) simulation \cite{A_J_Daley_2004,SCHOLLWOCK201196}. The corresponding setups and comparisons are described separately in the following subsections.

We emphasize that the rigorous analysis in Appendix~\ref{sec:small} is formulated using the two-term approximation of the Kerr action on a coherent state, Eq.~\eqref{eq:psi2-theta-def}. For the numerical simulations reported here, however, we employ more general Kerr-update schemes based on the finite-Fourier approximation introduced in Appendix~\ref{subsec:finite-fourier-cutoff}; see in particular Eq.~\eqref{eq:finite-fourier-coeffs}. In addition, the truncation procedure used in practice is more flexible; after each Kerr update, we retain only the coherent-state branches whose coefficients have the largest magnitudes $|C_j|$, which reduces to the truncation rule for the case $M=1$.

\subsubsection{Exact Fock-basis benchmark}

We benchmark coherent-state propagation against a dense Fock-basis state-vector simulation of the Trotterized dynamics generated by the rotating-frame driven Bose--Hubbard Hamiltonian \eqref{eq:H_rot_compact}, using the first-order product formula \eqref{eq:trotter_on}. In the Fock representation, the state is written as
\begin{equation}
\label{eq:fock_state_three_mode}
\ket{\psi}
=
\sum_{n_1,n_2,n_3=0}^{N_{\mathrm{Fock}}-1}
c_{n_1,n_2,n_3}\ket{n_1,n_2,n_3}.
\end{equation}
Within the chosen cutoff, this benchmark does not impose additional structural restrictions on the state: entanglement and non-Gaussianity are limited only by the finite Hilbert-space truncation. At the same time, for a local cutoff $N_{\mathrm{Fock}}$ per mode, the truncated Hilbert-space dimension scales as $N_{\mathrm{Fock}}^m$, so both runtime and memory suffer from the curse of dimensionality (exponential scaling with $m$). For the cutoffs relevant here, $N_{\mathrm{Fock}}\approx 20$, dense state-vector simulations beyond three qumodes are not practical. We therefore restrict the exact benchmark to the three-mode all-to-all driven Bose--Hubbard model.

The benchmark is carried out in an intermediate-coupling regime with comparable hopping and interaction scales, $J\approx U$, rather than in a quasiclassical regime dominated by a single scale. With total evolution time $T=1$ and $N_t=10$ Trotter steps, each local Kerr update carries the angle
\begin{equation}
\label{eq:kappa_numerics_small_ng}
|\kappa|
=
\frac{UT}{2N_t}
=
0.04,
\end{equation}
so this benchmark probes a regime of small non-Gaussianity.

The local Fock cutoff is chosen from the Poisson-tail estimate of Eq.~\eqref{eq:N-choice}. Explicitly,
\begin{equation}
\label{eq:fock_cutoff_numerics}
N_{\mathrm{Fock}}
=
\left\lceil
\lambda_{\mathrm{est}}
+
\sqrt{2\lambda_{\mathrm{est}}\ln\!\left(\varepsilon_{\mathrm{Fock}}^{-1}\right)}
+
\frac{2}{3}\ln\!\left(\varepsilon_{\mathrm{Fock}}^{-1}\right)
-1
\right\rceil.
\end{equation}
We take $\lambda_{\mathrm{est}}=2$ and $\varepsilon_{\mathrm{Fock}}=10^{-9}$, which gives $N_{\mathrm{Fock}}=24$. The code used for this comparison is available in Ref.~\cite{guseynov_gitlab_placeholder}.

On the coherent-state side, each Kerr update is implemented using the finite-Fourier approximation of Appendix~\ref{subsec:finite-fourier-cutoff}. The observables reported for coherent-state propagation are evaluated using the weak-simulation method of Appendix~\ref{subsec:readout-expectation-values}. The parameters used in this benchmark are summarized in Table~\ref{tab:three_mode_numerical_parameters}.

\begin{table}[h]
\centering
\begin{tabular}{lc}
\toprule
Parameter & Value \\
\midrule
Hopping amplitude $J$ & $0.8$ \\
Kerr interaction strength $U$ & $0.8$ \\
Detuning $\Delta$ & $0.6$ \\
Uniform drive amplitude $\Omega_i$ & $1.0$ \\
Finite-Fourier parameter $M$ & $3$ \\
Global coefficient cutoff $\mathcal{S}$ & $10^4$ \\
\bottomrule
\end{tabular}
\caption{Parameters used in the three-mode all-to-all driven Bose--Hubbard benchmark.}
\label{tab:three_mode_numerical_parameters}
\end{table}

\begin{figure}[h]
    \centering
    \includegraphics[width=0.44\textwidth]{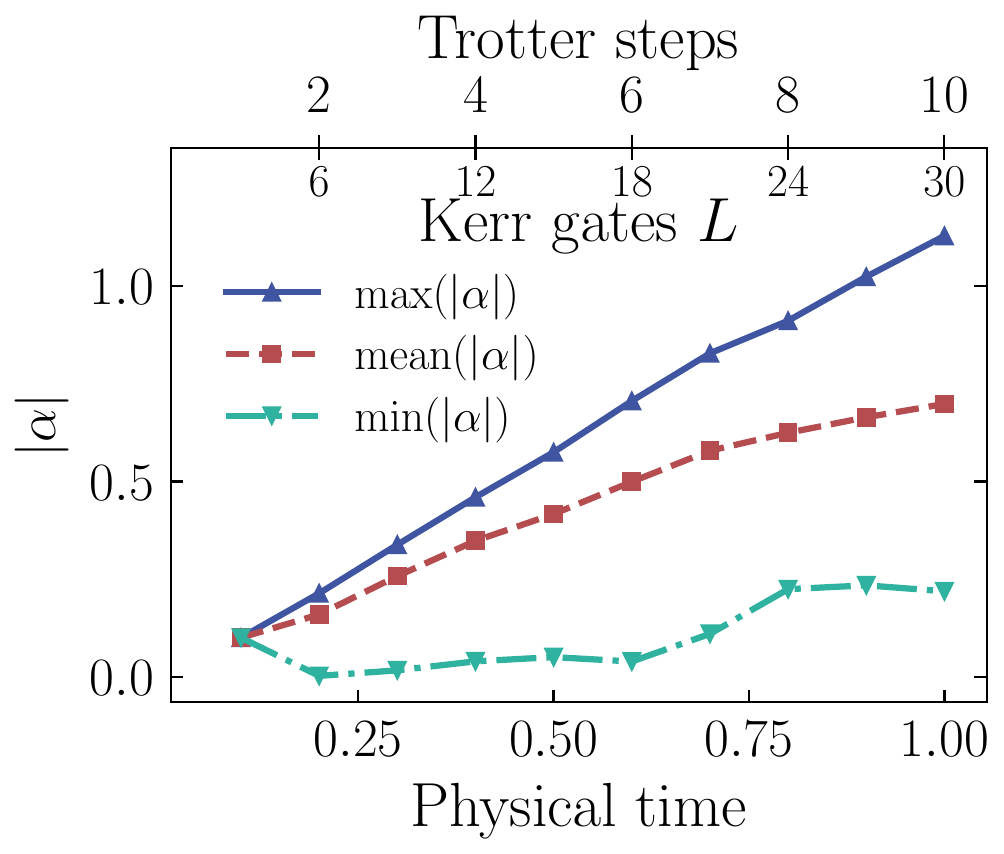}
    \caption{
Diagnostics of the coherent amplitudes in the coherent-state propagation for the three-mode driven Bose--Hubbard numerical experiment with all-to-all connectivity. At each recorded step, the plot shows $\max_{i,j}|\alpha_j^{(i)}|$, $\mathrm{mean}_{i,j}|\alpha_j^{(i)}|$, and $\min_{i,j}|\alpha_j^{(i)}|$ over the retained coherent-state branches and modes.
}
    \label{fig:alpha_diagnostics}
\end{figure}

Figure~\ref{fig:alpha_diagnostics} shows that
$\max_{i,j}|\alpha_j^{(i)}|$ grows approximately linearly with physical
time over the simulated interval. This is consistent with the DPLO
update \eqref{eq:csp_dplo_update_rule}: for short times,
$\bm{\gamma}(t)=-it\,\bm{\eta}+\mathcal{O}(t^2)$, so in the present
benchmark with uniform drive $\eta_j=\Omega$, the slope is set by
$|\Omega|$. The diagnostics therefore confirm that the growth of the
largest coherent amplitudes is controlled by the displacement sector, in
agreement with the bound \eqref{eq: lambda linear growth  main}.

The comparison in Fig.~\ref{fig:numerics_three_observables} shows that, in this small-non-Gaussianity regime, coherent-state propagation reproduces the Fock-basis benchmark for the displayed observables at the shown resolution. The timing comparison in Table~\ref{tab:runtime_fock} further shows that this is achieved with a substantially smaller evolution cost. Since coherent-state propagation also avoids storing a dense state vector in a space of dimension $N_{\mathrm{Fock}}^3$, it requires correspondingly lower memory than the Fock-basis simulation for this benchmark allowing simulation of larger bosonic systems.

\subsubsection{MPS benchmark}

We next benchmark coherent-state propagation against a matrix-product-state (MPS) simulation of the same Trotterized driven Bose--Hubbard dynamics, using parameters chosen analogously to the Fock benchmark; the main results are shown in Fig.~\ref{fig:mps_benchmark}. We implement the MPS reference with the TeNPy \cite{tenpy2024} library and its two-site time-dependent variational principle (TDVP) engine~\cite{haegeman2016tdvp}. In this approach, each qumode is first truncated to a finite local Fock space of size $n_{\max}+1$, similarly to the Fock benchmark above, and the many-body wavefunction is then written in matrix-product-state form, i.e., as a chain of local tensors connected by bond indices. This representation is efficient when the bipartite entanglement remains relatively small.

\begin{table}[h!]
\centering
\begin{tabular}{lc}
\toprule
Parameter & Value \\
\midrule
Hopping amplitude $J$ & $1.6$ \\
Kerr interaction strength $U$ & $0.25$ \\
Detuning $\Delta$ & $0.3$ \\
Uniform drive amplitude $\Omega_i$ & $1.0$ \\
Local bosonic cutoff $n_{\max}$ & $20$ \\
\bottomrule
\end{tabular}
\caption{Parameters used in the six-mode all-to-all driven Bose--Hubbard MPS benchmark.}
\label{tab:mps_numerical_parameters}
\end{table}

Concretely, we consider an $m=6$ bosonic system initialized in the vacuum product state and evolved up to total time $T=1$ with $N_t=20$ first-order Trotter steps. The numerical parameters are summarized in Table~\ref{tab:mps_numerical_parameters}. On the MPS side, $n_{\max}$ denotes the local bosonic cutoff per mode. On the coherent-propagation side, each Kerr update is implemented using the finite-Fourier approximation of Eq.~\eqref{eq:finite-fourier-coeffs}; here $M$ is the finite-Fourier parameter, so that each local Kerr action is approximated by $2M$ coherent states, while $\mathcal{S}$ denotes the global coefficient cutoff applied after each Kerr update, see Fig.~\ref{fig:mps_benchmark}.

We choose the all-to-all connectivity because it is the most challenging case from the point of view of entanglement growth and is therefore the most demanding geometry for an MPS description. At the same time, in our coherent-state encoding a DPLO layer maps each coherent-product branch to another coherent-product branch and is thus non-entangling on a single branch. As a result, a nontrivial part of the entanglement complexity in the representation \eqref{eq:coherent_superposition_background} is also carried by the superposition structure itself, namely by the growth of the number $N$ of retained coherent terms.

Panel~(a) of Fig.~\ref{fig:mps_benchmark} is directly analogous to the Fock-basis comparison: it shows the real-time trajectory of $\langle \hat X_0 \hat X_1\rangle$ together with the corresponding absolute error of coherent propagation relative to the MPS reference, for fixed truncation parameters $\mathcal{S}=2^{12}$ and $M=3$. By contrast, panels~(b) and~(c) are conceptually different. There we fix the final time $T=1$ and study how the error changes as a function of both the coefficient cutoff $\mathcal{S}$ and the finite-Fourier parameter $M$. The corresponding computational times for the parameter scans were measured on the workstation equipped with an NVIDIA RTX 6000 Ada Generation GPU. For an observable $\hat O$, we define the relative error as
\begin{equation}
\label{eq:mps_relative_error}
\epsilon_{\mathrm{rel}}(\hat O)
:=
\frac{\left|\langle \hat O\rangle_{\mathrm{coh}}-\langle \hat O\rangle_{\mathrm{MPS}}\right|}
{\left|\langle \hat O\rangle_{\mathrm{MPS}}\right|}.
\end{equation}
Panels~(b) and~(c) report this quantity for $\hat O=\hat X_0^2$ and $\hat O=\hat X_0\hat X_1$, respectively. The observed behavior reflects a tradeoff between local and global approximations: increasing $M$ improves the single-gate Kerr representation, but it also produces more coherent branches, which makes the global cutoff $\mathcal{S}$ more restrictive. Consequently, for small $\mathcal{S}$ a larger $M$ can initially worsen the final error, whereas once $\mathcal{S}$ is sufficiently large, the improved local Kerr approximation dominates and the error decreases.

Within the numerically accessible regime explored here, we did not observe a parameter window in which coherent propagation outperforms the Matrix-produc-states benchmark. We do not view this as evidence that such a regime does not exist. The main difficulty in identifying such a regime is that, in the coherent-state encoding \eqref{eq:coherent_superposition_background}, both the Kerr-generated non-Gaussianity and a substantial part of the entanglement complexity are reflected indirectly in the required number $N$ of retained coherent terms. In this sense, the same representation size $N$ simultaneously absorbs information about non-Gaussianity and about the superposition complexity needed to represent entangled states.

One can nevertheless imagine regimes in which coherent propagation becomes especially natural. For example, the input state may already contain many coherent-product terms $N_0$, thereby representing an entangled state, and a highly entangling Gaussian unitary $U_G$ may be applied at the end of the circuit; see Fig.~\ref{fig:circuit_scheme}. In such a setting, a direct forward MPS evaluation of $\langle O\rangle$ can become computationally demanding while still be tractable for Coherent-state propagation. However, this does not immediately imply an advantage over MPS, since one can combine MPS with observable-propagation ideas by backpropagating $\hat O$ through $U_G$ in Heisenberg picture \cite{RevModPhys.84.621} and then evolving the $N_0$ initial components independently. For polynomial in quadratures observables, this backward propagation generically produces combinatorially many terms---schematically of binomial type $\binom{m}{q}$ for degree-$q$ contributions---together with the overhead of recombining the $N_0$ branches. Such hybrid strategies may therefore reach a performance comparable to coherent propagation, albeit with overheads, and for this reason we do not claim an advantage over MPS.

\section{Universality of displaced linear optics with a Kerr nonlinearity}
\label{sec:cv-universality-kerr}

In this Appendix we consider the universality of displaced linear optics augmented by a single-mode Kerr nonlinearity. Since the underlying CV Hilbert space is infinite dimensional and the canonical generators are unbounded, quantitative notions of approximation must be handled with care: for general CV states (in particular, outside the Schwartz space) the difference of two operators that coincide on a dense domain can be unbounded, so an operator-norm statement on the full Hilbert space is typically ill-posed. We therefore work in the standard energy-cutoff setting \cite{arzani2025can,Lloyd_PhysRevLett.82.1784}.

\begin{proposition}
\label{prop:universal-kerr-dlo}
Fix an energy cutoff \(n_{\max}\) and let \(\Pi_{n_{\max}}\) be the projector onto the multimode Fock subspace of
total photon number at most \(n_{\max}\). Consider the gate set consisting of displacements, passive linear optics,
and a single-mode Kerr gate. Then this gate set is universal on \(\Pi_{n_{\max}}\mathcal{H}\), i.e., for any unitary
\(U\) acting on \(\Pi_{n_{\max}}\mathcal{H}\) and any \(\varepsilon>0\), there exists a finite circuit \(V\) from the gate
set such that \(\|\Pi_{n_{\max}}(U-V)\Pi_{n_{\max}}\|\le \varepsilon\).
\end{proposition}

\begin{proof}
Restrict all operators to \(\Pi_{n_{\max}}\mathcal H\), so every generator is bounded. Ref.~\cite{Lloyd_PhysRevLett.82.1784} states that arbitrary Gaussian operations together with a Kerr gate form a universal CV gate set. Since displacements and passive linear optics are already available here, it only remains to synthesize single-mode squeezing, namely
\begin{equation}
\label{eq:squeezing_unitary_appendix}
S(r):=\exp\!\left(-\frac{i r}{2}\,(\hat x\hat p+\hat p\hat x)\right),
\end{equation}
where $\hat x$, $\hat p$ are canonical quadratures satisfying
\begin{equation}
\label{eq:canonical_quadratures_appendix}
[\hat x,\hat p]=i,
\qquad
\hat n=\frac{\hat x^2+\hat p^2-1}{2}.
\end{equation}
The Kerr gate provides the quartic generator \(\hat n^2\), while displacements provide the linear generators \(\hat x\) and \(\hat p\). Next, we use the property of the group commutator \cite{loock_PhysRevLett.107.170501} 
\begin{equation}
\label{eq:gc}
G(A,B;\tau):=e^{iA\tau}e^{iB\tau}e^{-iA\tau}e^{-iB\tau}
=
\exp\!\big(-[A,B]\tau^2\big)+f(\tau^3,A,B).
\end{equation}
Hence, by standard nested-commutator constructions, commutators of available generators can be synthesized to arbitrarily high accuracy on \(\Pi_{n_{\max}}\mathcal H\). A direct computation yields
\begin{equation}
\label{eq:double_commutators_appendix}
[\hat x,[\hat n^2,\hat x]]
=
2\hat n+2\hat p^2,
\qquad
[\hat p,[\hat n^2,\hat p]]
=
2\hat n+2\hat x^2.
\end{equation}
Since \(\hat n\) is already available from passive linear optics, Eq.~\eqref{eq:double_commutators_appendix} shows that we can synthesize evolutions generated by \(\hat x^2\) and \(\hat p^2\), and using Trotter-like decomposition one can achieve \(\hat x^2-\hat p^2\).

Let us no consider the single-mode phase-shift operator
\begin{equation}
\label{eq:phase_shift_appendix}
R(\theta):=e^{i\theta \hat n}.
\end{equation}
Under conjugation by \(R(\theta)\), the quadratures rotate as
\begin{equation}
\label{eq:quadrature_rotation_appendix}
R(\theta)\hat x R(\theta)^\dagger
=
\hat x\cos\theta+\hat p\sin\theta,
\qquad
R(\theta)\hat p R(\theta)^\dagger
=
-\hat x\sin\theta+\hat p\cos\theta.
\end{equation}
Setting \(\theta=\pi/4\), we obtain
\begin{equation}
\label{eq:squeezing_generator_pi_over_4_appendix}
R\!\left(\frac{\pi}{4}\right)(\hat x^2-\hat p^2)R\!\left(\frac{\pi}{4}\right)^\dagger
=
\hat x\hat p+\hat p\hat x.
\end{equation}
Thus the squeezing generator is available, and so the unitary \(S(r)\) in Eq.~\eqref{eq:squeezing_unitary_appendix} can be synthesized.

We have therefore obtained squeezing, which together with displacements and passive linear optics yields arbitrary Gaussian operations. Together with the Kerr this forms a universal gate set.
\end{proof}

\section{What passive linear optics with Kerr gate can achieve?}
\label{sec:passive-lo-kerr-universality}

In this section we consider passive (number-conserving) linear optics augmented by a Kerr nonlinearity, and use it to
compile time evolutions for a broad class of number-conserving bosonic Hamiltonians built from linear mode-coupling
terms and Kerr-type nonlinearities. Similarly, Appendix~\ref{sec:cv-universality-kerr} discusses universality when
displacements are also available.

We work on the $m$-mode bosonic Hilbert space
\begin{equation}
\label{eq:hm-def}
\mathcal{H}^{(m)}:=\mathrm{span}\big\{\ket{n_1,\ldots,n_m}: n_j\in\mathbb{N}_0\big\},
\end{equation}
with annihilation operators $\{\hat{a}_j\}_{j=1}^m$ and number operators $\hat{n}_j:=\hat{a}_j^\dagger \hat{a}_j$.
The total photon number is
\begin{equation}
\label{eq:total-number}
\hat{N}:=\sum_{j=1}^m \hat{n}_j.
\end{equation}
For each $N\in\mathbb{N}_0$ we define the fixed-$N$ sector
\begin{equation}
\label{eq:HN-def}
\mathcal{H}^{(m)}_N:=\mathrm{span}\big\{\ket{n_1,\ldots,n_m}: \textstyle\sum_{j=1}^m n_j=N\big\}\subset \mathcal{H}^{(m)}.
\end{equation}

Passive linear optics is generated by quadratic Hamiltonians of the form
\begin{equation}
\label{eq:passive-quadratic}
\hat{H}_{\mathrm{LO}}=\sum_{j,k=1}^m h_{jk}\,\hat{a}_j^\dagger \hat{a}_k,\qquad h_{jk}=h^\dagger_{kj}.
\end{equation}
A single-mode Kerr
nonlinearity is generated by $\hat{n}_\ell^2$ for some mode index $\ell$, i.e.,
\begin{equation}
\label{eq:kerr-gate}
U_{\mathrm{Kerr}}(\kappa):=\exp(-i\kappa \hat{n}_\ell^2).
\end{equation}
Both $\hat{H}_{\mathrm{LO}}$ and $\hat{n}_\ell^2$ commute with $\hat{N}$, and thus any circuit $U$ composed of passive linear
optics and Kerr gates is photon-number conserving:
\begin{equation}
\label{eq:number-conserving}
[U,\hat{N}]=0.
\end{equation}
Consequently, the resulting unitary is block-diagonal with respect to the decomposition
\begin{equation}
\label{eq:block-decomp}
\mathcal{H}^{(m)}=\bigoplus_{N=0}^{\infty}\mathcal{H}^{(m)}_N,
\end{equation}
namely,
\begin{equation}
\label{eq:block-unitary}
U=\bigoplus_{N=0}^{\infty} U^{(N)},\qquad U^{(N)}\in U_\text{full}(\mathcal{H}^{(m)}_N).
\end{equation}
The dimension of the $N$-photon sector is
\begin{equation}
\label{eq:dim-HN}
\dim(\mathcal{H}^{(m)}_N)=\binom{N+m-1}{N}.
\end{equation}

Passive linear optics alone realizes only the subgroup corresponding to the $U(m)$ action on $\mathcal{H}^{(m)}_N$.
However, adding any gate outside this linear-optical group can promote universality on $\mathcal{H}^{(m)}_N$. In
particular, Theorem~1 of Ref.~\cite{PhysRevLett.119.220502} implies that for $m>2$ modes, passive bosonic linear optics
together with any additional gate $V\notin \mathrm{LO}_b$ generates the full unitary group on $\mathcal{H}^{(m)}_N$ denoted as $U_\text{full}(\mathcal{H}^{(m)}_N)$.
Since the Kerr gate in Eq.~\eqref{eq:kerr-gate} is not an element of passive linear optics, it follows that for $m>2$
the gate set consisting of passive linear optics and a single-mode Kerr nonlinearity is universal on each fixed-$N$
sector:
\begin{equation}
\label{eq:sector-universality}
\overline{\langle \mathrm{LO}_b,\;U_{\mathrm{Kerr}}\rangle}\;=\;U_\text{full}(\mathcal{H}^{(m)}_N),\qquad \text{for all }N\ge 0\text{ and }m>2,
\end{equation}
where the overline denotes closure (e.g., in operator norm on the finite-dimensional space $\mathcal{H}^{(m)}_N$) and
$\mathrm{LO}_b$ denotes passive bosonic linear optics restricted to $\mathcal{H}^{(m)}_N$.

\end{document}